\documentclass[a4paper,11pt]{amsart}
\usepackage{amssymb}
\usepackage{amscd}
\usepackage{comment}
\usepackage{amsmath,amsthm}
\usepackage[colorlinks=true]{hyperref}
\usepackage{enumerate}
\usepackage{booktabs,multirow}
\usepackage{tikz}
\usepackage{rotating}
\usetikzlibrary{patterns}
\usetikzlibrary{decorations.pathreplacing}
\usetikzlibrary{calc,through}

\allowdisplaybreaks[1]
\numberwithin{figure}{section}
\numberwithin{equation}{section}

\title[Shifted tableaux and products of Schur's symmetric functions]
{Shifted tableaux and products of Schur's symmetric functions}
\author[K.~Shigechi]{Keiichi~Shigechi}
\email{k1.shigechi AT gmail.com}
\date{\today}

\newcommand\tikzpic[2]{
\raisebox{#1\totalheight}{
\begin{tikzpicture}
#2
\end{tikzpicture}
}}
\newtheorem{theorem}[figure]{Theorem}
\newtheorem{example}[figure]{Example}
\newtheorem{lemma}[figure]{Lemma}
\newtheorem{defn}[figure]{Definition}
\newtheorem{prop}[figure]{Proposition}
\newtheorem{cor}[figure]{Corollary}

\newtheorem{remark}[figure]{Remark}
\begin{document}
\begin{abstract}
We introduce a new combinatorial object, semistandard increasing 
decomposition tableau and study its relation to a semistandard 
decomposition tableau introduced by Kra\'skiewicz and developed 
by Lam and Serrano.
We also introduce generalized Littlewood--Richardson coefficients for 
products of Schur's symmetric functions and give combinatorial 
descriptions in terms of tableau words.
The insertion algorithms play central roles for proofs.
A new description of shifted Littlewood--Richardson coefficients 
is given in terms of semistandard increasing decomposition 
tableaux.
We show that a ``big" Schur function is expressed as a sum of 
products of two Schur $P$-functions, and vice versa.
As an application, we derive two Giambelli formulae for big Schur 
functions: one is a determinant and the other is a Pfaffian.
\end{abstract}

\maketitle

\section{Introduction}
The ordinary representation of the symmetric group plays important 
roles in symmetric functions, representations, and the combinatorics
of Young tableaux (see {\it e.g.} \cite{Ful97,FulHar91,Sag01} and references therein).
The projective representations of the symmetric and alternating 
groups were studied by Schur in \cite{Schur11} and Schur $P$- and $Q$-functions 
were introduced.
The expansion of the $Q$-function in terms of power-sum symmetric 
functions gives the characters of the projective representation, 
as the Schur functions do for the ordinary representation.
After the invention of {\it shifted} tableaux by Thrall \cite{Thr52},
there has been progress in connecting shifted tableaux and the projective 
representation.
Shifted tableaux have combinatorial structures as the theory of 
ordinary tableaux: shifted analogue of the Robinson--Schensted--Knuth 
correspondence studied by Sagan \cite{Sag87} and Worley \cite{Wor84}, and 
a shifted analogue of the Littlewood--Richardson coefficient 
by Stembridge \cite{Ste89}. 
Sagan and Worley introduced two combinatorial tools for the study 
of $Q$-functions: one is a shifted insertion and the other 
is shifted {\it jeu-de-taquin}. 
The shifted insertion possesses a shifted analogue of Schensted correspond, 
namely a bijection between a permutation and a pair of shifted tableaux of 
the same shape. 
Haiman constructed a {\it mixed} insertion which is dual to the shifted 
insertion in \cite{Hai89}. Here, ``dual" means that application of the mixed insertion 
to the inverse of a permutation $w$ gives the same tableaux as the shifted 
insertion applied to $w$.
The pair of tableaux consists of the mixed insertion tableau and the 
mixed recording tableau.
In \cite{Ser10}, Serrano introduce the 
{\it shifted Knuth (or shifted plactic) relations} as a shifted analogue 
of the plactic relations introduced by Knuth \cite{Knu70} and further developed 
as the plactic monoid by Lascoux and Sch\"utzenberger \cite{LasSch81}. 
Two words have the same mixed insertion tableau if and only if 
they are shifted Knuth-equivalent. 
Another combinatorial object is a semistandard decomposition 
tableau (SSDT) \cite{Kra89,Lam96,Ser10}.
A SSDT is an output of Kra\'skiewicz insertion for a word.
The algorithm was introduced by Kra\'skiewicz for the hyperoctahedral
group~\cite{Kra89} and further developed by Lam to study the $B_{n}$ Stanley symmetric
functions~\cite{Lam96}.
A SSDT is also characterized by a shifted tableau, namely two words are shifted
plactic equivalent if and only if they have the same semistandard 
Kra\'skiewicz insertion tableau~\cite{Ser10}.

The expansion coefficients of a product of two Schur functions 
in terms of Schur functions are known as a Littlewood--Richardson
coefficient. 
There are many combinatorial models to describe this coefficient:  
Littlewood--Richardson (LR) tableaux~\cite{LitRic34}, lattice words based on 
plactic monoids (see {\it e.g.}~\cite{LasLecThi02,LasSch81}), and 
puzzles~\cite{KnuTao99,KnuTaoWood04}. 
Similarly, expansion coefficients of a product of two $P$-functions 
are known as Littlewood--Richardson--Stembridge (LRS) coefficients.
Its combinatorial description based on the study of the projective 
representation is given by Stembridge in \cite{Ste89}.
Another one in terms of SSDT's is given by Cho in \cite{Cho13}. 
A description of expansion coefficients of a $P$-function in terms of 
Schur functions is also given in \cite{Ste89}.
This coefficient also appears in an expansion of ``big" 
Schur function in terms of $Q$-functions.
Thus, the expansion coefficient of a product of Schur's symmetric 
functions in terms of Schur functions or $P$-functions can be 
calculated by successive applications of the above-mentioned 
expansion coefficients.

The purpose of this paper is three-fold.
First, we introduce shifted analogue of a Yamanouchi word (shifted 
Yamanouchi word for short) and a new concept called 
{\it semistandard increasing 
decomposition tableau} (SSIDT). 
We show the shifted and mixed insertion algorithms are characterized 
by Yamanouchi words and shifted Yamanouchi words, respectively.
We construct a SSIDT from a SSDT. By construction, we have a bijection 
between a SSDT and a SSIDT, which implies that a SSIDT is bijective 
to a shifted tableau.
A word obtained from an ordinary Young tableau is a concatenation of 
weakly increasing sequences. 
Since a SSDT is a concatenation of hook words and its shape is a strict 
partition $\lambda$, it can be viewed as a shifted analogue of an ordinary Young tableau.
On the other hand, a SSIDT is expressed as a concatenation of 
weakly increasing sequences. 
The shape of a SSIDT is not $\lambda$ but characterized by 
$\lambda$.
From these properties, a SSIDT is also viewed as another shifted 
analogue of an ordinary Young tableau.
We have two types of SSIDT for a given strict partition $\lambda$,
which have the shapes $\epsilon^+(\lambda)$ and $\epsilon^{-}(\lambda)$ 
(see Section \ref{sec:SSIDT} for definition).
By construction, a SSIDT of shape $\epsilon^{\pm}(\lambda)$ is a 
bijective to a SSDT of shape $\lambda$.
There is also a one-to-one correspondence between SSIDT's of shape $\epsilon^+(\lambda)$
and $\epsilon^{-}(\lambda)$.

Secondly, we survey generalized Littlewood--Richardson 
coefficients.
We have four types of symmetric functions: a Schur function 
$s_{\alpha}$, a ``big" Schur function $\hat{S}_{\beta}$, a 
$P$-function $P_{\lambda}$ and a $Q$-function $Q_{\lambda}$ where $\alpha$ and $\beta$ are 
ordinary partitions and $\lambda$ is a strict partition.
A Schur $Q$-function $Q_{\lambda}$ is related to the $P$-function 
by $Q_{\lambda}=2^{l(\lambda)}P_{\lambda}$ where $l(\lambda)$ is the length of $\lambda$.
Roughly speaking, $Q_{\lambda}$ contains the same information as $P_{\lambda}$.
We consider a product of these three functions and expand 
it in terms of Schur functions $s_{\alpha}$ or $P$-functions $P_{\lambda}$
(see Section \ref{sec:GLR} for definition).
We call these expansion coefficients generalized Littlewood--Richardson 
coefficients.
Generalized Littlewood--Richardson coefficients can be essentially calculated
by using LR coefficients and LRS coefficients successively.
However, this method is not efficient.
We propose simple combinatorial descriptions of generalized
Littlewood--Richardson coefficients in terms of tableau
words. 
The proofs are elementary but combinatorial. 
We make use of insertion algorithms introduced in Section \ref{sec:Schur}.
We also give alternative combinatorial descriptions 
for the Littlewood--Richardson--Stembridge coefficients, 
one of which is in terms of SSIDT (see Theorem \ref{thrm:PPPword}).

Finally, we study a relation between a ``big" Schur function 
$\hat{S}_{\alpha}$ and Schur $P$-functions.
A Schur function $\hat{S}_{\alpha}$ can be expressed as 
a sum of products of two Schur $P$-functions.
This expansion has a remarkable property: the expansion 
coefficient is either $1$ or $-1$ except an overall factor, 
which means that it is multiplicity free.
By inverting this relation, a product of two $P$-functions 
can be expanded in terms of $\hat{S}$-functions.
Again, the expansion coefficient is $1$ or $-1$ except an 
overall factor. 
Similarly, a skew ``big" Schur function $\hat{S}_{\alpha/\beta}$ can 
be also expanded in terms of products of two skew $P$-functions.
When two partitions $\alpha$ and $\beta$ are shift-symmetric, 
we recover the result by J\'ozefiak and Pragacz \cite{JozPra91}.
We also give an expression of $\hat{S}$-functions in terms of 
perfect matchings.
Recall that Schur function $s_{\alpha}$ and $P$-function $P_{\lambda}$ 
have Giambelli formulae: $s_{\alpha}$ is determinant and 
$P_{\lambda}$ is a Pfaffian.
By deducing from the expression in terms of perfect matchings, 
we propose two Giambelli formulae for a skew Schur 
$\hat{S}$-function $\hat{S}_{\alpha/\beta}$:
one is determinant and the other is Pfaffian.

The plan of this paper is as follows.
In Section \ref{sec:Schur}, we introduce basic facts about 
Schur symmetric functions and insertion algorithms.
In Section \ref{sec:SSIDT}, we give the definition of 
a semistandard increasing tableaux, and show that 
the map from a SSDT to a SSIDT is well-defined.
Section \ref{sec:GLR} is devoted to the analysis of 
generalized Littlewood--Richardson coefficients.
In Section \ref{sec:SinPP}, we show that a ``big" Schur 
function is expanded in terms of products of two $P$-functions and 
vice versa.
We give a description of big Schur function in terms of 
perfect matchings.  
Based on these expressions, we give two Giambelli formulae for 
a ``big" Schur function $\hat{S}_{\alpha/\beta}$.

\section{Schur's symmetric functions and insertion algorithms}
\label{sec:Schur}
\subsection{Schur's symmetric functions}
A {\it partition} of $n$ is a finite non-increasing sequence of 
positive integers $\lambda=(\lambda_1,\lambda_2,\ldots,\lambda_{l})$
such that $\sum_{i=1}^{l}\lambda_i=n$. 
The $\lambda_i$ are called the {\it parts} of the partition.
The integer $n$ is called the size of $\lambda$ and denoted $|\lambda|$.
The {\it length} $l(\lambda)$ of a partition $\lambda$ is defined as 
the number of parts, {\it i.e.,} $l(\lambda):=l$.
The {\it Young diagram} of $\lambda$ is an array of boxes with $\lambda_i$
boxes in the $i$-th row.
The {\it skew Young diagram} $\lambda/\mu$ is obtained by removing a Young 
diagram $\mu$ from $\lambda$ containing $\mu$.
We denote by $\alpha^{T}$ the conjugate of a partition $\alpha$, {\it i.e.},
$\alpha^{T}_{i}=\#\{j|\alpha_{j}\ge i\}$.

A {\it strict partition} $\lambda=(\lambda_1,\ldots,\lambda_l)$ is a partition 
with all parts distinct, {\it i.e.,} $\lambda_1>\lambda_2>\ldots>\lambda_l$.
For a strict partition $\lambda$, the {\it shifted diagram} or {\it shifted 
shape} of $\lambda$ is an array of boxes where the $i$-th row has $\lambda_i$
boxes and is shifted $i-1$ steps to the right with respect to the first row.
The {\it skew shifted diagram} $\lambda/\mu$ is obtained by removing 
a shifted diagram $\mu$ from $\lambda$ containing $\mu$.
The main diagonal of a skew shifted diagram is a set of boxes indexed $(i,i)$.

Let $\lambda$ and $\mu$ be a strict partition satisfying $l(\lambda)=l(\mu)$ or
$l(\lambda)=l(\mu)+1$.
We construct an ordinary Young diagram $\alpha$ from $\lambda$ and $\mu$ 
by transposing $\lambda$ and glueing together $\lambda$ and $\mu$ along the 
main diagonal (see the right picture in Fig. \ref{fig:ssymtab}). 
We denote the diagram by $\alpha=\lambda\otimes\mu$.
When $\lambda=\mu$, {\it i.e.,} $\alpha=\lambda\otimes\lambda$, we call 
$\alpha$ a {\it shift-symmetric tableau}.

\begin{figure}[ht]
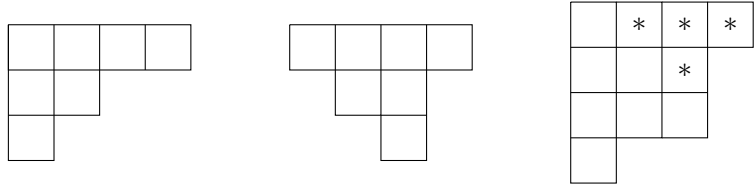

\tikzpic{-0.5}{
\draw(0,0)--(2.4,0)--(2.4,-0.6)--(0,-0.6)--(0,0);
\draw(0.6,0)--(0.6,-1.8)--(0,-1.8)--(0,0);
\draw(1.2,0)--(1.2,-1.2)--(0,-1.2);
\draw(1.8,0)--(1.8,-0.6);
}
\qquad
\tikzpic{-0.5}{
\draw(0,0)--(2.4,0)--(2.4,-0.6)--(0,-0.6)--(0,0);
\draw(1.2,0)--(1.2,-1.2)--(0.6,-1.2)--(0.6,0);
\draw(1.2,-1.2)--(1.2,-1.8)--(1.8,-1.8)--(1.8,0);
\draw(1.2,-1.2)--(1.8,-1.2);
}
\qquad
\tikzpic{-0.5}{
\draw(0,0)--(0,-2.4)--(0.6,-2.4)--(0.6,0)--(0,0);
\draw(0.6,0)--(1.2,0)--(1.2,-1.8)--(0,-1.8);
\draw(1.2,0)--(1.8,0)--(1.8,-1.8)--(1.2,-1.8);
\draw(1.8,0)--(2.4,0)--(2.4,-0.6)--(0,-0.6);
\draw(0,-1.2)--(1.8,-1.2);
\draw(0.9,-0.3)node{$\ast$};
\draw(1.5,-0.3)node{$\ast$};
\draw(2.1,-0.3)node{$\ast$};
\draw(1.5,-0.9)node{$\ast$};
}
\caption{An ordinary tableau of shape $(4,2,1)$ (left picture), a shifted 
tableau of shape $(4,2,1)$ (middle picture) and an ordinary tableau 
of shape $(4,3,3,1)=(4,2,1)\otimes(3,1)$ (right picture).
In the right picture, a tableau with $\ast$ is of shape $(3,1)$.
}
\label{fig:ssymtab}
\end{figure}

A {\it semistandard Young tableau} $T$ of shape $\lambda$ is a filling of the
shape $\lambda$ with letters from the alphabet $X=\{1<2<3<\ldots\}$ such that 
\begin{itemize}
\item each row is weakly increasing from left to right
\item each column is increasing from top to bottom.
\end{itemize}
A {\it skew semistandard Young tableau} is defined analogously.
The {\it content} of $T$ is an integer sequence $v:=(v_1,v_2\ldots)$, 
where $v_i$ is the number of the letter $i$ in $T$.
A semistandard tableaux is called {\it standard} when its content 
is $(1,1,\ldots,1)$.
We denote the set of semistandard tableaux of shape $\lambda$ by 
$\mathrm{SSYT}(\lambda)$.

A {\it semistandard shifted Young tableau} $T$ of (ordinary or shifted) 
shape $\lambda$ is a filling of the shifted shape $\lambda$ with letters 
from the marked alphabet $X'=\{1'<1<2'<2<\ldots\}$ such that
\begin{itemize}
\item rows and columns of $T$ are weakly increasing
\item each letter $i$ appears at most once in every column
\item each letter $i'$ appears at most once in every row
\item no primed letter on the main diagonal.
\end{itemize}
A {\it skew shifted Young tableau} is defined analogously.
The {\it content} of $T$ is an integer sequence $v:=(v_1,v_2\ldots)$, 
where $v_i$ is the number of the letter $i$ and $i'$ in $T$.
A semistandard shifted tableau called {\it marked standard} when 
its content is $(1,1,\ldots,1)$ and the letter is possibly primed once.
A marked standard tableau is called {\it standard} when it is a marked
standard tableau without primed letters.
We denote the set of semistandard shifted tableau of shape $\lambda$ by 
$\mathrm{SSShYT}(\lambda)$.

For an ordinary or shifted tableau $T$ with content $v=(v_1,v_2,\ldots)$, 
we denote the monomial by 
\begin{eqnarray*}
x^{T}:=x_1^{v_1}x_2^{v_2}\cdots.
\end{eqnarray*}

For a partition $\alpha$, we define {\it Schur functions} as
\begin{eqnarray*}
s_{\alpha}&:=&s_{\alpha}(X)=\sum_{T\in \mathrm{SSYT}(\alpha)}x^{T}, \\
\hat{S}_{\alpha}&:=&\hat{S}_{\alpha}(X)=\sum_{T\in \mathrm{SSShYT}(\alpha)}x^{T}.
\end{eqnarray*}
The {\it skew Schur functions} $s_{\alpha/\beta}$ and $\hat{S}_{\alpha/\beta}$
are defined similarly for the skew shape $\alpha/\beta$.

For a strict partition $\lambda$, the Schur $P$- and $Q$-functions are defined as
\begin{eqnarray*}
P_{\lambda}&:=&P_{\lambda}(X)=\sum_{T\in\mathrm{SSShT}(\lambda)}x^{T}, \\
Q_{\lambda}&:=&Q_{\lambda}(X)=2^{l(\lambda)}P_{\lambda}.
\end{eqnarray*}
The {\it skew Schur $P$- and $Q$-functions} $P_{\lambda/\mu}$ and 
$Q_{\lambda/\mu}:=2^{l(\lambda)-l(\mu)}P_{\lambda/\mu}$ are 
defined similarly.

\subsection{Littlewood--Richardson coefficients}
We define the Littlewood--Richardson coefficient $a_{\alpha\beta}^{\gamma}$
as 
\begin{eqnarray*}
s_{\alpha}s_{\beta}=\sum_{\gamma}a_{\alpha\beta}^{\gamma}s_{\gamma},
\end{eqnarray*}
that is, $a_{\alpha\beta}^{\gamma}$ is the multiplicity of $s_{\gamma}$ in the 
product of two Schur functions $s_{\alpha}$ and $s_{\beta}$.

The LR coefficient $a_{\alpha\beta}^{\gamma}$ is expressed in terms of 
Yamanouchi words \cite{LitRic34}.
We have 
\begin{eqnarray*}
a_{\alpha\beta}^{\gamma}=\#\{T\in T(\alpha/\beta;\gamma)|\ \mathrm{read}(T)\text { is a Yamanouchi word}\}.
\end{eqnarray*}

The coefficients $a_{\alpha\beta}^{\gamma}$ also appear in the case of 
$\hat{S}$-functions (see e.g. \cite{Mac95}): 
\begin{eqnarray}
\label{eqn:LRSS}
\hat{S}_{\alpha}\hat{S}_{\beta}=\sum_{\gamma}a_{\alpha\beta}^{\gamma}\hat{S}_{\gamma}.
\end{eqnarray}
Let $\lambda,\mu$ and $\nu$ be strict partitions satisfying $\lambda,\mu\subseteq\nu$.
Similarly, a Littlewood--Richardson--Stembridge coefficient 
$d_{\lambda\mu}^{\nu}$ is defined in terms of Schur $P$-function as 
\begin{eqnarray*}
P_{\lambda}P_{\mu}=\sum_{\nu}d_{\lambda\mu}^{\nu}P_{\nu}.
\end{eqnarray*}
\begin{theorem}[Stembridge \cite{Ste89}]
\label{thrm:StemPP}
We have 
\begin{eqnarray*}
d^{\lambda}_{\mu\nu}=
\#\{T\in T'(\lambda/\mu;\nu)|\ \mathrm{read}(T)\text { is an LRS word} \}.
\end{eqnarray*}
\end{theorem}

\subsection{Skew functions}
The skew Schur functions are expressed in terms of Littlewood--Richardson coefficients 
as 
\begin{eqnarray*}
s_{\alpha/\beta}&=&\sum_{\gamma}a_{\beta\gamma}^{\alpha}s_{\gamma}, \\
Q_{\lambda/\mu}&=&\sum_{\nu}d_{\mu\nu}^{\lambda}Q_{\nu}, \\
\hat{S}_{\alpha/\beta}&=&\sum_{\gamma}a_{\beta\gamma}^{\alpha}\hat{S}_{\gamma}.
\end{eqnarray*}
The skew $P$-function is given by $P_{\lambda/\mu}=2^{l(\mu)-l(\lambda)}Q_{\lambda/\mu}$.

\subsection{Basic properties of Schur functions}
We summarize properties of Schur $s$-functions and $P$-functions.
The reader is referred to \cite{Mac95} for detailed definitions and proofs.

Let $\Lambda:=\bigoplus_{n\ge0}\Lambda^{n}$ be the graded ring of symmetric 
functions in the variables $x_1,x_2,\ldots$ with coefficients in $\mathbb{Z}$.
The {\it power-sum symmetric function} $p_{r}$, $r>0$, is defined by 
$p_{r}:=x_1^{r}+x_{2}^{r}+\cdots$, and 
we denote $p_{\lambda}:=p_{\lambda_{1}}p_{\lambda_{2}}\cdots$ for a 
partition $\lambda$.
The set $\{p_{\lambda}|\ |\lambda|=n\}$ forms a basis of $\Lambda^{n}$.
Another basis of the ring $\Lambda$ is Schur functions $s_{\Lambda}$.
We have $\{s_{\lambda}|\ |\lambda|=n\}$ is a $\mathbb{Z}$-bases of $\Lambda^{n}$. 
We define the bilinear form $\langle,\rangle$ on $\Lambda$ via the generating 
function 
\begin{eqnarray*}
\prod_{i,j}\frac{1}{1-x_{i}y_{j}}=\sum_{\lambda}s_{\lambda}(x)s_{\lambda}(y).
\end{eqnarray*}
Thus Schur functions are orthonormal bases of $\Lambda^{n}$, namely 
\begin{eqnarray*}
\langle s_{\lambda},s_{\mu}\rangle=\delta_{\lambda,\mu}.
\end{eqnarray*}

Let $\Omega_{\mathbb{Q}}=\bigoplus_{n\ge0}\Omega^{n}_{\mathbb{Q}}$ be 
the graded subalgebra of $\Lambda_{\mathbb{Q}}$ 
generated by $\{p_{2r-1}|\ r\ge1\}$ and let $\Omega=\Omega_{\mathbb{Q}}\cap\Lambda$
denote the $\mathbb{Z}$-coefficient graded subring of $\Omega_{\mathbb{Q}}$.
Let $OP_{n}$ denote the partitions of $n$ with odd parts.
Then, the set $\{p_{\lambda}|\ \lambda\in OP_{n}\}$ forms a basis of 
$\Omega_{Q}^{n}$.
We define an inner product $[,]$ on $\Omega_{\mathbb{C}}^{n}$ by
the generating function
\begin{eqnarray*}
\prod_{i,j}\frac{1+x_{i}y_{j}}{1-x_{i}y_{j}}
&=& \sum_{\lambda\in OP_{n}}\frac{1}{z_{\lambda}}2^{l(\lambda)}
p_{\lambda}(x)p_{\lambda}(y) \\
&=&\sum_{\lambda\in DP}Q_{\lambda}(x)P_{\lambda}(y),
\end{eqnarray*}
where $DP$ denote the set of distinct partitions.
where $z_{\lambda}=\prod_{i\ge1}i^{m_{i}}m_{i}!$ for 
$\lambda=(1^{m_{1}},2^{m_{2}},\cdots)$.
Thus, we have 
\begin{eqnarray*}
&&[p_{\lambda},p_{\mu}]=z_{\lambda}2^{-l(\lambda)}\delta_{\lambda\mu} 
\quad\text{ for }\lambda,\mu\in OP_{n}, \\
&&[P_{\lambda},Q_{\mu}]=\delta_{\lambda\mu} 
\quad\text{ for } \lambda,\mu\in DP.
\end{eqnarray*}
The bilinear form $\langle,\rangle$ and $[,]$ are related 
by Proposition 9.1 in \cite{Ste89}:
\begin{eqnarray}
\label{ex:sinPSinP}
\langle s_{\alpha},P_{\lambda}\rangle=[S_{\alpha},P_{\lambda}],
\end{eqnarray}
where $\alpha$ is an ordinary partition.

Define symmetric functions $q_{n}\in\Lambda^{n}$ by the generating 
function
\begin{eqnarray*}
\sum_{n\ge0}q_{n}t^{n}=\prod_{i}\frac{1+x_{i}t}{1-x_{i}t}.
\end{eqnarray*}
A $Q$-function $Q_{\lambda}$ with $l(\lambda)=2$ can be expressed 
in terms of $q_{n}$, namely we have  
\begin{eqnarray*}
Q_{(r,s)}=q_{r}q_{s}+2\sum_{i=1}^{s}(-1)^{i}q_{r+i}q_{s-i}.
\end{eqnarray*} 
The skew Schur $\hat{S}$-function has a determinant expression:
\begin{eqnarray}
\label{Sdet1}
\hat{S}_{\alpha/\beta}=\det\left[q_{\alpha_{i}-\beta_{j}-i+j}\right]_{1\le i,j\le l},
\end{eqnarray}
where $q_{-r}=0$ for $r>0$ and $l:=l(\alpha)$.

Let $\mathcal{S}_{2n}$ be the symmetric group of order $2n$ and $\mathcal{F}_{n}$
be the subset 
\begin{eqnarray*}
\mathcal{F}_{n}:=\left\{\rho\in\mathcal{S}_{2n}\Big\vert
\begin{array}{c}
\rho(1)<\rho(3)<\cdots<\rho(2n-1), \\
\rho(2i-1)<\rho(2i) \ (1\le i\le n)
\end{array}
\right\}.
\end{eqnarray*}
For a skew symmetric matrix $A:=(a_{ij})$ of order $2n$, 
we define the Pfaffian of $A$ as 
\begin{eqnarray}
\label{eqn:defPf}
\mathrm{pf}[A]
:=
\sum_{\rho\in\mathcal{F}_{n}}\epsilon(\rho)
\prod_{i=1}^{n}a_{\rho(2i-1)\rho(2i)},
\end{eqnarray}
where $\epsilon(\rho)$ is the sign of the permutation $\rho$.

Let $\tilde{l}(\lambda):=l(\lambda)$ for $l(\lambda)\equiv0\pmod2$ and 
$\tilde{l}(\lambda):=l(\lambda)+1$ for $l(\lambda)\equiv1\pmod2$.
We denote $P[\lambda/\mu]:=P_{\lambda/\mu}$.
We define a skew-symmetric matrix $P_{i,j}(\lambda,\mu)$ as 
\begin{eqnarray*}
P_{i,j}(\lambda,\mu):=
\begin{cases}
P[(\lambda_{i},\lambda_{j})], & 1\le i< j\le \tilde{l}(\lambda), \\
P[(\lambda_{i}-\mu_{\tilde{l}(\mu)+\tilde{l}(\lambda)-j+1})], 
& 1\le i\le\tilde{l}(\lambda), 
\tilde{l}(\lambda)+1\le j\le \tilde{l}(\lambda)+\tilde{l}(\mu), \\
0, & otherwise,
\end{cases}
\end{eqnarray*}
for $1\le i<j\le \tilde{l}(\lambda)+\tilde{l}(\mu)$.
A skew $P$-function is expressed as (see {\it e.g.} \cite{JozPra91})
\begin{eqnarray}
\label{eqn:skewPpf}
P_{\lambda/\mu}=\mathrm{pf}[P_{i,j}(\lambda,\mu)].
\end{eqnarray}
Especially, a (non-skew) $P$-function has a simple Pfaffian expression:
\begin{eqnarray}
\label{eqn:Ppf}
P_{\lambda}=\mathrm{pf}\left[P_{(\lambda_{i},\lambda_{j})}\right]_{1\le i,j\le l(\lambda)},
\end{eqnarray} 
where $\lambda_{l+1}=0$ for $l$ odd.

\subsection{Insertion algorithms}
We introduce four insertion algorithms used in this paper. 
We start with the the Robinson--Schensted--Knuth (RSK) correspondence.

\begin{defn}[RSK insertion \cite{Schensted61}]
For a word $w=w_{1}\cdots w_{n}$ in the alphabet $X$, 
we recursively define a sequence of tableaux 
$(T_{0},U_{o}):=(\emptyset,\emptyset), (T_{1},U_{1}),\ldots, (T_{n},U_{n})=(T,U)$.
For $1\le i\le n$, we insert $w_{i}$ into $T_{i-1}$ as follows.

We start with $z:=w_{i}$, $S:=T_{i-1}$ and $p=1$.
\begin{enumerate}
\item[($\clubsuit$)] Insert $z$ into the $p$-th row of $S$, bumping out the smallest letter 
$a\in X$ which is strictly greater than $z$.
Let new $S$ be the tableau where $a$ is replaced by $z$ in $S$.
We define new $z:=a$ and new $p\mapsto p+1$, and go to ($\clubsuit$).
\end{enumerate}
The insertion algorithm stops when a letter is placed at the end of a row.

A tableau $U_{i}$ is obtained from $U_{i-1}$ by adding a box with content $i$ 
on the same location as a new box added to $T_{i-1}$ to obtain $T_{i}$. 
\end{defn}
We call $T$ the insertion tableau and $U$ the recording tableau.
When a word $w$ has a pair $(T,U)$, we denote $w=\mathrm{RSK}^{-1}(T,U)$.

A {\it shifted mixed insertion} was introduced by Haiman for the correspondence
between permutations and pairs of shifted Young tableaux.
This mixed insertion can be viewed as a shifted analogue of the 
Robinson--Schensted--Knuth correspondence.
A semistandard generalization of the mixed insertion of Haiman was given 
by Serrano. 
This extended insertion gives a correspondence between words in the alphabet $X$
and pairs of semistandard shifted and standard shifted Young tableaux.

\begin{defn}[Mixed insertion \cite{Hai89}]
For a word $w=w_1\cdots w_n$ in the alphabet $X$, we recursively define 
a sequence of tableaux, 
$(T_0,U_0):=(\emptyset,\emptyset),(T_1,U_1),\ldots,(T_n,U_n)=(T,U)$.
For $1\le i\le n$, we insert $w_i$ into $T_{i-1}$ as follows.

We start with $z:=w_i$, $S:=T_{i-1}$ and $(p,q):=(1,1)$.
\begin{itemize}
\item[($\bigstar$)] Insert $z$ into the $p$-th row of $S$, bumping out 
the smallest letter $a\in X'$ which is strictly greater than $z$.
Let new $S$ be the tableau where $a$ is replaced by $z$ in $S$ and 
$(p,q)$ be the position of $a$ in old $S$.
\end{itemize}
\begin{enumerate}
\item If $a$ is not on the main diagonal, then
\begin{enumerate}
\item if $a$ is unprimed, then we insert $a$ into the $p+1$-th row.
We bump out the smallest element $b$ which is strictly greater than $a$. 
We set $a:=b$, $(p,q)$ be the position of $b$ in $S$ and go to ($\bigstar$) .
\item if $a$ is primed, then we insert $a$ into the $q+1$-th column to the right.
We bump out the smallest element $b$ which is strictly greater than $a$.
We set $a:=b$, $(p,q)$ be the position of $b$ in $S$ and go to $(1)$.
\end{enumerate}
\item If $a$ is on the main diagonal ($a$ is unprimed), then we prime it and 
insert it into the $q+1$-th column to the right. 
We bump out the smallest element $b$ which is strictly greater than $a$.
We set $a:=b$, $(p,q)$ be the position of $b$ in $S$ and 
go to $(1)$. 
\end{enumerate}
The insertion algorithm stops when a letter is placed at the end of a row 
or a column.

A tableau $U_{i}$ is obtained by adding a box with content $i$ on the same 
location as a new box added to $T_{i-1}$ to obtain $T_{i}$.
\end{defn}

We call $T$ the {\it mixed insertion tableau} and $U$ the {\it mixed recording 
tableau} and denote them $P_{\mathrm{mix}}(w)$ and $Q_{\mathrm{mix}}(w)$, 
respectively.
When a word $w$ has a pair $(T,U)$ by the mixed insertion, 
we denote $w=\mathrm{mixRSK}^{-1}(T,U)$.

In case of ordinary tableaux, Knuth has proved that two words $u$ and $v$ 
in $X$ have the same RSK insertion tableau if and only if $u$ and $v$ 
are equivalent modulo the plactic relations \cite{Knu70}.
The following theorem is a shifted analogue of the plactic relations for 
the mixed insertion.
\begin{theorem}[shifted plactic relations \cite{Ser10}]
\label{thrm:splactic}
Two words have the same mixed insertion tableau if and only if 
they are equivalent modulo the following {\it shifted plactic relations}:
\begin{eqnarray}
\label{eqn:pl1}
&abdc\sim adbc \quad\text{ for } a\le b\le c<d, \\
\label{eqn:pl2}
&acdb\sim acbd \quad\text{ for } a\le b<c\le d, \\
\label{eqn:pl3}
&dacb\sim adcb \quad\text{ for } a\le b<c<d, \\
\label{eqn:pl4}
&badc\sim bdac \quad\text{ for } a<b\le c<d, \\
\label{eqn:pl5}
&cbda\sim cdba \quad\text{ for } a<b<c\le d,\\
\label{eqn:pl6}
&dbca\sim bdca \quad\text{ for } a<b\le c\le d, \\
\label{eqn:pl7}
&bcda\sim bcad \quad\text{ for } a<b\le c\le d, \\
\label{eqn:pl8}
&cadb\sim cdab \quad\text{ for } a\le b<c\le d.
\end{eqnarray}
\end{theorem}

There is another insertion process for the correspondence between permutations
and pairs of standard and marked standard shifted tableaux.
We call this insertion process {\it shifted insertion}.
We define the shifted insertion following \cite{Sag87}.

\begin{defn}[Shifted insertion \cite{Sag87}]
For a word $w=w_1\cdots w_n$ in the alphabets $X$, we recursively 
define a sequence of shifted tableaux 
$(T_0,U_0):=(\emptyset,\emptyset), (T_1,U_1),\ldots,(T_n,U_n)=(T,U)$.
For $1\le i\le n$, we insert a letter $w_i$ into $T_{i-1}$ as follows:

We start with $z:=w_{i}$, $S:=T_{i-1}$, $(p,q):=(1,1)$ and 
$C:=\mathrm{Schensted}$.
\begin{enumerate}
\item  If $C=\mathrm{Schensted}$, then 
we insert $z$ into the $p$-th row of $S$, bumping out the smallest letter 
$a\in X$ which is strictly greater than $z$. 
Let (new) $S$ be the tableau where $a$ is replaced by $z$ in $S$.
If $a$ is on the main diagonal, we set $z=a$, $q\mapsto q+1$ and 
$C=\mathrm{non\text{-}Schensted}$ and go to $(2)$. Otherwise, we set $z=a$,
$p\mapsto p+1$ and go to $(1)$.

\item If $C=\mathrm{non\text{-}Schensted}$, then 
we insert $z$ into the $q$-th column of $S$, bumping out the smallest letter 
$a\in X$ which is greater than or equal to $z$. 
Let (new) $S$ be the tableau where $a$ is replaced by $z$ in $S$.
We set $z=a$, $q\mapsto q+1$ and go to $(2)$.
\end{enumerate}
The insertion process stops when a letter is placed at the end of a row or 
a column. If $C=\mathrm{Schensted}$ (resp. $C=\mathrm{non\text{-}Schensted}$) 
when the insertion process stops, we call the process Schensted 
(resp. non-Schensted) move.  

A tableau $U_{i}$ is obtained by adding a box with content $i$ on the same 
location as a new box added to $T_{i-1}$ to obtain $T_{i}$.
If the insertion is non-Schensted move, we put a prime on $i$.
\end{defn}

We call $T$ the {\it shifted insertion tableau} and $U$ the 
{\it shifted recording tableau} and denote them $P_{\mathrm{shift}}(w)$ and 
$Q_{\mathrm{shift}}(w)$.

For a permutation $\pi\in\mathcal{S}_{l}$, the mixed insertion and 
the shifted insertion are related by Theorem 6.10 in \cite{Hai89}:
\begin{eqnarray*}
(P_{\mathrm{mix}}(w),Q_{\mathrm{mix}}(w))
=
(Q_{\mathrm{shift}}(w^{-1}),P_{\mathrm{shift}}(w^{-1})).
\end{eqnarray*}

We introduce the notion of {\it semistandard decomposition tableaux}
in the following.

A word $w=w_{1}\ldots w_{l}$ on $X$ is called a {\it hook word}
if there exists $1\le m\le l$ such that 
\begin{eqnarray}
\label{eqn:hookword}
w_{1}>w_{2}\cdots>w_{m}\le w_{m+1}\le\cdots\le w_{l}.
\end{eqnarray}
We denote by $(w\downarrow)$ the subword $w_{1}w_{2}\cdots w_{m}$ of $w$ 
and by $(w\uparrow)$ the subword $w_{m+1}\cdots w_{l}$.

\begin{defn}[Semistandard Kra\'skiewicz (SK) insertion \cite{Kra89,Lam96,Ser10}]
\label{defn:SKins}
For a given hook word $w=(w\downarrow)\ast(w\uparrow)$ with 
$(w\downarrow)=w_{1}\cdots w_{m}$ and $(w\uparrow)=w_{m+1}\cdots w_{l}$
and a letter $x\in X$, the insertion of $x$ into $w$ is the word 
$ux$ if $wx$ is a hook word, or the $w'$ with an element $u$ which is 
bumped out in the following way:
\begin{enumerate}
\item let $y_j$ be the leftmost element in $(w\uparrow)$ which is strictly greater than x
\item replace $y_j$ by $x$
\item let $y_i$ be the leftmost element in $(w\downarrow)$ which is less than or equal to $y_j$
\item replace $y_i$ by $y_j$ and bump out $u:=y_i$. A word $w'$ is a word obtained from $w$ 
by replacing $y_i$ and $y_j$ in $w$ by $y_j$ and $x$.  
\end{enumerate}
\end{defn}

The insertion of $x$ into a SSDT $T$ with rows $w_1,\ldots,w_l$ is defined as follows.
First, we insert $x$ into $w_1$ and bumps out $x_1$. Then, we insert $x_1$ into the 
second row $w_2$. We continue this process until an element $x_i$ is placed at the end
of the row $w_i$.

The {\it SK insertion tableau} of a word $w=w_{1}\cdots w_{n}$ is obtained from the 
empty tableau by inserting the letters $w_{1},\ldots,w_n$. At an each steps, we obtain 
a SSDT. We denote by $P_{\mathrm{SK}}(w)$ the SK insertion tableau for a word $w$.
The {\it SK recording tableau} is the standard shifted Young tableau obtained by adding 
a box with content $i$ on the same location as a new box added to 
$P_{\mathrm{SK}}(w_1\cdots w_{i-1})$ to obtain $P_{\mathrm{SK}}(w_1\cdots w_{i})$

\subsection{Tableau words}
\label{sec:tw}
Let $T$ be an ordinary or shifted tableau.
The {\it reading word} $\mathrm{read}(T)$ is obtained by 
reading the contents of $T$ from the bottom row to the top row 
and in an each row from left to right.

For any letter $i$ in $X'$, we set $(i')'=i$ and $(i)'=i'$.
Let $w=w_1\cdots w_n$ be a reading word of $T$ in the alphabet $X'$.
We write $w'={w_n}'\cdots {w_1}'$.
The {\it weak reading word} $\mathrm{wread}(T)$ is obtained by 
reading unprimed letters of the word $w'w$.

For a word $w:=w_{1}\cdots w_{n}$, we define a reversed word 
$\mathrm{rev}(w):=w_{n}\cdots w_{1}$.

Given two words $v=v_{1}\cdots v_{m}$ and $w=w_{1}\cdots w_{n}$ , 
we define the concatenation of $v$ and $w$ as 
$v\ast w:=v_1\cdots v_{m}w_{1}\cdots w_{n}$.

Let $w=w_1\cdots w_n$ be a word in the alphabet $X$. 
A word with {\it lattice property} is a word satisfying 
$\#\{w_q: w_{q}=i,1\le q\le p\}\ge\#\{w_q: w_{q}=i+1,1\le q\le p\}$
for all $i$ and $1\le p\le n$. In other words, the occurrence of $i$ in 
$w_1\cdots w_p$ is greater than or equal to the occurrence of $i+1$ in
$w_1\ldots w_p$.
A {\it Yamanouchi} word $w$ is a word such that its reversal word satisfies
the lattice property.

We say that a word $w$ is a {\it weak Yamanouchi} word if 
weak reading word is a Yamanouchi word.

We follow \cite{Ste89} to define a Littlewood--Richardson--Stembridge 
(LRS) word. 
Let $w:=w_{1}w_{2}\cdots w_{n}$ be a word over the alphabet $X'$.
We define $m_{i}(j)$ $(0\le j\le 2n, i\ge1)$ depending on $w$ 
as 
\begin{eqnarray*}
m_{i}(j)
&:=&\text{number of } i \text{ in } w_{n-j+1},\ldots, w_{n} 
\qquad\text{ for } 0\le j\le n \\
m_{i}(n+j)&=&m_{i}(n)+\text{number of } i' \text{ in } w_{1},\ldots,w_{j}
\qquad\text{ for } 0<j\le n.
\end{eqnarray*}
The word $w$ is said to satisfy the {\it lattice property} if 
$m_{i}(j)=m_{i-1}(j)$, we have 
\begin{eqnarray*}
w_{n-j}&\neq& i,i'\qquad \text{if } 0\le j<n, \\
w_{j-n+1}&\neq&i-1,i' \qquad \text{if }n\le j<2n.
\end{eqnarray*}
A word $w$ is said to be an LRS word if it satisfies (i) $w$ 
satisfies the lattice property (ii) the first occurrence of 
a letter $i$ or $i'$ in $w$ is $i$. 

The condition (i) for an LRS word can be rephrased as follows.
For any letter $i$ in $X'$, we set $\widehat{i}=(i+1)'$ and 
$\widehat{i'}=i$.
Let $w=w_1\cdots w_n$ be a word in the alphabet $X'$.
We write $\widehat{w}=\widehat{w_n}\cdots \widehat{w_1}$.
The concatenated word of $\widehat{w}w$ is a Yamanouchi 
word, that is, every letters $i$ or $i'$ is preceded by more occurrence of 
$i-1$ than that of $i$ in $\widehat{w}w$ from right to left.

Let $w$ be a word in the alphabet $X$ and 
$\max(w)$ be the maximum letter in $w$.
We recursively define a strictly increasing sequence of length $r$, 
$\mathrm{seq}(r):=(\mathrm{seq}(r,i)|\ 1\le i\le r)$, starting from 
$r=\max(w)$ to $r=1$ as follows. 
A number $\mathrm{seq}(r,1)$ is the position (from left end) of the 
leftmost $r$ in $w$.
A number $\mathrm{seq}(r,i+1)$, $1\le i\le r-1$, is the position 
(from left end) of the leftmost $r-i$ which is right to the position 
$\mathrm{seq}(r,i)$. 
We delete the letters appearing in $\mathrm{seq}(r)$ from $w$ and 
construct $\mathrm{seq}(r-1)$ from the remaining letters in $w$ 
by the same procedure as above.
We say a word $w$ is {\it shifted Yamanouchi} word if 
(i) $w$ is a Yamanouchi word, (ii) there exits at least one 
$i-1$ to the left of the leftmost $i$, and 
(iii) A sequence of positive integers $\mathrm{seq}(r,r)$ 
for $1\le r\le \max(w)$ is strictly increasing, {\it i.e.,} 
$\mathrm{seq}(r,r)<\mathrm{seq}(r+1,r+1)$.

The shifted insertion and the mixed insertion are characterized by 
Yamanouchi words and shifted Yamanouchi words.
\begin{prop}
Let $w$ be a word of content $\lambda$.
Then, we have 
\begin{eqnarray}
\label{eqn:shapeYam1}
\#\{w|\ \mathrm{shape}(P_{\mathrm{shift}}(w))=\lambda\}
&=&\#\{w|\ w \text{ is a Yamanouchi word}\},\\
\label{eqn:shapeYam2}
\#\{w|\  \mathrm{shape}(P_{\mathrm{mix}}(w))=\lambda\}
&=&\#\{w|\  w \text{ is a shifted Yamanouchi word}\}.
\end{eqnarray}
\end{prop}

\begin{proof}
We first show the right term of Eqn.(\ref{eqn:shapeYam1}) implies
the left term.
Suppose that the word $w$ is written as a concatenation
of subwords $w=w'\ast y\ast w''$.
Since $w$ is a Yamanouchi word, the number of $x$ satisfying 
$x<y$ is strictly greater than that of $y$ in $w''$.
The insertion of $y$ into $P_{\mathrm{shift}}(w')$ results in 
a tableau such that $y$ is in the first row. 
The Yamanouchi property ensures that $y$ is bumped out by 
a letter $y-1$ and $y$ is inserted into the second row, 
a letter $y-2$ bumps out these $y-1$ and $y$ and insert them 
into the next rows, and finally a letter $1$ bumps out letters 
from $2$ to $y$ which are inserted into the next rows.
As a result, the letter $i$, $1\le i\le l(\lambda)$, is in
the $i$-th row of a tableau $P_{\mathrm{shift}}(w)$.
This is nothing but $\mathrm{shape}(P_{\mathrm{shift}}(w))=\lambda$. 
Further, if $w$ is not a Yamanouchi word, it is easy to see that 
$\mathrm{shape}(P_{\mathrm{shift}}(w))\neq\lambda$.
Thus, Eqn.(\ref{eqn:shapeYam1}) is true.

For Eqn.(\ref{eqn:shapeYam2}), observe that $P_{\mathrm{mix}}(w)$ 
does not have a primed entry.
Suppose that $w$ is a shifted Yamanouchi word of content $\lambda$.
From the properties (ii) of a shifted Yamanouchi word,
at least one $x-1$ appears left to $x$ in $w$.
The property (iii) implies that, in the word $w$, 
there are at least one $x$, $x<y$, in-between the $i$-th leftmost 
letter $y$ and the $(i+1)$-th leftmost letter $y$ for 
$1\le i\le \max(w)-y$.
Suppose $w=w'\ast y\ast v\ast v'$.
We insert $v'$ into a tableau $T=P_{\mathrm{mix}}(w'\ast y\ast v)$.
The above constraint implies that if $y$ is in the $j$-th row
of $T$, bumped out by $x$ and inserted into the $(j+1)$-th row,
there are at least one $z$, $j\le z\le y-1$, in the $(j+1)$-th 
row.
Thus, a letter $y$ cannot be placed on the main diagonal in the 
$j$-th row with $1\le j<y$.
After the mixed insertion, $P_{\mathrm{mix}}(w)$ does not have 
a primed entry and the Yamanouchi property ensures that 
$\mathrm{shape}(P_{\mathrm{mix}}(w))=\lambda$.
If $w$ is not a shifted Yamanouchi word, we have at least one 
primed entries, or the shape of $P_{\mathrm{mix}}(w)$ is not
equal to $\lambda$.
This completes the proof.
\end{proof}

Let $\lambda$ be a strict partition and $w$ be a word of content $\lambda$. 
Then, the left hand side of Eqn. (\ref{eqn:shapeYam2}) is equal to 
the number of $Q_{\mathrm{mix}}(w)$, that is the number of standard 
shifted tableaux of shape $\lambda$.
Let $\mathrm{SShTab}(\lambda)$ denote the set of standard shifted tableaux 
of shape $\lambda$.
\begin{prop}
We have 
\begin{eqnarray*}
|\mathrm{SShTab}(\lambda)|
=\#\{w |\ w \text{ is a shifted Yamanouchi word of content }\lambda \}.
\end{eqnarray*}
\end{prop}

\section{Semistandard increasing decomposition tableau}
\label{sec:SSIDT}
For a strict partition $\lambda=(\lambda_1,\lambda_2,\cdots,\lambda_l)$, 
we define two skew shapes $\epsilon^+(\lambda)$ and $\epsilon^-(\lambda)$ 
as follows.
The skew shape $\epsilon^{+}(\lambda)$ is obtained by transposing $\lambda$ 
and shifting the $i$-th row $l-i$ steps to the right  with respect to the $\lambda_1$-th 
row.
Let $\lambda'$ be the ordinary partition whose parts are the same as the ones of 
$\lambda$.
The skew shape $\epsilon^-(\lambda)$ is obtained from $\lambda'$ by rotating 
$\lambda'$ $90$ degrees in anti-clockwise direction and shifting the $i$-th row 
$l-i$ steps to the right with respect to the $\lambda_1$-th row.
The shapes $\epsilon^{\pm}(\lambda)$ have $l(\lambda)$ anti-diagonals and its 
$i$-th anti-diagonal is of length $\lambda_{i}$.
See Figure \ref{fig:epsipm} for an example.
\begin{figure}[ht]
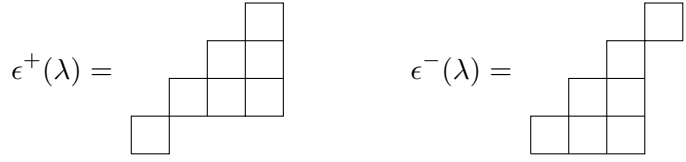

\begin{eqnarray*}
\epsilon^{+}(\lambda)=
\tikzpic{-0.5}{
\draw(0,0)--(0,1/2)--(1/2,1/2)--(1/2,0)--(0,0);
\draw(1/2,1/2)--(2,1/2)--(2,2)--(1.5,2)--(1.5,1/2);
\draw(1/2,1/2)--(1/2,1)--(2,1)(1,1/2)--(1,1.5)--(2,1.5);
}
\qquad\qquad
\epsilon^{-}(\lambda)=
\tikzpic{-0.5}{
\draw(0,0)--(3/2,0)--(3/2,2)--(2,2)--(2,1.5)--(1.5,1.5);
\draw(0,0)--(0,1/2)--(1.5,1/2);
\draw(1/2,0)--(1/2,1)--(1.5,1);
\draw(1,0)--(1,1.5)--(1.5,1.5);
}
\end{eqnarray*}
\caption{The shapes $\epsilon^{\pm}(\lambda)$ for $\lambda=(4,2,1)$.}
\label{fig:epsipm}
\end{figure}

We denote by $\epsilon^{\pm}(T)$ a semistandard tableau of the shape 
$\epsilon^{\pm}(\lambda)$ and corresponding to a shifted tableau $T$.
Especially, we consider a semistandard tableau $\epsilon^{\pm}(T)$ 
such that the shape of $T$ is $\lambda$ and 
$P_{\mathrm{mix}}(\mathrm{read}(\epsilon^{\pm}(T))=T$. 
We call a semistandard tableau $\epsilon^{\pm}(T)$ {\it a semistandard 
increasing decomposition tableau} (SSIDT).
A SSIDT can be seen as a shifted analogue of an ordinary Young 
tableau.
A reading word $w$ of $\epsilon^{\pm}(T)$ is decomposed into 
sequences $w=u_{l}u_{l-1}\cdots u_{1}$ 
where $l:=l(\lambda)$, each word $u_{i}$ is weakly increasing and 
$w$ is compatible with the shape $\epsilon^{\pm}(\lambda)$.

In the rest of this section, we construct a SSIDT $\epsilon^{\pm}(T)$
corresponding to a shifted tableau $T$. 
We first construct $\epsilon^{-}(T)$ from a SSDT.
Then, by using $\epsilon^{-}(T)$, we construct a SSIDT $\epsilon^{+}(T)$. 

\subsection{Construction of \texorpdfstring{$\epsilon^{-}(T)$}{e-(T)}}
Let $w$ be a word in the alphabet $X$. We denote by $\mathrm{SK}(w)$
the SSDT obtained by semistandard Kra\'skiewicz insertion and 
by $\lambda$ the shape of $\mathrm{SK}(w)$.
Let $R_{i}, 1\le i\le l(\lambda),$ be the $i$-th row of $\mathrm{SK}(w)$.
By definition, $R_{i}$ is a hook word of length $\lambda_i$.
Recall that $(R_{i}\downarrow)$ is the decreasing part of $R_{i}$.
Then, we have 
\begin{prop}[Theorem 4.1 in \cite{Lam96}]
\label{prop:dp}
The decreasing parts $(R_{i}\downarrow), 1\le i\le l(\lambda)$, form 
a strict partition.
\end{prop}

Below, we construct the tableau $\epsilon^{-}(\lambda)$ from a SSDT $\mathrm{SK}(w)$.

We denote by $\mu$ the strict partition formed by the decreasing parts of 
$\mathrm{SK}(w)$.
From the construction, the shape $\mu$ is of length $l(\lambda)$ and 
its part $\mu_i$ is smaller than or equal to $\lambda_i$.
We consider a up-right path from the leftmost box in the $l(\lambda)$-th row 
to the rightmost box in the first row.
A path is inside of the shape $\lambda$ and therefore the length of a path 
is $\lambda_1$.
Let $(x_{i},y_{i}), 1\le i\le \lambda_1$, be the coordinate of a box in 
the up-right path. 
Here, the coordinate system is matrix notation, that is, $x$ increases 
from top to bottom and $y$ increases from left to right.
The coordinate of the starting point is $(x_1,y_1):=(l(\lambda),l(\lambda))$.
First, we move along the shape of $\mu$, {\it i.e.,} the $i$-th coordinate is 
$(x_{i},y_{i})=(x_{i-1},y_{i-1}+1)$ if the box on $(x_{i-1},y_{i}+1)$ 
is in $\mu$ and $(x_{i},y_{i})=(x_{i-1}-1,y_{i-1})$ if the box on 
$(x_{i-1},y_{i-1}+1)$ is not in $\mu$.
If we arrive at $(1,\mu_{1})$, we move only right from $(1,\mu_{1})$ to 
$(1,\lambda_{1})$.
We denote by $p_{1}$ the obtained path.

Three boxes on $(x_{i-1},y_{i-1}), (x_{i-1},y_{i-1}+1)$ and $(x_{i-1}+1,y_{i-1})$
are in the path $p_{1}$ and the box on $(x_{i-1}+1,y_{i-1}+1)$ is not in $p_{1}$.
We say that a path $p_{1}$ has a corner at $(x_{i-1},y_{i-1})$.
We define a path $p_2$ from $p_1$ by bending the path $p_{1}$ at a corner as follows.
Suppose that $p_1$ has a corner at $z_{i}:=(x_{i},y_{i})$. Then, it is obvious that 
a box $z_{i+1}:=(x_{i}+1,y_{i}+1)$ is not in $p_1$.
If the content of the box on $z_{i}$ in the SSDT $\mathrm{SK}(w)$ is greater than 
or equal to that of the box on $z_{i+1}$, 
we obtain a new path $p_2$ by locally changing the path $p_{1}$ such that 
the path $p_2$ pass through $z_{i+1}$ instead of $z_{i}$.
We call this local change of a path as a bending.
We perform a bending at a corner of $p_2$ until no bending occurs at corners in 
the path any more. 
We denote by $p$ the obtained up-right path in $\lambda$.
Let $\mathrm{word}(p)$ be a word reading the content of $\mathrm{SK}(w)$ along 
the path $p$ starting from the box on $(x_1,y_1)$. 
Recall that the length of $\mathrm{word}(p)$ is $\lambda_1$.
We put the letters in $\mathrm{word}(p)$ on the first anti-diagonal of the shape 
$\epsilon^{-}(\lambda)$ from bottom to top.

We remove the boxes on the path $p$ from the SSDT $\mathrm{SK}(w)$ and denote 
by $\mathrm{SK}'(w)$ the tableau with removed boxes.
We introduce the {\it reverse SK insertion of type I} to obtain a new 
SSDT $\mathrm{SK}(w')$ whose shape is $(\lambda_2,\lambda_3,\ldots \lambda_{l})$.
We construct $\mathrm{SK}(w')$ from $\mathrm{SK}'(w)$ by the reverse insertion 
defined below.

\begin{defn}[reverse SK insertion of type I]
Let $R$ be a hook word with $(R\downarrow):=s_1\ldots s_{m}$ and 
$(R\uparrow):=s_{m+1}\ldots s_{l}$.
We reversely insert a letter $r\in X$ into $R$ in the following way:
\begin{enumerate}
\item let $s_{i}$ be the rightmost element in $(R\downarrow)$ which is 
greater than or equal to $r$
\item replace $s_{i}$ by $r$
\item we insert $s_i$ into $(R\uparrow)$ such that the obtained sequence 
is a weakly increasing one. 
\end{enumerate}
\end{defn}
We denote by $R\twoheadleftarrow \{r_1,r_2,\ldots,r_{l}\}$ the successive reverse 
SK insertions of $r_1,r_2,\ldots,r_{l}$ into $R$.

Let $L_{i}$ (resp. $R_{i}$) be a word 
in the $i$-th row of $\mathrm{SK}'(w)$ and left (resp. right) to the path $p$.
Suppose that $R_{i+1}:=r_{1}\ldots r_{k}$. 
The $i$-th row of $\mathrm{SK}(w')$ is obtained by the following 
reverse SK insertion: 
$L_{i}\twoheadleftarrow \{r_{k},r_{k-1},\ldots,r_{1}\}$.
By construction of a path $p$, the word $R_{i}$ is 
weakly increasing. 

We construct a path $p'$ from $\mathrm{SK}(w')$ and put the word $\mathrm{word}(p)$
on the second anti-diagonal of $\epsilon^{-}(\lambda)$ from bottom to top.
We repeat above procedures $l(\lambda)$ times and obtain the tableau 
$\epsilon^{-}(T)$ from $\mathrm{SK}(w)$.
See Example \ref{ex:epminus} below.

\begin{theorem}
\label{thrm:epminus}
The construction of $\epsilon^{-}(T)$ is well-defined, {\it i.e.}, 
a word obtained from $\mathrm{SK}(w)$ is compatible with the shape 
$\epsilon^{-}(\lambda)$.
Further, the tableau word $\mathrm{read}(\epsilon^{-}(T))$ produces the same 
mixed insertion tableau as $\mathrm{SK}(w)$.
\end{theorem}
Before we move to a proof of Theorem \ref{thrm:epminus}, we introduce 
five lemmas needed later.

\begin{lemma}
\label{lemma:epminus1}
Suppose that the word $w=w_{1}\cdots w_{2l}$ satisfies 
\begin{enumerate}
\item $w_{1}\cdots w_{l}$ and $w_{l+1}\cdots w_{2l}$ are weakly increasing, 
\item $w_{i}>w_{i+l-1}$ for $2\le i\le l$, 
\item $w_{l}\le w_{2l}$.
\end{enumerate}
Then, $w\sim w'$ where $w':=w_{1}\cdots w_{l}w_{2l}w_{l+1}\cdots w_{2l-1}$.
\end{lemma}
\begin{proof}
By using shifted plactic relations, we have 
\begin{eqnarray*}
w_{1}\cdots w_{l}w_{l+1}\cdots w_{2l-1}w_{2l}&\sim& 
w_{1}w_{2}w_{l+1}w_{3}w_{l+2}w_{4}\cdots w_{2l-2}w_{l}w_{2l-1}w_{2l}
\qquad \text{ by } (\ref{eqn:pl1}) \text{ and } (\ref{eqn:pl7}) \\
&\sim& w_{1}w_{2}w_{l+1}w_{3}w_{l+2}w_{4}\cdots w_{2l-2}w_{l}w_{2l}w_{2l-1}
\qquad \text{ by } (\ref{eqn:pl2}) \\ 
&\sim& w_{1}w_{2}w_{l+1}w_{3}w_{l+2}w_{4}\cdots w_{l-1}w_{l}w_{2l}w_{2l-2}w_{2l-1}
\quad \text{ by } (\ref{eqn:pl2}) \text{ and } (\ref{eqn:pl7}) \\
&\sim& w_{1}\cdots w_{l}w_{2l}w_{l+1}\cdots w_{2l-1}
\qquad \text{ by } (\ref{eqn:pl2}) \text{ and } (\ref{eqn:pl7}).
\end{eqnarray*}
\end{proof}

We introduce the inverse of the SK insertion.
\begin{defn}[reverse SK insertion of type II]
Let $w$ be a hook word satisfying Eqn. (\ref{eqn:hookword}).
Define $(w\Downarrow):=w_{1}\cdots w_{m-1}$ and 
$(w\Uparrow):=w_{m}\cdots w_{l}$. 
We reversely insert a letter $x\in X$ into $w$ and obtain 
a new word $w'$ with the elements $(y,z)$ as follows:
\begin{enumerate}
\item let $w_{i}$ be the rightmost element in $(w\Downarrow)$ 
which is greater than or equal to $x$, and set $y:=w_{i}$,
\item replace $w_{i}$ by $x$, 
\item let $w_{j}$ be the rightmost element in $(w\Uparrow)$ 
which is smaller than $w_{i}$,
\item replace $w_{j}$ by $w_{i}$ and bump out $w_{j}$. 
A word $w'$ is obtained from $w$ by replacing $w_{i}$ and 
$w_{j}$ by $x$ and $w_{i}$, and $z:=w_{j}$.  
\end{enumerate}
If there exists no $w_{j}$ in the step (3), we insert $w_{i}$ into 
$(w\Uparrow)$ such that the obtained sequence is weakly increasing.
\end{defn}

\begin{remark}
We apply the reverse SK insertion of type II only when $x\le w_{1}$
and $(w\Downarrow)\neq\emptyset$.
Otherwise, $xw$ is a hook word and the length of $xw$ is the length 
of $w$ plus one. This means that $w$ can not be a row of a semistandard
decomposition tableau.

The reverse SK insertion of type II is an inverse of the SK insertion.
For a given SK recording tableau, one can obtain a word by 
applying the reverse SK insertion of type II to the corresponding element 
in the SK insertion tableau.
\end{remark}

Let $w:=w_{l}\cdots w_{1}$ be a word such that $w_{i}$ 
is a hook word of length $\lambda_{i}$ and $w$ satisfies 
$w=\mathrm{read}(SK(w))$.
Since the word $\mathrm{read}(\mathrm{SK}(w))$ gives the 
same SK insertion tableau as $\mathrm{SK}(w)$ itself, 
$\mathrm{SK}(w)$ is of shape $\lambda$ if and only if 
a partial word $w_{i+1}w_{i}$ form a tableau word of shape 
$(\lambda_{i}, \lambda_{i+1})$.
Further, if $\lambda_{i}>\lambda_{i+1}+1$,
we have a SSDT of shape $(\lambda_{i+1}+1,\lambda_{i+1})$ 
by deleting from the $(\lambda_{i+1}+2)$-th to $\lambda_{i}$-th 
elements in $w_{i}$.
Therefore, it is enough to give criteria for a SSDT of shape 
$(n,n-1)$ to check whether a word $w$ gives a SK insertion 
tableau of shape $\lambda$.
Let $w:=w_{2}w_{1}$ and we denote $w_{1}:=u_{1}\cdots u_{n}$ 
and $w_{2}:=v_{1}\cdots v_{n-1}$.
We construct two sequences of letters $y:=y_{1}\cdots y_{n-1}$ 
and $z:=z_{1}\cdots z_{n-1}$ by using 
reverse SK insertion of type II as follows.
We start the process below with $i=n$ and $w'_{n}=w_{1}$, decrease $i$ one-by-one.
By reverse SK insertion of type II, we insert $v_{i-1}$ into $w'_{i}$ 
and obtain a new word $w''_{i-1}$ with elements $(y_{i-1},z_{i-1})$.
Since a word $w''_{i-1}$ is of length $i-1$, we delete the last element 
in $w''_{i-1}$ and obtain a new word $w'_{i-1}$.
We continue the process until we obtain words $y$ and $z$ of length $n-1$.
We denote by $w'_{i,j}$ an $j$-th element from left in a word $w'_{i}$. 

\begin{lemma}
\label{lemma:epminus2}
If a word $w:=w_{2}w_{1}$ produces a SSDT of shape $(n,n-1)$ by 
the SK insertion, then 
\begin{enumerate}
\item $v_{i-1}\le w'_{i,1}$ for $2\le i\le n$, 
\item $y_{i-1}>w'_{i,i}$ for $2\le i\le n$,
\item $(w'_{i}\Downarrow)\neq\emptyset$ for $2\le i\le n$.
\end{enumerate}
\end{lemma}
\begin{proof}
We first show that $v_{n-1}\le u_{1}$, $y_{n-1}>w'_{n,n}=u_{n}$ and 
$(w_{1}\Downarrow)\neq\emptyset$. 
Suppose that $v_{n-1}>u_{1}$. 
Since $w_{1}$ is a hook word, a word $v_{n-1}w_{1}$ is also a hook 
word of length $n+1$. 
By the SK insertion, it is obvious that the shape of $\mathrm{SK}(w_{1}w_{2})$
is not $(n,n-1)$.
Thus, we have $v_{n-1}\le u_{1}$.
Since $v_{n-1}\le u_{1}$, a word $w_{2}u_{1}$ is a hook word 
of length $n$.
Recall that the reverse SK insertion of type II is the inverse of 
the SK insertion.
By definitions of these insertions, $y_{n-1}$ is an element 
which is bumped out by the insertion of $u_{n}$ into 
$\mathrm{SK}(w_{2}u_{1}\ldots u_{n-1})$.
From the process (1) in Definition \ref{defn:SKins}, 
we have $y_{n-1}>u_{n}$. 
Suppose that $(w_{1}\Downarrow)=\emptyset$.
Then, $w_{1}=(w_{1}\Uparrow)$. 
Since $v_{n-1}\le u_{1}$, $w_{1}w_{2}$ is a hook word 
of length $2n-1$. 
The SK insertion tableau of $w_{2}w_{1}$ does not have 
the shape $(n,n-1)$, which is a contradiction to the assumption.
Thus, we have $(w_{1}\Downarrow)\neq\emptyset$.

It is easy to check that the reverse SK insertion of type II is an 
inverse of the SK insertion.
This implies that the word $v_{1}\cdots v_{i-1}w''_{i}z_{i}$
produces the same SK insertion tableau as $v_{1}\cdots v_{i}w'_{i+1}$.
Note that a word $w'_{i}$ is obtained from $w''_{i}$ by deleting the last 
element.
If $P_{\mathrm{SK}}(w)$ is of shape $(n,n-1)$, then 
$P_{\mathrm{SK}}(v_{1}\cdots v_{i-1}w'_{i})$ has to be 
of shape $(i,i-1)$.  
Combining this with the argument above, we obtain  
$v_{i-1}<w'_{i,1}$, $y_{i-1}>w'_{i,i}$ 
and $(w'_{i}\Downarrow)\neq\emptyset$.
\end{proof}

\begin{example}
For example, let $w:=w_{1}w_{2}$ with $w_{1}=123$ and $w_{2}=5433$.
The shape of $P_{\mathrm{SK}}(w)$ is not $(4,3)$, since we have 
a sequence of tableaux by the reverse SK insertion of type II:
\begin{eqnarray*}
\begin{matrix}
5 & 4 & 3 & 3 \\
  & 1 & 2 & 3 
\end{matrix}
\qquad
\Rightarrow
\qquad
\begin{matrix}
5 & 3 & 3 \\
  & 1 & 2 
\end{matrix}
\leftarrow4,3\qquad\Rightarrow\qquad
\begin{matrix}
2 & 3 \\
  & 1 
\end{matrix}
\leftarrow5,3,4,3
\end{eqnarray*}
The third tableau corresponding to the word $123$ 
violates the condition (3) in Lemma \ref{lemma:epminus2}.
Note that, for example, the SK insertion of $43$ to $533$ produces a word $35433$, 
and one can check that the reverse SK insertion of type II is 
the inverse of the SK insertion.
\end{example}

Let $w:=w_{l}\ldots w_{2}w_{1}$ be the reading word of a SSDT of shape 
$\lambda$.  
The word $w_{i}$ is a hook word of length $\lambda_{i}$.
We denote by $w_{i,j}$ the $j$-th element of $w_{i}$.
A strict partition $\lambda'$ is obtained from $\lambda$
by the reverse SK insertion of type II, namely, insert 
$w_{2,\lambda_{2}}$ into a word $w_{1}$.
From Lemma \ref{lemma:epminus2}, if we insert $w_{i,\lambda_{i}}$
into $w_{i-1}$ by the reverse SK insertion of type II, 
then we bump out $w_{i-1,\lambda_{i-1}}$ and successively 
insert $w_{i-1,\lambda_{i-1}}$ into $w_{i-2}$. 
Let $p$ be the path of length $\lambda_{1}$ constructed from 
the SSDT $\mathrm{SK}(w)$. 
By construction, the path $p$ contains the element $w_{1,\lambda_{1}}$.
If $\lambda_{1}>\lambda_{2}+1$, the path $p$ contains the elements 
$w_{1,j}$ with $\lambda_{2}+1\le j\le \lambda_{1}$.  
Thus, without loss of generality, it is enough to consider the 
case where $\lambda_{1}=\lambda_{2}+1$.
We consider two cases: (a) the element $w_{2,\lambda_{2}}$ is 
in the path $p$, and (b) $w_{2,\lambda_{2}}$ is not in $p$.

\begin{lemma}
\label{lemma:epminus3}
Suppose that the element $w_{2,\lambda_{2}}$ is in the path $p$.
Then, we have $(w_{1}\uparrow)=\emptyset$. 
\end{lemma}
\begin{proof}
Suppose that $(w_{1}\uparrow)\neq\emptyset$. 
From Proposition \ref{prop:dp}, decreasing parts $(w_{i}\downarrow)$ form
a strict partition $\mu$.
We have two cases: (i) $\mu_{1}>\mu_{2}+1$, and (ii) $\mu_{1}=\mu_{2}+1$.
\paragraph{Case (i)} The SSDT $w$ is locally given by 
\begin{eqnarray*}
\begin{matrix}
\cdots & a_{0} &a_{1} & a_{2} & \cdots & a_{m_{1}} & b_{1} & \cdots & b_{m_{2}} \\
& \cdots & c & d_{1} & \cdots & d_{m_{1}-1} & d_{m_{1}}  & \cdots & d_{m_1+m_2-1} \\
\end{matrix}
\end{eqnarray*}
where $b_{1}=\min\{w_{1}\}$, $c=\min\{w_{2}\}$, $b_{m_2}=w_{1,\lambda_{1}}$
and $d_{m_{1}+m_{2}-1}=w_{2,\lambda_{2}}$. 
Since $w_{1}$ and $w_{2}$ are hook words, we have 
\begin{eqnarray*}
&a_{1}>a_{2}>\cdots>a_{m_{1}}>b_{1}\le \cdots\le b_{m_{2}}, \\ 
&c\le d_{1}\le \cdots \le d_{m_1+m_2-1}.
\end{eqnarray*}
From the assumption, $b_{m_2}$ and $d_{m_1+m_2-1}$ are in the path $p$.
By construction of a path, a path before bending is 
$ca_{1}\cdots a_{m_1}b_{1}\cdots b_{m_{2}}$.
By bending the path, we have constraints on $a_{i}$ and $d_{i}$, namely 
$d_{i}\le a_{i}$ for $1\le i\le m_{1}$.
Similarly, we have constraints on $b_{i}$ and $d_{i}$, that is, 
$d_{m_{1}+i}\le b_{i}$ for $1\le i\le m_{2}$. 
By inserting elements from $d_{m_{1}+m_{2}-1}$ to $d_{m_{1}+1}$ into 
the first row by the reverse SK insertion of type II, the obtained 
SSDT looks locally like
\begin{eqnarray*}
\begin{matrix}
a'_{0} & a'_{1} & a'_{2} & \cdots & a'_{m_{1}} & b'_{1} & \cdots & b'_{m_{2}} \\
 & c & d_{1} & \cdots & d_{m_{1}-1}  & d_{m_{1}}
\end{matrix},
\end{eqnarray*}
If $a_{m_{1}}\le d_{m_{1}+1}$, we have $a'_{m_{1}}=a_{m_1}$.
If $a_{m_{1}}>d_{m_{1}+1}$, we have $a'_{j}=d_{m_{1}+1}$ for some $j\le m_{1}$.
Note that $d_{1}\le d_{2}\le \cdots\le d_{m_{1}}\le d_{m_{1}+1}\le b_{1}$ and 
we insert a letter into $(w\Downarrow)$ rather than $(w\downarrow)$ in case 
of the reverse SK insertion of type II.
Therefore, if we insert the elements form $d_{m_{1}}$  to $d_{1}$ into the 
first row by the reverse SK insertion of type II,
we have a increasing sequence $d_{1},d_{2},\ldots,d_{m_{1}+1}$ left to 
the position of $a_{m_{1}}$ in $w$.
Thus, the minimum of the first row appears at the position $a'_{0}$ or 
left to $a'_{0}$. On the other hand, $c$ is the minimum in the second row,
this contradicts to the fact that decreasing parts form a strict 
partition in a SSDT. 
Thus, we have $(w_{1}\uparrow)=\emptyset$.

\paragraph{Case (ii)}
The SSDT $w$ looks locally like 
\begin{eqnarray*}
\begin{matrix}
\ldots & a_{1} & c_{1} & \ldots & c_{m} \\
\ldots & a_{2} & d_{1} & \ldots & d_{m}
\end{matrix}
\end{eqnarray*}
where $a_{1}=\min\{w_{1}\}$ and $a_{2}=\min\{w_{2}\}$.
Since the path $p$ includes the elements $c_{m}$ and $d_{m}$, 
bending gives the constraints: $a_{1}\ge d_{1}$ and 
$c_{i}\ge d_{i+1}$ for $1\le i\le m-1$.
If $d_{m}\le a_{1}$, 
$d_{m}$ is placed left to $a_{1}$ by the reverse insertion 
of type II.
Then, the element $d_{m}$ is the minimum of the first row and 
left to $a_{2}$. This contradicts to Proposition \ref{prop:dp}.
If $d_{m}>a_{1}$, let $j$ be the maximal integer such that 
$d_{j}\le a_{1}$. 
The integer $j$ exists since we have $d_{1}\le a_{1}$.
By a similar argument to above, $d_{j}$ is placed 
left to $a_{1}$ by the reverse SK insertion of type II.
This also contradicts to Proposition \ref{prop:dp}.
Thus, we have $(w_{1}\uparrow)=\emptyset$.
\end{proof}

\begin{lemma}
\label{lemma:epminus4}
Suppose that the element $w_{2,\lambda_{2}}$ is in the path $p$.
Then, we have $w_{1,\lambda_{1}-1}\ge w_{2,\lambda_{2}}$.
\end{lemma}
\begin{proof}
From Lemma \ref{lemma:epminus3}, we have $(w_{1}\uparrow)=\emptyset$.
Let $w_{2,i}$ be the smallest element in $w_{2}$.
Then, we have two cases: (a) $i\neq \lambda_{2}$ and (b) $i=\lambda_{2}$.

For case (a), let $b=w_{2,i}$. The SSDT $\mathrm{SK}(w)$ looks locally like
\begin{eqnarray*}
\begin{matrix}
\ldots & a_1 & a_2 & \cdots & a_{n} \\
\ldots & b &  c_1 & \cdots & c_{n-1}
\end{matrix},
\end{eqnarray*}
where $a_{n}=w_{1,\lambda_{1}}$, $c_{n-1}=w_{2,\lambda_{2}}$, 
$a_{i}>a_{i+1}$ for $1\le i\le n-1$ and $c_{i}\le c_{i+1}$ for 
$1\le i\le n-2$.
A path before bending is $ba_{1}\cdots a_{n}$. 
Since $c_{n-1}$ is included in the path $p$, we have to bend 
the path at corners $a_{i}$ for $1\le i\le n-1$.
From the definition of bending, we have $a_{i}\ge c_{i}$ 
for $1\le i\le n-1$.
Thus we have $a_{n-1}=w_{1,\lambda_{1}-1}\ge w_{2,\lambda_{2}}=c_{n-1}$.

For case (b), suppose that $w_{1,\lambda_{1}-1}<w_{2,\lambda_{2}}$.
We insert $w_{1,i}$ for $1\le i\le \lambda_{1}$ into $w_{2}$ 
by the SK insertion.
Since $w_{1}$ and $w_{2}$ are strictly decreasing sequences, 
the length of the first row after the insertion is greater than 
$\lambda_{1}$.
This is a contradiction to the shape of $\mathrm{SK}(w)$.
Thus, we have $w_{1,\lambda_{1}-1}\ge w_{2,\lambda_{2}}$.
\end{proof}

Suppose that we have $\lambda_{i}=\lambda_{i+1}+1$. 
From Lemma \ref{lemma:epminus2},
we have $w_{i+1,\lambda_{i+1}}\le w_{i,1}$ and $(w_{i}\Downarrow)\neq\emptyset$.
Let $y$ be the rightmost element in $(w_{i}\Downarrow)$ 
which is greater than or equal to $w_{i+1,\lambda_{i+1}}$.
\begin{lemma}
\label{lemma:epminus5}
Suppose that $w_{i+1,\lambda_{i+1}}$ is not in the path $p$.
Then, the element $y$ is not included in the path $p$.
\end{lemma}
\begin{proof}
Suppose that a letter $y$ is included in the path $p$.
From the definition of $y$, we have $y\ge w_{i+1,\lambda_{i+1}}$.
Since $y\in(w_{1}\Downarrow)$, there exists $z\ge y$ such that 
a path has a corner at $z$.
A SSDT $\mathrm{SK}(w)$ locally looks like as follows:
\begin{eqnarray*}
\begin{matrix}
\ldots & z & \cdots & y & \cdots\\
\ldots & c_{0} & c_{1} & \cdots & c_{m} & w_{i+1,\lambda_{i+1}}
\end{matrix}.
\end{eqnarray*}
The elements $c_{0}z\cdots y\cdots$ form 
the path.
By definition, the elements right to the path form a weakly 
increasing sequence, we have 
$c_{1}\le c_{2}\le\ldots\le c_{m}\le w_{i+1,\lambda_{i+1}}$.
Since $p$ has a corner at $z$, we have $z<c_{1}$.
From these observations, we have $y\le z<c_{1}\le w_{i+1,\lambda_{i+1}}$, which 
is a contradiction to the assumption.
Therefore, $y$ is not included in $p$.
\end{proof}

\begin{proof}[Proof of Theorem \ref{thrm:epminus}]
We prove Proposition by induction.
Let $v:=\mathrm{read}(\mathrm{SK}(w))=v_{1}\cdots v_{|\lambda|}$ 
and $v':=v_{1}\cdots v_{|\lambda|-1}$.
We denote by $\lambda$ the shape of $\mathrm{SK}(v)$. 
We assume that Proposition is true for $\mathrm{SK}(v')$.
If $l(\lambda)=1$, it is obvious that 
$\mathrm{read}(\epsilon^{-}(T))=\mathrm{read}(\mathrm{SK}(w))$. 
This implies the word is compatible with the shape $\epsilon^{-}(\lambda)$ and 
$\mathrm{read}(\epsilon^{-}(T))$ has the same insertion tableau as the word $v$.

Suppose that $\lambda_{1}>\lambda_{2}+1$. 
We denote by $\lambda'$ a strict partition obtained from $\lambda$ 
by deleting the rightmost box of the first row of $\lambda$.
The shape of $\mathrm{SK}(v')$ is given by $\lambda'$.
From the construction of paths from $\mathrm{SK}(v')$ and 
$\mathrm{SK}(v)$, the element $v_{|\lambda|}$ is placed 
at the rightmost box of the first anti-diagonal 
in $\epsilon^{-}(\lambda)$.
This means that the word obtained from $\mathrm{SK}(v)$ is compatible with 
the shape $\epsilon^{-}(\lambda)$.
Since a word in $\epsilon^{-}(\lambda')$ has the same 
mixed insertion tableau $P_{\mathrm{mix}}(v')$ as $\mathrm{SK}(v')$ 
by assumption and $v_{|\lambda|}$ is the last element of the word in 
$\epsilon^{-}(\lambda)$, the word in $\epsilon^{-}(\lambda)$ and $v$
have the same mixed insertion tableau. 
Thus, Proposition holds true in this case.

Below, we assume $\lambda_{1}=\lambda_{2}+1$ without loss of generality.
Let $k$ be the maximal integer such that $\lambda_{1}=\lambda_{k}+k-1$ and 
$\lambda'$ be a strict partition obtained from $\lambda$ by replacing 
$\lambda_{k}$ with $\lambda_{k}-1$.
We denote by $v:=\mathrm{read}(\mathrm{SK}(w))$ the reading tableau 
word of $\mathrm{SK}(w)$.
By reversing the SK insertion, we have a word $w'$ of length $|\lambda|-1$ 
such that $\mathrm{SK}(w')$ is of shape $\lambda'$ and 
there exists a letter $x$ satisfying $\mathrm{SK}(w)=\mathrm{SK}(w'x)$.
Since $v$ produces the shape $\lambda$, one can take $x=v_{|\lambda|}$. 
Let $v'$ be the tableau word $v':=\mathrm{read}(\mathrm{SK}(w'))$.
Suppose that $v'_{|\lambda|-1}\le x$.
The letter $v'_{|\lambda|-1}$ is the rightmost content of the first 
row of $\mathrm{SK}(w')$.
The insertion of $x$ to the first row of $\mathrm{SK}(w')$ results in 
$v'x$. This implies that the length of the first row increases after 
the insertion. This contradicts $\lambda_{1}=\lambda'_{1}$.
Thus, we have $v'_{|\lambda|-1}>x$.
Let $R_{1}$ be the first row of $\mathrm{SK}(w')$.
Suppose that a hook word $R_{1}$ does not have an increasing part, 
{\it i.e.}, $(R\downarrow)=\emptyset$.
If we insert $x$ into $R_{1}$, the word $R_{1}x$ is a hook word 
of length $\lambda'_{1}+1$.
This contradicts $\lambda_{1}=\lambda'_{1}$.
Thus, a hook word $R_{1}$ has an increasing part, 
{\it i.e.}, $(R_{1}\uparrow)\neq\emptyset$.

Let $\epsilon^{-}(T)$ (resp. $\epsilon^{-}(T')$) be a tableau word of 
shape $\epsilon^{-}(\lambda)$ (resp. $\epsilon^{-}(\lambda')$) corresponding 
to $\mathrm{SK}(w)$ (resp. $\mathrm{SK}(w')$). 
We denote by $q_{i}$ the tableau word of length $\lambda_{i}$ obtained 
by reading the $i$-th anti-diagonal of $\epsilon^{-}(T)$. 
We define $q'_{i}$ for $\epsilon^{-}(T')$ in a similar way.

Let $p'$ be a path of length $\lambda'_{1}$ in $\mathrm{SK}(w')$.
Suppose that the path $p'$ is obtained from a path $p''$ by bending at corners.
The path $p''$ is along the shape $\mu$ where $\mu$ is a strict partition 
formed by decreasing parts of $\mathrm{SK}(w')$.
We denote by $\mathrm{read}(p''):=p''_{1}\cdots p''_{\lambda_{1}}$ 
the reading word along the path $p''$ in $\mathrm{SK}(w')$.
Since $(R_{1}\uparrow)\neq\emptyset$, we have 
$p''_{\lambda_{1}-1}\le p''_{\lambda_{1}}$. 
By definition, we have $p'_{\lambda_{1}}=p''_{\lambda_{1}}$.
However, since a bending at a corner gives an equal or smaller letter, 
we have $p'_{\lambda_{1}-1}\le p''_{\lambda_{1}}=p'_{\lambda_{1}}$.
Similarly, for $1\le i\le k$, the $i$-th row of $\mathrm{SK}(w')$ has a non-empty 
increasing part since the decreasing parts of $\mathrm{SK}(w')$ 
form a strict partition (see Proposition \ref{prop:dp}).
We denote by $p'_{i}$ the path obtained from $\mathrm{SK}_{i}(w')$, 
where $\mathrm{SK}_{1}(w'):=\mathrm{SK}(w')$ and 
$\mathrm{SK}_{i+1}(w')$ is a SSDT obtained from $\mathrm{SK}_{i}(w')$ by 
removing boxes in $p'_{i}$ and successively applying the reverse 
SK insertion of type I.
Let $R_{1}^{i}$ be the first row of $\mathrm{SK}_{i}(w')$.
Since the $i$-th row of $\mathrm{SK}(w')$ has an increasing part and 
the reverse SK insertion of type I does not decrease the length of 
the increasing part, we have $(R_{1}^{i}\downarrow)\neq\emptyset$.
From Lemma \ref{lemma:epminus3}, the rightmost element of the second 
row of $\mathrm{SK}_{i}(w')$ is not in the path $p'_{i}$.
Therefore, the rightmost and the second rightmost elements of the 
first row in $\mathrm{SK}_{i}(w')$ are both in the path $p'_{i}$.
This implies $p'_{i,\lambda_{i}-1}\le p'_{i,\lambda_{i}}$.

Let $r$ be the rightmost element in the second row of $\mathrm{SK}(w')$.
Since $r$ is not in the path $p'_{1}$, we apply the reverse SK insertion 
of type I to $r$. 
By definition, $r$ bumps out an element $y$ greater than or equal to $r$, 
that is, $r\le y$.
Since we can not apply bending in $p'_{1}$, we have $p'_{1,\lambda_{1}-1}<r$.
From these, we have $p'_{1,\lambda_{1}-1}<r\le y\le p'_{2,\lambda_{2}}$, 
which implies $p'_{1,\lambda_{1}-1}\le p'_{2,\lambda_{2}}$.
By a similar argument, we have 
$p'_{i,\lambda_{i}-1}\le p'_{i+1,\lambda_{i+1}}$.

Suppose that $\epsilon^{-}(T')$ corresponds to $\mathrm{SK}(w')$. 
Let $a_{1},\cdots,a_{k}$ be the elements forming the second rightmost column of 
$\epsilon^{-}(T')$ and $b_{1},\cdots,b_{k}$ be the elements forming the rightmost 
column.
By summarizing the above observations, these elements satisfy
$a_{i}\le b_{i}$ for $1\le i\le k-1$, $a_{i}\le b_{i+1}$ for $1\le i\le k-2$.
If we denote $b_{0}:=x$, we have $b_{0}<b_{1}$.
Thus, we can put $b_{0}$ above $b_{1}$ in $\epsilon^{-}(T')$ which corresponds
to the insertion of $x$ into $\epsilon^{-}(T')$.
By using Lemma \ref{lemma:epminus1} successively, we obtain a new 
tableau $\epsilon^{-}(T'')$ of the shape $\epsilon^{-}(\lambda)$.
A schematic procedure is as follows:
\begin{eqnarray*}
\begin{matrix}
 & b_{0}  \\
  & b_{1} \\
a_{1} & b_{2} \\
\vdots & \vdots \\ 
a_{k-2} & b_{k-1} \\ 
a_{k-1} &   \\ 
a_{k} &  
\end{matrix}\quad \rightarrow \quad
\begin{matrix}
  & b_{0} \\
a_{1} & b_{1} \\
\vdots & \vdots \\ 
a_{k-1} & b_{k-1}   \\ 
a_{k} &  
\end{matrix}.
\end{eqnarray*}
By construction, $\epsilon^{-}(T'')$ is compatible with the shape 
$\epsilon^{-}(\lambda)$ and its tableau word produces the same 
mixed insertion tableau as $\mathrm{SK}(w)$.

To prove Theorem, we have to show that $\epsilon^{-}(T)$ obtained 
from $\mathrm{SK}(w)$ is nothing but $\epsilon^{-}(T'')$.
Suppose that the path $p$ of length $\lambda_{1}$ constructed from 
$\mathrm{SK}(w)$ contains $w_{i,\lambda_{i}}$ for $1\le i\le q\le k$.
From  Lemma \ref{lemma:epminus3}, a word $w_{i}$, $1\le i\le q$, is 
strictly decreasing.
The element $w_{q,\lambda_{q}-1}$ is contained by $p$ since a path 
is an up-right path.
We insert $w_{k,\lambda_{k}}$ into $w_{k-1}\cdots w_{1}$ by the reverse 
SK insertion of type II.
The insertion is characterized by Lemma \ref{lemma:epminus2}.
From Lemma \ref{lemma:epminus2} and Lemma \ref{lemma:epminus5}, 
the partial path from $w_{l(\lambda),1}$ to 
$w_{q,\lambda_{q}-1}$ is not changed by the insertion.
The partial path from $w_{q,\lambda_{q}-1}$ to $w_{1,\lambda_1}$ 
changes locally by the insertion as follows:
\begin{eqnarray*}
\begin{matrix}
 & w_{1,\lambda_{1}-1} & w_{1,\lambda_{1}} \\
 & \vdots & \vdots \\
 & w_{q-1,\lambda_{q-1}-1} & w_{q-1,\lambda_{q-1}} \\
 & w_{q,\lambda_{q}-1} &  w_{q,\lambda_{q}}
\end{matrix}\qquad\Rightarrow\qquad
\begin{matrix}
 & w_{2,\lambda_{2}} & w_{1,\lambda_{1}-1} \\
 & \vdots & \vdots \\
 & w_{q,\lambda_{q}} & w_{q-1,\lambda_{q-1}-1} \\
 & w_{q,\lambda_{q}-1} &  \alpha
\end{matrix}\leftarrow w_{1,\lambda_{1}},
\end{eqnarray*}
where $\alpha\ge w_{q,\lambda_{q}}$ and $\leftarrow w_{1,\lambda_{1}}$ means 
that this element is bumped out by the insertion.
Since $w_{i,\lambda_{i}}\le w_{i-1,\lambda_{i-1}-1}$ for 
$2\le i\le q$ (by Lemma \ref{lemma:epminus4}), 
the new path of length $\lambda_{1}$ is identical 
to the old path except the last element. 
Note that $w_{1,\lambda_{1}}<w_{1,\lambda_{1}-1}$.
This condition corresponds to 
putting $w_{1,\lambda_{1}}$ above $w_{1,\lambda_{1}-1}$ in $\epsilon^{-}(T')$. 
To obtain paths from $\mathrm{SK}(w)$, we delete the path of length $\lambda_{1}$ and 
perform the reverse SK insertion of type I on it.
Let $p_i$ (resp. $p'_i$) be the $i$-th path of length $\lambda_{i}$ 
(resp. $\lambda'_{i}$) obtained from $\mathrm{SK}(w)$ (resp. $\mathrm{SK}(w')$).
The difference between $p_i$ and $p'_i$ is the last element by using a similar 
argument above.
More precisely, we have $p'_{i,\lambda'_{i}}=p_{i+1,\lambda_{i+1}}$
for $1\le i\le k-1$ and $p_{i}=p'_{i}$ for $k<i\le l(\lambda)$.
In this way, we obtain $\epsilon^{-}(T')$ from $\mathrm{SK}(w')$.
Combining these observations, we have shown that $\epsilon^{-}(T)$ is 
$\epsilon^{-}(T'')$. 
This completes the proof.
\end{proof}

\begin{example}
\label{ex:epminus}
Let $\lambda=(7,4,3)$. 
An example of a SSDT of shape $(7,4,3)$ is the 
leftmost tableau in Figure \ref{fig:epminus}.
From this SSDT, we obtain a path $p_{1}$ of length $7$ and 
its word is given by $\mathrm{word}(p_{1})=3112344$.
\begin{figure}[ht]
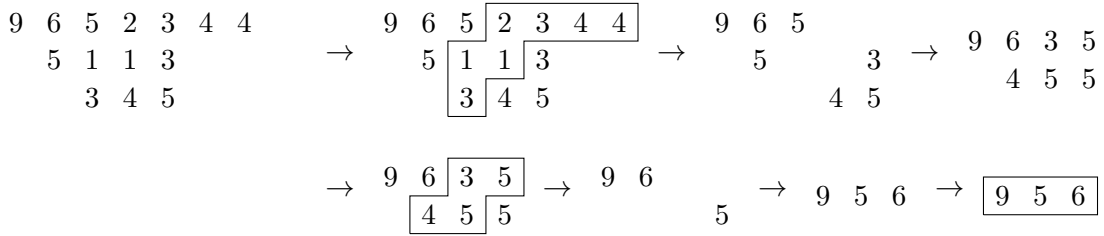

\begin{eqnarray*}
\tikzpic{-0.5}{
\draw(0,0)node{$9$}(0.5,0)node{$6$}(1,0)node{$5$}(1.5,0)node{$2$}
          (2,0)node{$3$}(2.5,0)node{$4$}(3,0)node{$4$};
\draw(0.5,-0.5)node{$5$}(1,-0.5)node{$1$}(1.5,-0.5)node{$1$}
          (2,-0.5)node{$3$};
\draw(1,-1)node{$3$}(1.5,-1)node{$4$}
          (2,-1)node{$5$};
}&&\rightarrow
\tikzpic{-0.5}{
\draw(0,0)node{$9$}(0.5,0)node{$6$}(1,0)node{$5$}(1.5,0)node{$2$}
          (2,0)node{$3$}(2.5,0)node{$4$}(3,0)node{$4$};
\draw(0.5,-0.5)node{$5$}(1,-0.5)node{$1$}(1.5,-0.5)node{$1$}
          (2,-0.5)node{$3$};
\draw(1,-1)node{$3$}(1.5,-1)node{$4$}
          (2,-1)node{$5$};
\draw(0.75,-1.25)--(0.75,-0.25)--(1.25,-0.25)--(1.25,0.25)--(3.25,0.25)
      --(3.25,-0.25)--(1.75,-0.25)--(1.75,-0.75)--(1.25,-0.75)
      --(1.25,-1.25)--(0.75,-1.25);
}\rightarrow
\tikzpic{-0.5}{
\draw(0,0)node{$9$}(0.5,0)node{$6$}(1,0)node{$5$};
\draw(0.5,-0.5)node{$5$}(2,-0.5)node{$3$};
\draw(1.5,-1)node{$4$}(2,-1)node{$5$};
}\rightarrow
\tikzpic{-0.5}{
\draw(0,0)node{$9$}(0.5,0)node{$6$}(1,0)node{$3$}(1.5,0)node{$5$};
\draw(0.5,-0.5)node{$4$}(1,-0.5)node{$5$}(1.5,-0.5)node{$5$};
} \\[11pt]
&&\rightarrow
\tikzpic{-0.5}{
\draw(0,0)node{$9$}(0.5,0)node{$6$}(1,0)node{$3$}(1.5,0)node{$5$};
\draw(0.5,-0.5)node{$4$}(1,-0.5)node{$5$}(1.5,-0.5)node{$5$};
\draw(0.25,-0.75)--(0.25,-0.25)--(0.75,-0.25)--(0.75,0.25)--(1.75,0.25)
     --(1.75,-0.25)--(1.25,-0.25)--(1.25,-0.75)--(0.25,-0.75);
}\rightarrow
\tikzpic{-0.5}{
\draw(0,0)node{$9$}(0.5,0)node{$6$};
\draw(1.5,-0.5)node{$5$};
}\rightarrow
\tikzpic{-0.5}{
\draw(0,0)node{$9$}(0.5,0)node{$5$}(1,0)node{$6$};
}\rightarrow
\tikzpic{-0.5}{
\draw(0,0)node{$9$}(0.5,0)node{$5$}(1,0)node{$6$};
\draw(-0.25,-0.25)--(-0.25,0.25)--(1.25,0.25)--(1.25,-0.25)--(-0.25,-0.25);
}
\end{eqnarray*}
\caption{A semistandard decomposition tableau and its paths.}
\label{fig:epminus}
\end{figure}
By removing the boxed elements, we obtain a SSDT of shape $(4,3)$ (see 
rightmost picture in the first row in Figure \ref{fig:epminus}).
By extracting a path $p_{2}$, $\mathrm{word}(p_{2})=4535$, we obtain 
a SSDT of shape $(3)$. A path $p_{3}$ is equal to the third SSDT, namely 
$\mathrm{word}(p_{3})=956$.
By putting the elements in the paths $p_{1}, p_{2}$ and $p_{3}$ on the anti-diagonals, 
we obtain a semistandard increasing decomposition tableau $\epsilon^{-}(T)$ of shape 
$\epsilon^{-}(\lambda)$ (see Figure \ref{fig:epminus2}) 
corresponding to the SSDT of shape 
$(7,4,3)$.
One can easily check that the  words constructed from the SSDT and 
$\epsilon^{-}(\lambda)$ produce the same shifted tableau by the mixed 
insertion. 
\begin{figure}[ht]
\begin{eqnarray*}
\epsilon^{-}(T)=
\begin{matrix}
  & & & & & & 4 \\
  & & & & & 4 & \\
  & & & & 3 & & \\
  & & & 2 & 5 \\
  & & 1 & 3 & 6 \\
  & 1 & 5 & 5 \\
3 & 4 & 9 
\end{matrix}\quad\quad
P_{\mathrm{mix}}(\epsilon^{-}(T))
=
\begin{matrix}
1 & 1 & 2 & 3' & 3 & 4 & 4\\
  & 3 & 4' & 5 & 5 \\
  &   & 5 & 6 & 9'
\end{matrix}
\end{eqnarray*}
\caption{A SSIDT $\epsilon^{-}(T)$ and a shifted tableau corresponding to 
$\epsilon^{-}(T)$.}
\label{fig:epminus2}
\end{figure}
\end{example}

\subsection{Construction of \texorpdfstring{$\epsilon^{+}(T)$}{e+(T)}}
\label{sec:SSIDTplus}
Let $\lambda$ be a strict partition.
We construct a bijection between a semistandard shifted tableau $T$ of shape 
$\lambda$ and a semistandard shifted tableau $\epsilon^{+}(T)$. 
We call also $\epsilon^{+}(T)$ a semistandard increasing decomposition tableau.
There are five steps for this bijection: 
\begin{enumerate}
\item perform the standardization on $T$ and obtain a standard tableau $\mathrm{stand}(T)$, 
\item perform a ``dual" operation on $\mathrm{stand}(T)$ with respect to 
primes and obtain $\overline{T}$,
\item construct a SSDT from $\overline{T}$, then obtain 
a tableau $\epsilon^{-}(\overline{T})$ from the SSDT,
\item take a flip of $\epsilon^{-}(\overline{T})$ and obtain 
a standard tableau $\epsilon^{+}(\overline{T})$,
\item destandardization of $\epsilon^{+}(\overline{T})$.
\end{enumerate}

\paragraph{Step 1:}
We denote by $\#(\alpha)$ the number of a letter $\alpha\in X'$ in a 
tableau $T$. 
We define a standard tableau $\mathrm{stand}(T)$ of $T$ from $T$ as follows.
We replace $1'$'s in $T$ with the letter $1,2,\ldots,\#(1')$ from 
top to bottom. 
Successively, replace $1$'s by the letter 
$\#(1')+1,\#(1')+2,\ldots,\#(1')+\#(1)$ from left to right.
Continue with $2'$'s in $T$ which are replaced by 
$\#(1')+\#(1)+1,\ldots,\#(1')+\#(1)+\#(2')$, and so on, 
until we obtain a standard tableau.
We denote by $\mathrm{stand}(T)$ the standard shifted tableau corresponding 
to a tableau $T$.
Note that when we replace primed (resp. unprimed) letter with a letter, 
we work from top to bottom (resp. from left to right).
We define a content as $\#(\lambda):=(\#(1')+\#(1),\#(2')+\#(2),\ldots)$.

\paragraph{Step 2:}
A given standard shifted tableau $\mathrm{stand}(T)$, we consider the following 
``dual" operation:
\begin{enumerate}[(i)]
\item letters on the main diagonal are unchanged,  
\item a letter $\alpha\in X'$ which is not on the main diagonal is replaced 
by $\alpha'$.  Here, we set $(i)'=i'$ and $(i')'=i$ for $i\in X$.
\end{enumerate}
The new semistandard shifted tableau is denoted by $\overline{T}$.
Note that this operation is an involution, {\it i.e.}, 
$\overline{\overline{T}}=T$.

\paragraph{Step 3:}
For our purpose, it is enough to find a word $w\in X$ such that 
$P_{\mathrm{mix}}(w)=\overline{T}$. 
This is easily done by reversing the mixed insertion starting from 
a pair of tableaux $(\overline{T},U)$ where $U$ is an arbitrary chosen 
standard shifted tableau.
Once $w$ is fixed, we can obtain a SSDT by the SK insertion and 
a tableau $\epsilon^{-}(\overline{T})$ from the SSDT.

\paragraph{Step 4:}
We reflect $\epsilon^{-}(\overline{T})$ over the diagonal line and 
obtain a standard tableau $\epsilon^{+}(\overline{T})$. 

\paragraph{Step 5:}
Let $\alpha$ be the content of $T$. 
We replace the letters from $1$ to $\alpha_1$ in $\epsilon^{+}(\overline{T})$
by $1$, the letters from $\alpha_1+1$ to $\alpha_1+\alpha_2$ in $\epsilon^{+}(\overline{T})$
by $2$, and so on.
This destandardization of $\epsilon^{+}(\overline{T})$ gives 
a semistandard increasing decomposition tableau $\epsilon^{+}(T)$.

\begin{theorem}
\label{thrm:epplus}
Let $T$ be a shifted tableau of shape $\lambda$ and $\epsilon^{+}(T)$ be a semistandard 
increasing decomposition tableau of shape $\lambda$ constructed by the above steps.
Then, $\epsilon^{+}(T)$ is well-defined, {\it i.e.}, the reading 
word $w:=\mathrm{read}(\epsilon^{+}(T))$ satisfies 1) the word $w$ is compatible with 
the shape $\epsilon^{+}(\lambda)$ and 2) $P_{\mathrm{mix}}(w)=T$.
\end{theorem}

Before proceeding to a proof of Theorem \ref{thrm:epplus}, we introduce 
a proposition and two lemmas needed for the proof.

\begin{prop}[Proposition 8.8 in \cite{Hai89}]
\label{prop:epplus}
We have $\overline{P_{\mathrm{mix}}(w)}=P_{\mathrm{mix}}(\mathrm{rev}(w))$.
\end{prop}

Given a shifted tableau $T$, we denote $S=\mathrm{stand}(T)$. 
\begin{lemma}
\label{lemma:epplus1}
Let $w$ be the reading word of $\epsilon^{-}(S)$ and $\overline{w}$
be the reading word of a tableau which is obtained  
by reflecting $\epsilon^{-}(S)$ over the diagonal line.
Then, we have $\overline{w}\sim \mathrm{rev}(w)$.
\end{lemma}
\begin{proof}
We first show Lemma is true when $\lambda=(n,n-1,\ldots,1)$.
In this case, we have $\epsilon^{+}(\lambda)=\epsilon^{-}(\lambda)$.
We enumerate the rows of $\epsilon^{-}(\lambda)$ by $1,\ldots,n$ 
from top to bottom and the boxes in the $i$-th row by $1,\ldots,i$ 
for $1\le i\le n$.
Let $a_{i,j}$ be the content of the $j$-th box in the $i$-th 
row in $\epsilon^{-}(S)$.
By definition, we have 
\begin{eqnarray*}
w&=&a_{n,1}\cdots a_{n,n}a_{n-1,1}\cdots a_{n-1,n-1}\cdots a_{2,1}a_{2,2}a_{1,1} \\
\overline{w}&=&a_{1,1}a_{2,2}\cdots a_{n,n}a_{2,1}\cdots a_{n,n-1}\cdots a_{n-1,1}a_{n,2}a_{n,1} \\
\end{eqnarray*}
Since $S$ is standard, $a_{i,j}$ satisfies 
\begin{eqnarray*}
&a_{i,1}<a_{i,2}<\ldots<a_{i,i},\\
&a_{i,1}<a_{i+1,2}<\ldots<a_{n,n-i+1},
\end{eqnarray*}
for $1\le i\le n$.
We depict $\mathrm{rev}(w)$ and $\overline{w}$ as 
\begin{eqnarray*}
\mathrm{rev}(w)=
\tikzpic{-0.5}{
\draw(0,0)--(3/2,0)--(3/2,3/2)--(0,0);
\draw(1.1,0.5)node{$\Leftarrow$};
}\qquad,\qquad
\overline{w}=
\tikzpic{-0.5}{
\draw(0,0)--(3/2,0)--(3/2,3/2)--(0,0);
\draw(1.1,0.5)node{$\Downarrow$};
},
\end{eqnarray*}
where $\Leftarrow$ means that we read a word from right to left 
starting from the top row to bottom, and $\Downarrow$ means that 
we read a word from top to bottom starting from the rightmost column
to the leftmost column.
We make use of induction on $n$. Suppose that we have 
$\overline{w}\sim\mathrm{rev}(w)$ for some $n$.
For $n+1$, we have 
\begin{eqnarray}
\label{eqn:revw}
\tikzpic{-0.5}{
\draw(0,0)--(3/2,0)--(3/2,3/2)--(0,0);
\draw(1.1,0.5)node{$\Downarrow$};
}\quad\sim\quad
\tikzpic{-0.5}{
\draw(0,0)--(1.35,0)--(1.35,1.35)--(0,0);
\draw(1,0.5)node{$\Downarrow$};
\draw(1.35,0)--(1.65,0)--(1.65,1.65)--(1.35,1.65)--(1.35,0);
\draw(1.5,0.8)node{$\Downarrow$};
}\quad\sim\quad
\tikzpic{-0.5}{
\draw(0.3,0.3)--(1.35,0.3)--(1.35,1.35)--(0.3,0.3);
\draw(1,0.6)node{$\Leftarrow$};
\draw(1.35,0)--(1.65,0)--(1.65,1.65)--(1.35,1.65)--(1.35,0);
\draw(1.5,0.8)node{$\Downarrow$};
\draw(0,0)--(1.35,0)--(1.35,0.3)--(0,0.3)--(0,0);
\draw(0.7,0.15)node{$\Leftarrow$};
}.
\end{eqnarray}
In Eqn. (\ref{eqn:revw}), the second picture means we first read the rightmost
column and then read the word in the triangle of size $n$, the third picture
is obtained by applying the induction assumption on the second picture and by 
decomposing the triangle of size $n$ into a single row and a triangle of size $n-1$.
Since $a_{n,n}$ is the maximal content, by applying Lemma \ref{lemma:epminus1} 
successively, the last picture in Eqn. (\ref{eqn:revw}) is equal to 
\begin{eqnarray*}
\tikzpic{-0.5}{
\draw(0.3,0.3)--(1.35,0.3)--(1.35,1.35)--(0.3,0.3);
\draw(1,0.6)node{$\Leftarrow$};
\draw(1.35,0.3)--(1.65,0.3)--(1.65,1.65)--(1.35,1.65)--(1.35,0.3);
\draw(1.5,0.8)node{$\Downarrow$};
\draw(0,0)--(1.65,0)--(1.65,0.3)--(0,0.3)--(0,0);
\draw(0.7,0.15)node{$\Leftarrow$};
}\quad\sim\quad
\tikzpic{-0.5}{
\draw(0.3,0.3)--(1.35,0.3)--(1.35,1.35)--(0.3,0.3);
\draw(1,0.6)node{$\Downarrow$};
\draw(1.35,0.3)--(1.65,0.3)--(1.65,1.65)--(1.35,1.65)--(1.35,0.3);
\draw(1.5,0.8)node{$\Downarrow$};
\draw(0,0)--(1.65,0)--(1.65,0.3)--(0,0.3)--(0,0);
\draw(0.7,0.15)node{$\Leftarrow$};
}\quad\sim\quad
\tikzpic{-0.5}{
\draw(0.3,0.3)--(1.65,0.3)--(1.65,1.65)--(0.3,0.3);
\draw(1.2,0.65)node{$\Downarrow$};
\draw(0,0)--(1.65,0)--(1.65,0.3)--(0,0.3)--(0,0);
\draw(0.7,0.15)node{$\Leftarrow$};
}\quad\sim\quad
\tikzpic{-0.5}{
\draw(0.3,0.3)--(1.65,0.3)--(1.65,1.65)--(0.3,0.3);
\draw(1.2,0.65)node{$\Leftarrow$};
\draw(0,0)--(1.65,0)--(1.65,0.3)--(0,0.3)--(0,0);
\draw(0.7,0.15)node{$\Leftarrow$};
}\quad\sim\quad
\tikzpic{-0.5}{
\draw(0,0)--(3/2,0)--(3/2,3/2)--(0,0);
\draw(1.1,0.5)node{$\Leftarrow$};
}
\end{eqnarray*}
where we have used the induction assumption in the first and 
third equality.

For a general $\lambda$, recall that $\epsilon^{+}(\lambda)$ has 
$\lambda_{i}$ boxes in the $i$-th anti-diagonal.
The shape $\epsilon^{+}(\lambda)$ is obtained as a union of 
the shapes $\epsilon^{+}(\mu)$'s where $\mu$ is a staircase.
For example, if $\lambda=(4,3,1)$, $\epsilon^{+}(\lambda)$
is a union of $\epsilon^{+}(\mu_{1})$ and $\epsilon^{+}(\mu_{2})$
with $\mu_{1}=(3,2,1)$ and $\mu_{2}=(2,1)$ (see Figure \ref{fig:revw}).
\begin{figure}
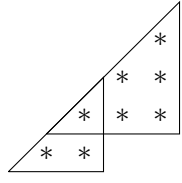

\begin{eqnarray*}
\tikzpic{-0.5}{
\draw(0.5,0.5)node{$\ast$}(1,0.5)node{$\ast$}
     (1,1)node{$\ast$}(1.5,1)node{$\ast$}(2,1)node{$\ast$}
     (1.5,1.5)node{$\ast$}(2,1.5)node{$\ast$}
     (2,2)node{$\ast$};
\draw(0,0.25)--(1.25,0.25)--(1.25,1.5)--(0,0.25);
\draw(0.5,0.75)--(2.25,0.75)--(2.25,2.5)--(0.5,0.75);
}
\end{eqnarray*}
\caption{An example of $\epsilon^{+}(\lambda)$.}
\label{fig:revw}
\end{figure}
Suppose that $\mu_{1}$ and $\mu_{2}$ are staircases, and 
$\epsilon^{+}(\mu_{1})$ and $\epsilon^{+}(\mu_{2})$ 
overlap.
We have 
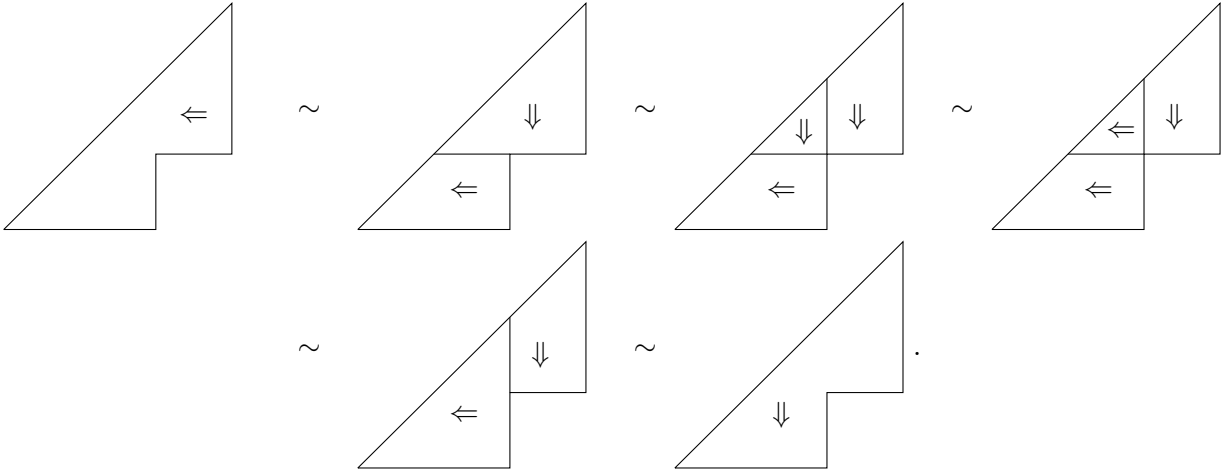
\begin{eqnarray*}
\tikzpic{-0.5}{
\draw(0,0)--(2,0)--(2,1)--(3,1)--(3,3)--(0,0);
\draw(2.5,1.5)node{$\Leftarrow$};
}\quad&\sim&
\tikzpic{-0.5}{
\draw(0,0)--(2,0)--(2,1)--(3,1)--(3,3)--(0,0);
\draw(2.3,1.5)node{$\Downarrow$};
\draw(2,1)--(1,1);
\draw(1.4,0.5)node{$\Leftarrow$};
}\quad\sim
\tikzpic{-0.5}{
\draw(0,0)--(2,0)--(2,1)--(3,1)--(3,3)--(0,0);
\draw(2.4,1.5)node{$\Downarrow$};
\draw(2,1)--(1,1);
\draw(1.4,0.5)node{$\Leftarrow$};
\draw(2,1)--(2,2);
\draw(1.7,1.3)node{$\Downarrow$};
}\quad\sim
\tikzpic{-0.5}{
\draw(0,0)--(2,0)--(2,1)--(3,1)--(3,3)--(0,0);
\draw(2.4,1.5)node{$\Downarrow$};
\draw(2,1)--(1,1);
\draw(1.4,0.5)node{$\Leftarrow$};
\draw(2,1)--(2,2);
\draw(1.7,1.3)node{$\Leftarrow$};
} \\
&\sim&
\tikzpic{-0.5}{
\draw(0,0)--(2,0)--(2,1)--(3,1)--(3,3)--(0,0);
\draw(2.4,1.5)node{$\Downarrow$};
\draw(1.4,0.7)node{$\Leftarrow$};
\draw(2,1)--(2,2);
}\quad\sim
\tikzpic{-0.5}{
\draw(0,0)--(2,0)--(2,1)--(3,1)--(3,3)--(0,0);
\draw(1.4,0.7)node{$\Downarrow$};
}.
\end{eqnarray*}
By applying the procedure above to $\epsilon^{+}(\lambda)$, 
we obtain $\overline{w}\sim \mathrm{rev}(w)$.
\end{proof}

\begin{lemma}
\label{lemma:epplus2}
Let $w$ and $T$ be a word and a shifted 
tableau such that $P_{\mathrm{mix}}(w)=T$.
Then
\begin{enumerate}
\item If the unprimed letters from $i_1$ to $i_r$ form a vertical strip from top 
to bottom in $T$, then $i_r,\ldots,i_1$ form a decreasing sequence in $w$.
\item Let $i_j:=i_1+j-1$ for $1\le j\le r$ be unprimed letters.
If the letters $\alpha,i'_{2},\ldots,i'_{r}$ with $\alpha=i_{1}$ or $i'_{1}$ 
form a horizontal strip from left to right in $T$, then $i_r,\ldots,i_1$ 
form a decreasing sequence in $w$.
\end{enumerate}
\end{lemma}
\begin{proof}
We assume that $i_{j}$ and $i_{j+1}$ form an increasing sequence in $w$ 
for (1) and (2).
Suppose that $i_{j}$ (resp. $i_{j+1}$) are in the $k$-th (resp. $k'$-th) row
of a tableau.

For (1), since we insert the letter $i_{j}$ into a tableau before $i_{j+1}$, and 
$i_j$ and $i_{j+1}$ are unprimed, it is obvious that $k\ge k'$.
The letters $i_{j}$ and $i_{j+1}$ form a horizontal strip.
By taking contraposition, (1) follows.

For (2), if $i_{j}$ and $i_{j+1}$ are unprimed, 
we have $k\ge k'$ and $i_{j+1}$ is in a column right to $i_{j}$ from (1).
An unprimed letter becomes primed when it is placed on the main diagonal and 
bumped out by a smaller letter.
Suppose that $i_{j}$ is in the $m_{0}$-th column of the main diagonal 
and $i_{j+1}$ is in the $m_{1}$-th column with $m_{1}>m_{0}$.
If $m_{1}=m_{0}+1$, the letters $i'_{j}$ is placed above $j_{j+1}$ in the $m_{1}$-th 
column after bumping.
If $m_{1}\ge m_{0}+2$, $i'_{j}$ and $i_{j+1}$ form a horizontal strip from 
left to right.
To have $i_{j+1}$ primed, we consider bumpings of $i_{j+1}$ by smaller letters.
In the case of $m_{1}\ge m_{0}+2$, we arrive at a configuration such that 
$i'_{j}$ is above $i_{j+1}$ by bumping of $i_{j+1}$.
When $i_{j+1}$ becomes primed, $i'_{j}$ and $i'_{j+1}$ form a vertical strip 
from top to bottom.
By taking contraposition, (2) follows.
\end{proof}

\begin{proof}[Proof of Theorem \ref{thrm:epplus}]
For (1), observe that $\overline{T}$ is a standard tableau obtained from 
$\mathrm{stand}(T)$.
Let $i_{1},\ldots,i_{q}$ (resp. $i_{q+1},\ldots,i_{r}$) be unprimed letters 
in $\mathrm{stand}(T)$ corresponding to $j'$ (resp. $j$) in $T$.
In $\overline{T}$, the letters $i_{1},\ldots,i_{q+1}$  
form a vertical strip if there exists a letter $j$ below all $j'$'s.
In $\overline{T}$, the letters $i_{1},\ldots,i_{q}$ form a vertical strip if 
there is no letter $j$ below all $j'$'s.
Similarly, the letters $\alpha,i'_{q+2},\ldots,i_{r}$ with $\alpha=i_{q+1}$ or $i'_{q+1}$
form a horizontal strip in $\overline{T}$.
From Lemma \ref{lemma:epplus2}, $i_{r},\ldots,i_{1}$ form a decreasing sequence 
in $\epsilon^{-}(\overline{T})$.
Since a row in $\epsilon^{-}(\overline{T})$ is an increasing sequence, 
there exists at most one $i_{s}$, $1\le s\le r$, in a row.
The tableau $\epsilon^{+}(\overline{T})$ is obtained from $\epsilon^{-}(\overline{T})$ 
by reflecting it over the anti-diagonal. 
Thus, there exists at most one $i_{s}$ in a column of $\epsilon^{-}(\overline{T})$.
After destandardization, the tableau $\epsilon^{+}(T)$ contains at most one $j$ 
in a column. This implies that $\epsilon^{+}(T)$ is compatible with the shape 
$\epsilon^{+}(\lambda)$.

For (2), let $w$ be the reading word of $\epsilon^{-}(\mathrm{stand}(T))$ 
and $\overline{w}$ be the reading word of a tableau which is obtained by 
reflecting $\epsilon^{-}(\mathrm{stand}(T))$ over the main diagonal.
Then, from Proposition \ref{prop:epplus} and Lemma \ref{lemma:epplus1}, 
we have 
\begin{eqnarray*}
\overline{P_{\mathrm{mix}}(w)}=P_{\mathrm{mix}}(\mathrm{rev}(w))
=P_{\mathrm{mix}}(\overline{w}).
\end{eqnarray*}
Note that the reflection of $\epsilon^{-}(\overline{T})$ over the main diagonal
corresponds to performing the ``dual" operation with respect to primes on it.
Thus, we have $P_{\mathrm{mix}}(\mathrm{read}(\epsilon^{+}(T)))=T$.
\end{proof}

\begin{example}
Let $T$ be a shifted tableau in the left picture of 
Figure \ref{fig:eps1}.
\begin{figure}[ht]
\begin{eqnarray*}
T=\begin{matrix}
1 & 2' & 3 & 4' \\
  & 2 & 4 & 5 \\
  &   & 6
\end{matrix}
\qquad\qquad
\overline{T}=
\begin{matrix}
1 & 2 & 4' & 5 \\
  & 3 & 6' & 7' \\
  &   & 8
\end{matrix}
\end{eqnarray*}
\caption{A tableau $T$ and $\overline{T}$.}
\label{fig:eps1}
\end{figure}
After standardization and taking dual with respect to primes, we have 
a tableau $\overline{T}$ as the right picture of Figure \ref{fig:eps1}.
There are several words which form $\overline{T}$ by the mixed insertion,
an example of them is $w=41786352$.
We obtain a SSDT $SK(w)$ and $\epsilon^{-}(\overline{T})$ as
in Figure \ref{fig:eps2}. 
\begin{figure}[ht]
\begin{eqnarray*}
SK(w)=
\begin{matrix}
8 & 6 & 5 & 2 \\
  & 7 & 1 & 3 \\
  &   & 4
\end{matrix}
\qquad\qquad
\epsilon^{-}(\overline{T})
=
\begin{matrix}
 &   &   & 2 \\
 &   & 3 & 5  \\
 & 1 & 6 \\
4 & 7 & 8
\end{matrix}
\end{eqnarray*}
\caption{A SSDT $SK(w)$ and $\epsilon^{-}(\overline{T})$.}
\label{fig:eps2}
\end{figure} 
By reflecting $\epsilon^{-}(\overline{T})$ over the diagonal line 
and performing a destandardization,
we obtain 
\begin{eqnarray*}
\begin{matrix}
  &  &  & 4 \\
  &  & 1 & 7 \\
  & 3 & 6 & 8 \\
2 & 5 
\end{matrix}
\qquad\mapsto\qquad
\epsilon^{+}(T)=
\begin{matrix}
  &  &  & 3 \\
  &  & 1 & 5 \\
  & 2 & 4 & 6 \\
2 & 4 
\end{matrix}
\end{eqnarray*}
The mixed insertion of the reading word $\mathrm{read}(\epsilon^{+}(T))=24246153$
is $T$.
\end{example}

\section{Generalized Littlewood--Richardson coefficients}
\label{sec:GLR}
\subsection{Generalized Littlewood--Richardson rules}
We define generalized Littlewood--Richardson (LR) coefficients $b_{\alpha\lambda}^{\beta}$, 
$c_{\alpha\beta}^{\gamma}$, $e_{\lambda\mu}^{\alpha}$, 
$f_{\lambda\alpha}^{\mu}$, $g_{\alpha\beta}^{\gamma}$ and $h_{\alpha\beta}^{\lambda}$ as follows:
\begin{eqnarray*}
s_{\alpha}P_{\lambda}=\sum_{\beta}b_{\alpha\lambda}^{\beta}s_{\beta}, \\
s_{\alpha}\hat{S}_{\beta}=\sum_{\gamma}c_{\alpha\beta}^{\gamma}s_{\gamma}, \\
P_{\lambda}P_{\mu}=\sum_{\nu}e_{\lambda\mu}^{\alpha}s_{\alpha},\\
P_{\lambda}\hat{S}_{\alpha}=\sum_{\beta}f_{\lambda\alpha}^{\mu}P_{\mu}, \\
\hat{S}_{\alpha}\hat{S}_{\beta}=\sum_{\gamma}g_{\alpha\beta}^{\gamma}s_{\gamma}, \\
\hat{S}_{\alpha}\hat{S}_{\beta}=\sum_{\lambda}h_{\alpha\beta}^{\lambda}P_{\lambda}.
\end{eqnarray*}
Let $b_{\lambda}^{\alpha}:=b_{\emptyset\lambda}^{\alpha}$ and 
$e^{\alpha}_{\lambda}:=e_{\lambda\emptyset}^{\alpha}$, {\it i.e.}, 
$b_{\lambda}^{\alpha}$ and $e_{\lambda}^{\alpha}$ are 
an expansion of $P_{\lambda}$ in terms of Schur 
functions $s_{\alpha}$.  
Then we have
\begin{eqnarray*}
b_{\lambda}^{\alpha}=e^{\alpha}_{\lambda}.
\end{eqnarray*}
Similarly, we have $c_{\emptyset\alpha}^{\gamma}=g_{\alpha\emptyset}^{\gamma}$ 
and $f_{\emptyset\alpha}^{\lambda}=h_{\alpha\emptyset}^{\lambda}$.
Eqn. (\ref{ex:sinPSinP}) implies 
\begin{eqnarray*}
b_{\lambda}^{\alpha}=2^{-l(\lambda)}f_{\emptyset\alpha}^{\lambda}.
\end{eqnarray*}

\subsection{Tableau words}
We will give expressions for generalized LR coefficients in terms of tableau 
words.

Let $\alpha$, $\beta$ and $\gamma$ be an ordinary or shifted partitions satisfying 
$\beta\subseteq\alpha$.
We denote by $T(\alpha/\beta;\gamma)$ be a set of tableaux (without primed letters)
of shape $\alpha/\beta$ whose content is $\gamma$.
Similarly, we denote by $T'(\alpha/\beta;\gamma)$ be a set of shifted tableaux 
(possibly with primed letters) of shape $\alpha/\beta$ whose content is $\gamma$.
We denote by $T''(\alpha/\beta;\gamma)$ be a set of shifted tableaux 
(possibly with primed letters and diagonal elements can be primed) 
of shape $\alpha/\beta$ whose content is $\gamma$.
We define $T''(\alpha;*):=\bigcup_{\beta}T''(\alpha;\beta)$.

\begin{theorem}
\label{thrm:sP}
Let $\alpha,\beta$ be ordinary partitions and $\lambda$ be a strict partition.
We have 
\begin{eqnarray*}
b_{\alpha\lambda}^{\beta}=\#\{T\in T'(\beta/\alpha;\lambda)|\ \mathrm{read}(T)\text { is an LRS word}\}.
\end{eqnarray*}
\end{theorem}
By setting $\alpha=\emptyset$ in Theorem \ref{thrm:sP}, we have 
\begin{cor}[Stembridge \cite{Ste89}]
Let $\alpha$ be an ordinary partition and $\lambda$ be a strict partition.
\begin{eqnarray*}
b_{\lambda}^{\alpha}=\#\{T\in T'(\alpha;\lambda)|\ \mathrm{read}(T)\text { is an LRS word} \}.
\end{eqnarray*}
\end{cor}

\begin{proof}[Proof of Theorem \ref{thrm:sP}]
We have
\begin{eqnarray*}
b^{\beta}_{\alpha\lambda}
&=&\langle s_{\alpha}P_{\lambda},s_{\beta}\rangle \\
&=&\langle s_{\beta/\alpha},P_{\lambda} \rangle \\
&=&\sum_{\gamma}a^{\beta}_{\alpha\gamma}\langle s_{\gamma},P_{\lambda} \rangle \\
&=&\sum_{\gamma}a^{\beta}_{\alpha\gamma}\left[ \hat{S}_{\gamma},P_{\lambda} \right] \\
&=&\left[ \hat{S}_{\beta/\alpha},P_{\lambda}\right].
\end{eqnarray*}
The Schur function $\hat{S}_{\beta/\alpha}$ is written by a skew $Q$-function as 
$\hat{S}_{\beta/\alpha}=Q_{\beta+\delta/\alpha+\delta}=\sum_{\nu}d^{\beta+\delta}_{\alpha+\delta,\nu}Q_{\nu}$
where $\delta:=(l-1,l-2,\ldots,0)$ with $l=l(\beta)$. 
By the orthogonality $\left[Q_{\lambda},P_{\mu}\right]=\delta_{\lambda\mu}$, we obtain 
\begin{eqnarray}
\label{sP1}
b^{\beta}_{\alpha\lambda}=d^{\beta+\delta}_{\alpha+\delta,\lambda}.
\end{eqnarray}
The skew shape $\beta/\alpha$ is equivalent to $\beta'/\alpha'$ where 
$\alpha'=\alpha+\delta$ and $\beta'=\beta+\delta$.
We complete the proof by the equality (\ref{sP1}) together with Theorem \ref{thrm:StemPP}. 
\end{proof}

\begin{theorem}
\label{thrm:sS}
Let $\alpha, \beta$ and $\gamma$ be ordinary partitions.
We have 
\begin{eqnarray*}
c_{\alpha\beta}^{\gamma}=\#\{T\in T''(\gamma/\alpha;\beta): 
\mathrm{wread}(T) \text{ is a Yamanouchi word}\}.
\end{eqnarray*}
\end{theorem}

In \cite{Gas98}, Gasharov has proved the Littlewood--Richardson 
rule by using the Bender--Knuth involution \cite{BenKnu72}.
We show Theorem \ref{thrm:sS} by a shifted analogue of the proof proposed 
by Gasharov. 
We make use of the involutions introduced by Stembridge in \cite{Ste90}.

We introduce two lemmas needed later for a proof of Theorem \ref{thrm:sS}.
For a partition $\beta$ of length $l$ and a permutation $\pi\in\mathcal{S}_{l}$, we define 
\begin{eqnarray*}
\pi\ast\beta:=(\beta_{\pi(i)}-\pi(i)+i)_{1\le i\le l}.
\end{eqnarray*}
Since $\hat{S}_{\beta}$ has a determinant expression (Eqn.(\ref{Sdet1})), we have 
\begin{eqnarray}
\label{eqn:Sdet}
\hat{S}_{\beta}=\sum_{\pi\in\mathcal{S}_{l}}\mathrm{sgn}(\pi)q_{\pi\ast\beta}.
\end{eqnarray}
Then, $c^{\gamma}_{\alpha\beta}$ is rewritten as
\begin{eqnarray*}
c^{\gamma}_{\alpha\beta}&=&\langle s_{\alpha}\hat{S}_{\beta},s_{\gamma}\rangle \\
&=&\langle s_{\gamma/\alpha},\hat{S}_{\beta}\rangle \\
&=&\sum_{\pi\in\mathcal{S}_{l}}\mathrm{sgn}(\pi)\langle s_{\gamma/\alpha},q_{\pi\ast\beta}\rangle.
\end{eqnarray*}
When a strict partition $\lambda$ is a single row, {\it i.e.,} $\lambda=(n)$, we have 
$q_{(n)}=2P_{(n)}$.
From Theorem \ref{thrm:sP}, one can easily obtain a Pieri formula for a product 
of $s_{\alpha}$ and $q_{(n)}$.
The factor two comes from the fact that the first unprimed integer in 
$T\in T'(\gamma/\alpha;(n))$ is free to choose primed or unprimed.
Therefore, we have 
\begin{eqnarray}
\label{eq:sS0}
c^{\gamma}_{\alpha\beta}=\sum_{(\pi,T)\in\mathcal{L}}\mathrm{sgn}(\pi),
\end{eqnarray}
where 
\begin{eqnarray*}
\mathcal{L}:=\{(\pi,T)\ \vert\ \pi\in\mathcal{S}_{l}, T\in T''(\gamma/\alpha;\pi\ast\beta)\}.
\end{eqnarray*}
Let $\mathcal{L}_{0}$ be the subset of $\mathcal{L}$:
\begin{eqnarray*}
\mathcal{L}_{0}:=\{(\pi,T)\in\mathcal{L} \ \vert\ \mathrm{wread}(T) \text{ is a Yamanouchi word} \}.
\end{eqnarray*}

\begin{lemma}
\label{lemma:sS0}
If $(\pi,T)\in\mathcal{L}_{0}$, $\pi=\mathrm{Id}$.
\end{lemma}
\begin{proof}
Since $(\pi,T)\in\mathcal{L}_{0}$, the number of $i$ and $i'$ is equal to or 
greater than that of $i+1$ and $(i+1)'$. 
This implies that $(\pi\ast\beta)_{i}\ge(\pi\ast\beta)_{i+1}$, 
{\it i.e.,} $\beta_{\pi(i)}-\pi(i)+i\ge\beta_{\pi(i+1)}-\pi(i+1)+i+1$.
On the other hand, we have $\beta_{1}-1>\beta_{2}-2>\cdots>\beta_{l}-l$.
Thus we have $\pi(i)<\pi(i+1)$ for $1\le i\le l$. 
This property can be satisfied by only $\pi=\mathrm{Id}$.
\end{proof}

\begin{lemma}
\label{lemma:sS1}
There exists a bijection 
$\phi:\mathcal{L}-\mathcal{L}_{0}\ni(\pi,T)\mapsto(\pi',T')\in\mathcal{L}-\mathcal{L}_{0}$
with $\mathrm{sgn}(\pi)=-\mathrm{sgn}(\pi')$.
\end{lemma}

For a construction of a bijection $\phi$, we make use of the involution on
tableaux developed by Bender and Knuth for ordinary partitions \cite{BenKnu72} and 
by Stembridge for shifted partitions \cite{Ste90}. 
We briefly review these involutions before we move to a proof of Lemma \ref{lemma:sS1}.

A skew diagram $\lambda/\mu$ is said to be {\it detached} if 
$\lambda/\mu$ has at most one box on the main diagonal.

Let $R_{1,2}(\lambda/\mu)$ be the set of shifted tableaux whose contents
are $1$ and $2$.
Then, we have 
\begin{lemma}[Bender--Knuth \cite{BenKnu72}]
Suppose that $\lambda/\mu$ is detached. 
There exists a content-reversing involution $\omega_{1,2}$ on $R_{1,2}(\lambda/\mu)$.
\end{lemma}
\begin{proof}
Since contents of $\lambda/\mu$ are $1$ and $2$, there are at most two boxes in a column.
Further, if there are two boxes in a column, the contents are $1$ and $2$ from top to bottom.
We say that the content $1$ (resp. $2$) is free if there is no $2$ (resp. $1$) in the same 
column.
Suppose that a $1$ is not free in $\lambda/\mu$ and denote by $p_1$ the position of this $1$. 
Since $\lambda/\mu$ is detached, a $1$ left to $p_1$ in the same row is not free.
Similarly, suppose that $2$ is not free in $\lambda/\mu$ and denote by $p_2$ the position of
this $2$. Then, a $2$ right to $p_2$ in the same row is not free.
In a row of $\lambda/\mu$, we have non-free $1$'s, $r_{1}$ free $1$'s, $r_{2}$ free $2$'s and 
non-free $2$'s from left to right.
The involution $\omega_{1,2}$ exchange $r_{1}$ and $r_{2}$, {\it i.e.,}
the row of $\omega_{1,2}(\lambda/\mu)$ has non-free $1$'s, $r_{2}$ free $1$'s, $r_{1}$ 
free $2$'s and non-free $2$'s from left to right.
By construction, $\omega_{1,2}$ is content-reversing.
\end{proof}

Similarly, let $C_{1',2'}(\lambda/\mu)$ be the set of shifted tableaux whose 
contents are $1'$ and $2'$.
Then, we have 
\begin{cor}
Suppose that $\lambda/\mu$ is detached. 
There exists a content-reversing involution $\omega_{1',2'}$ on $C_{1,2}(\lambda/\mu)$.
\end{cor}
The involution $\omega_{1',2'}$ is obtained by transposing $\omega_{1,2}$, {\it i.e.,}
we replace a row in the definition of $\omega_{1,2}$ by a column.

Let $RC(\lambda/\mu)$ be the set of tableaux with contents $1$ and $2'$.
Similarly, $CR(\lambda/\mu)$ be the set of tableaux with contents $2'$ and 
$1$, but with a nonstandard ordering $2'<1$.

\begin{lemma}[Stembridge  \cite{Ste90}]
Suppose that $\lambda/\mu$ is detached.
There exits a content-preserving bijection $\psi:RC(\lambda/\mu)\rightarrow CR(\lambda/\mu)$.
\end{lemma}  
\begin{proof}
A skew shape $\lambda/\mu$ is a union of connected components. 
The action of $\psi$ is defined on each connected component. 
We assume that $\lambda/\mu$ is connected, which implies that 
$\lambda/\mu$ is a strip.

Let $T$ be a tableau in $RC(\lambda/\mu)$ and $\mathrm{read}(T):=a_{n}\cdots a_{1}$.
Since $T$ is a semistandard tableau, we have $a_{i}=1$ if $a_{i}$ and $a_{i+1}$ are 
in the same row, and $a_{i}=2'$ if $a_{i}$ and $a_{i+1}$ are in the same column 
except for $a_{1}$. The content $a_{1}$ is either $1$ or $2'$.
We define $\psi$ as a rotation of $\mathrm{read}(T)$:
$\psi(a_{i-1})=a_{i}$ for $2\le i\le n$ and $\psi(a_{n})=a_{1}$.
By construction, $\psi(T)\in CR(\lambda/\mu)$ and $\psi$ is invertible.
Therefore, $\psi$ is content-preserving and bijective.
\end{proof}

We construct a bijection $\phi$ in Lemma \ref{lemma:sS1} as follows. 
Given $(\pi,T)\in\mathcal{L}-\mathcal{L}_{0}$, we 
define $\mathrm{wread}(T)=w_{N}\cdots w_{1}$.
From the definition of reading and weak reading words,
there exists an integer $N'$ such that $w_{i}$ for $1\le i\le N'$
(resp. $w_{i}$ for $i>N'$) corresponds to an unprimed (resp. primed) 
element in $\mathrm{read}(T)$.
Let $r:=r(T)$ be an integer such that 
\begin{eqnarray*}
r(T)=\min\{k\ \vert\ w_{k}\cdots w_{1} \text{ is not a Yamanouchi word}\}.
\end{eqnarray*}
We have two cases: $r\le N'$ or $r>N'$.
First, we consider $r\le N'$. 
The element $w_{r}$ is unprimed in $\mathrm{wread}(T)$.
For a tableau word $w$, We denote by $m_{i}(w)$ the number of $i$ in $w$.
By definition of $r$, the word $w_{r-1}\cdots w_{1}$ is a Yamanouchi word 
and the word $w_{r}\cdots w_{1}$ is not.
Therefore, $w_{r}\ge2'$ and one can set $w_{r}=i+1$ or $(i+1)'$.
More precisely, we have $w_{r}=i+1$ if $r\le N'$ and 
$w_{r}=(i+1)'$ if $r>N'$.
We have 
\begin{eqnarray*}
m_{1}(w_{r-1}\cdots w_{1})\ge\cdots\ge m_{r}(w_{r-1}\cdots w_{1})
\end{eqnarray*}
and 
\begin{eqnarray}
\label{eq:sS1}
m_{i+1}(w_{r}\cdots w_{1})=m_{i}(w_{r}\cdots w_{1})+1.
\end{eqnarray}

Recall that $\psi_{i,(i+1)'}$ is a content-preserving bijection and 
we have a non-standard ordering $(i+1)'<i$.
Let $S$ be a strip formed by $\{i,(i+1)'\}$ including the element 
$w_{r}$ if $w_{r}=(i+1)'$. 
We denote by $s_{1}\in S$ the most upper-right element in $S$.
Let $\widetilde{T}:=\psi(T)$, $\mathrm{wread}(\widetilde{T}):=v_{N}\cdots v_{1}$ 
and $r':=r(\widetilde{T})$. 
Note that $\psi$ is also shape-preserving.
\begin{lemma}
\label{lemma:sS2}
We have: 
\begin{enumerate}
\item If $r\le N'$, then the position of $v_{r'}$ in $\psi(T)$ is the same
as the one of $w_{r}$ in $T$.
\item For $r>N'$, we have the following:
\begin{enumerate}
\item if $s_{1}=i$, then the position of $v_{r'}$ in $\psi(T)$ is one step 
upper than $w_{r}$.
\item if $s_{1}=(i+1)'$ and $w_{r}$ is not in the most upper-right box in $S$, 
then the position of $v_{r'}$ in $\psi(T)$ is the same as the one of $w_{r}$ in $T$.
\item if $s_{1}=(i+1)'$ and $w_{r}=s_{1}$, then 
the position of $v_{r'}$ in $\psi(T)$ is the leftmost box in $S$ which is 
in the same row as $w_{r}$.
\end{enumerate}
\end{enumerate}
Further, we have 
\begin{eqnarray*}
m_{i}(w_{r}\dots w_{1})=m_{i}(v_{r'}\cdots v_{1})
\end{eqnarray*}
for $1\le i\le r$.
\end{lemma}

\begin{proof}
For simplicity, we make use of the identification 
$i\leftrightarrow1, i+1\leftrightarrow2$.
We denote by $b_{2}$ the box corresponding to $w_{r}$ in $T$.

For (1), we denote by $b_{1}$ the box just above $b_{2}$ in $T$.
Suppose that the content of $b_{1}$ (resp. $b_{2}$) is $1$ (resp. $2$) 
and there are $k$ $2$'s right to the box $b_{2}$ in $T$.
Denote by $b_{3}$ the box $k$ steps right to the box $b_{1}$ and 
by $w_{s}$ the content of $b_{3}$.
Since the order of letters is $1<2'<2$, we have $k$ $1$'s 
right to the box $b_{1}$ or $k-1$ $1$'s and one $2'$ right to 
the box $b_{1}$, namely $w_{s}=1$ or $w_{s}=2'$.
In the former case, since $m_{2}(w_{r}\ldots w_{1})=m_{1}(w_{r}\ldots w_{1})+1$
and $m_{2}(w_{r}\ldots w_{s})\le m_{1}(w_{r}\ldots w_{s})$, 
there exists $s'<r$ satisfying $m_{2}(w_{s'}\ldots w_{1})=m_{1}(w_{s'}\ldots w_{1})+1$.
This contradicts the minimality of $r$. 
Thus, the content $w_s$ of the box $b_{3}$ is $2'$ if $b_1$ is $1$.
Further, if there is a box with content $1$ left to the box $b_{1}$,
this also contradicts the minimality of $r$.
The above observations can be extended to the case where 
the letter $2$'s form a horizontal strip and the letters $1$ and $2'$
form a strip $S$ in $T$.
\begin{figure}
\begin{eqnarray*}
\begin{matrix}
 & & & 1 & 1 & 1 & 2' \\
 & & & 2' & 2 & 2 & 2\\
 & & & 2' \\
1 & 1 & 1 & 2' \\
\boxed{2} & 2 & 2 & 2
\end{matrix}\quad
\mapsto\quad
\begin{matrix}
 & & & 2' & 1 & 1 & 1 \\
 & & & 2' & 2 & 2 & 2\\
 & & & 2' \\
2' & 1 & 1 & 1 \\
\boxed{2} & 2 & 2 & 2
\end{matrix}
\end{eqnarray*}
\caption{An example of the action of $\psi$ on a tableau. 
The boxed letter $2$'s correspond to $w_{r}$ and $v_{r'}$.}
\label{fig:sS1}
\end{figure}
The strip $S$ has to be above the box $b_{2}$. 

In general, we consider a configuration of letters $1,2'$ 
and $2$ such that $1$'s and $2'$'s form a strip $S$ and $2$'s form a 
horizontal strip below $S$.
Suppose that the $i$-th row from top contains $n_{1}$ $1$'s, the $(i+1)$-th
row contains $n_{2}$ $2$'s and the $i$-th row is not placed at the bottom of $S'$.
We have $n_{1}\ge n_{2}$ if the $i$-th row contains $2'$ and $n_{1}\ge n_{2}+1$ 
if the $i$-th row does not contain $2'$.
After applying $\psi$ on $S$, the number of $1$'s in the $i$-th row is 
greater than or equal to the one of $2$'s in the $(i+1)$-th row.
The total numbers of $1$'s in $S$ and $\psi(S)$ are the same.
This means that if $w_{p}\ldots w_{1}$ is Yamanouchi, then $v_{p}\ldots v_{1}$
is also Yamanouchi for all $p<r$.
Note that in the case where the content of $b_{1}$ is not $1$, the letter $2$ in $b_{2}$ 
is irrelevant to the bijection $\psi$.
Since the bijection $\psi$ is defined as a rotation of the strip $S$ and 
content-preserving, the position of $v_{r'}$ is the same as the one of $w_{r}$.  
See Figure \ref{fig:sS1} for an example.
Therefore, the statement (1) is true.

For (2), we have $w_{r}=2'$ and denote by $b_{1}$ the box one step left to 
the box $b_{2}$ in $T$.
Suppose that the content of $b_{1}$ is $1'$. For primed integers, we read the 
content of $T$ from left to right in a row. 
Primed integers appear at most once in a row of $T$. 
Therefore, we have that $m_{1}(w_{r}w_{r-1})=m_{2}(w_{r}w_{r-1})$.
This is equivalent to $m_{2}(w_{r-2}\ldots w_{1})=m_{1}(w_{r-2}\ldots w_{1})+1$,
which is a contradiction of the minimality of $r$.
Thus, the content of $b_{1}$ is not $1'$.
\begin{figure}[ht]
\begin{eqnarray*}
\begin{matrix}
 & 1 & 1 \\
 & \boxed{2'} \\ 
 & 2' \\
 & 2' \\
1 & 2' 
\end{matrix}\quad\mapsto\quad
\begin{matrix}
 & \boxed{2'} & 1 \\
 & 2' \\ 
 & 2' \\
 & 2' \\
1 & 1 
\end{matrix}
\quad,\quad
\begin{matrix}
 & & 1' \\
 & & 2'  \\
 & & \boxed{2'} \\
 & & 2' \\
1 & 1 & 2'
\end{matrix},\quad\mapsto\quad
\begin{matrix}
 & & 1' \\
 & & 2'  \\
 & & \boxed{2'} \\
 & & 2' \\
2' & 1 & 1
\end{matrix}
\end{eqnarray*}
\begin{eqnarray*}
\begin{matrix}
 & 1 & 1 & \boxed{2'} \\
 & 2' \\
 & 2' \\
1 & 2' \\
\end{matrix}\quad\mapsto\quad
\begin{matrix}
 & \boxed{2'} & 1 & 1 \\
 & 2' \\
 & 2' \\
2' & 1 \\
\end{matrix}
\end{eqnarray*}
\caption{Examples of an action of $\psi$. The boxed $2'$ corresponds to
$w_{r}$ and $v_{r'}$.}
\label{fig:sS2}
\end{figure}
If we apply $\psi$ on a strip with $s_{1}=1$, a letter $2'$ is moved 
upward by one step.
A condition $s_{1}=1$ implies that this letter $1$ appears in $w$ 
as $w_{j}$ with some $j\le N'$ before and after the application of $\psi$.
The statement (2a) directly follows from this observation.
If $s_{1}=2'$ and we apply $\psi$ on the strip $S$, 
the lowest row of a new strip contains $2'$ but the top row does not.
Combining the observation above with the fact that a letter $2'$ is 
moved upward by $\psi$,
the statements (2b) and (2c) holds true.
Figure \ref{fig:sS2} is examples of the action of $\psi$ on $S$.

By summarizing the above observations, 
it is obvious that $m_{i}(w_{r}\ldots w_{1})=m_{i}(v_{r'}\ldots v_{1})$
for $1\le i\le r$.
\end{proof}

Let $T'$ be a tableau $T':=\psi(T)$ and denote by $T'(p,q)$ the $q$-th 
element in the $p$-th row in $T'$. 
We say that the element of $T'$, $T'(p,q)=i+1$ (resp. $i$), is free 
when there exists no $i$ (resp. $i+1$) in the same column of $T'(p,q)$. 
Similarly, we say that the element $i'$ (resp. $(i+1)'$) is free 
when there exists no $(i+1)'$ (resp. $i'$) in the same row.

We claim: 
\begin{itemize}
\item[(S1)]
Suppose that $w_{r}=i+1$.
The integer $i+1$ in $T'$, which is left to $w_{r}$ and 
in the same row as $w_r$, is free.
\item[(S2)]
Suppose that $w_{r}=(i+1)'$.
The alphabet $(i+1)'$ in $T'$, which is lower than $w_{r}$ and 
in the same column as $w_{r}$, is free. 
\end{itemize}
Let $T'(p,q)$ be the element $w_r$ in $T'$.
Suppose that $T'(p,q')=i+1$ with $q'<q$ and $T'(p,q')$ is not free.
Then, we have $T'(p-1,q')=i$.
Since the ordering of alphabet in $T'$ is $i'<(i+1)'<i<i+1$ and $T'$ 
is a semistandard tableau, the integer $i+1$, which is right to $w_r$ and 
in the same row as $w_r$, is not free.
We denote by $T'(p,q^{\prime\prime})$ the rightmost $i+1$, which is right to 
$w_r$ and in the same row. 
We have $T'(p-1,q^{\prime\prime})=i$ and denote by $w_s:=i$ the content 
of $T'(p-1,q^{\prime\prime})$.
The above consideration also implies that the elements $T(p-1,m)=i$ 
with $q'\le m<q^{\prime\prime}$.
We have 
\begin{eqnarray*}
m_{i}(w_r\cdots w_s)\ge m_{i+1}(w_r\cdots w_s).
\end{eqnarray*}
Together with Eqn.(\ref{eq:sS1}), we obtain 
\begin{eqnarray*}
m_{i}(w_{s-1}\cdots w_{1})<m_{i+1}(w_{s-1}\cdots w_{1}),
\end{eqnarray*}
which is a contradiction against the minimality of $r$.
Thus, the statement (S1) holds true.
By transposing the above argument, the statement (S2) follows.

\begin{proof}[Proof of Lemma \ref{lemma:sS1}]
In the above notation, we construct a bijection 
$\phi(\pi,T)=(\pi',T')$ as follows.
First, we define $\pi':=\pi\circ(i,i+1)$ where $(i,i+1)$ is 
a transposition in $\mathcal{S}_l$.
We perform two operations on $\widetilde{T}:=\psi_{i,(i+1)'}(T)$.
We consider two cases: 1) $v_{r'}=i+1$ and 2) $v_{r'}=(i+1)'$. 

\paragraph{Case 1}
We perform $\omega_{i,i+1}$ on $i$'s and $i+1$'s except for $v_{r'}$ 
which are in the same row as $v_{r'}$.
We also perform $\omega_{i,i+1}$ on $i$'s and $i+1$'s 
which are in the lower rows than $v_{r'}$.
We call this operation $\Omega_{i,i+1}$.
Successively, we perform $\omega_{i',(i+1)'}$ on all $i'$'s and $(i+1)'$'s
in $\Omega_{i,i+1}(\widetilde{T})$.
Then, we obtain 
$T'=\psi^{-1}\circ \omega_{i',(i+1)'}\circ\Omega_{i,i+1}\circ \psi(T)$.

\paragraph{Case 2}
Since $v_{r'}=(i+1)'$, we do not perform any operation with respect to 
$i$ and $i+1$.
We perform $\omega_{i',(i+1)'}$ on $i'$'s and $(i+1)'$'s which are strictly 
upper than $v_{r'}$.
We call this operation $\Omega'_{i',(i+1)'}$.
A tableau $T'$ is given by 
$T':=\psi^{-1}\circ\Omega'_{i',(i+1)'}\circ\psi(T)$.

Let $\tilde{r}:=r(T')$ and $\mathrm{wread}(T')=w'_{N}\cdots w'_{1}$. 
We show that the weights of $T$ and $T'$ are related as 
$\mathrm{wt}(T')=\pi\ast\mathrm{wt}(T)$.
From the construction of $T'$ and Lemma \ref{lemma:sS2}, 
we have $\tilde{r}=r$, $w_{r}\cdots w_{1}=w'_{r}\cdots w'_{1}$ and 
\begin{eqnarray*}
m_{i}(w_{N}\cdots w_{r+1})&=&m_{i+1}(w'_{N}\cdots w'_{r+1}), \\
m_{i+1}(w_{N}\cdots w_{r+1})&=&m_{i}(w'_{N}\cdots w'_{r+1}).
\end{eqnarray*}
Therefore, we have 
\begin{eqnarray*}
m_{i}(\mathrm{wread}(T'))&=&m_{i}(\mathrm{wread}(T))-1, \\
&=&\beta_{\pi(i+1)}-\pi(i+1)+i, \\
m_{i+1}(\mathrm{wread}(T'))&=&m_{i}(\mathrm{wread}(T))+1, \\
&=&\beta_{\pi(i)}-\pi(i)+i+1, 
\end{eqnarray*}
which implies $\mathrm{wt}(T')=\pi\ast\mathrm{wt}(T)$.
Further, $(\pi',T')\in \mathcal{L}-\mathcal{L}_{0}$ since 
$w'_{r}\cdots w'_{1}$ is not a Yamanouchi word.

We have $\phi(\pi',T')=(\pi,T)$ by construction, which implies that 
$\phi$ is a bijection.

\end{proof}

\begin{proof}[Proof of Theorem \ref{thrm:sS}]
From Eqn.(\ref{eq:sS0}), we have 
\begin{eqnarray}
\label{eq:sS2}
c^{\gamma}_{\alpha\beta}
=\sum_{(\pi,T)\in\mathcal{L}_{0}}\mathrm{sgn}(\pi)
+\sum_{(\pi,T)\in\mathcal{L}-\mathcal{L}_{0}}\mathrm{sgn}(\pi).
\end{eqnarray}
Form Lemma \ref{lemma:sS0}, the first term of Eqn.(\ref{eq:sS2})
becomes the total number of $T\in T''(\gamma/\alpha;\beta)$ whose 
weak reading word is a Yamanouchi word.
The second term of Eqn.(\ref{eq:sS2}) is zero from Lemma \ref{lemma:sS1}.
Therefore, we complete the proof.
\end{proof}

\begin{theorem}
\label{thrm:PPs}
We have
\begin{eqnarray*}
e_{\lambda\mu}^{\alpha}=\sum_{\beta}\#\{(T,U)\in T(\beta;\lambda)\times T(\alpha/\beta;\mu): 
\mathrm{read}(T) \text{ and }\mathrm{read}(U)\text { are LRS words}\}.
\end{eqnarray*}
\end{theorem}
\begin{proof}
The theorem directly follows from Theorems \ref{thrm:sP} and Theorem \ref{thrm:StemPP}.
\end{proof}

\begin{theorem}
\label{thrm:PPs2}
Let $T\in T'(\lambda;\ast)$ and $U\in T'(\mu;\ast)$ be shifted tableaux in $X'$ of 
shape $\lambda$ and $\mu$.
We denote by $w_{\lambda}$ and $w_{\mu}$ the weak reading words of $T$ and $U$, respectively.
Then, we have 
\begin{eqnarray*}
e_{\lambda\mu}^{\alpha}
=\#\{ (w_\lambda,w_\mu): w_{\mu}\ast w_{\lambda} \text{ is a Yamanouchi word of content }\alpha\}.
\end{eqnarray*}
\end{theorem}

Before we move to a proof of Theorem \ref{thrm:PPs2}, we introduce 
lemmas needed later.

Let $\alpha,\beta$ and $\gamma$ be ordinary partitions and $\lambda$ be 
a shifted partition.
From Theorems \ref{thrm:sP} and \ref{thrm:StemPP}, 
we have 
\begin{eqnarray}
\label{eq:PPs21}
b_{\alpha\lambda}^{\beta}
=
\sum_{\gamma}b_{\emptyset\lambda}^{\gamma}\cdot
a_{\gamma\alpha}^{\beta},
\end{eqnarray}
where $\gamma$ is an ordinary partition satisfying 
$\gamma\subseteq\beta$ and $|\gamma|=|\lambda|$.
Let $T$ be a tableau of shape $\beta$ such that 
the reading word inside $\gamma\subseteq\beta$ is 
an LRS word of content $\lambda$ and the reading
word for the shape $\beta/\gamma$ is a Yamanouchi 
word of content $\alpha$.
Note that contents inside $\gamma$ can be primed
and contents inside $\beta/\gamma$ do not have 
primes.
We denote by $\mathrm{Tab}(\alpha,\lambda;\beta)$ the set of 
tableaux $T$ with properties as above.

We have a unique semistandard tableau $T_{\alpha}$ whose shape and weight
are both $\alpha$. 
Then, the reading word $w_{\alpha}:=\mathrm{read}(T_{\alpha})$ is a 
Yamanouchi word, {\it i.e.}, $w_{\alpha}$ consists of $\alpha_{l}$ $l$'s, 
$\alpha_{l-1}$ $l-1$'s,\ldots, and $\alpha_{1}$ $1$'s.
We denote by $\mathrm{Tab}'(\alpha,\lambda;\beta)$ the set of 
tableaux $U$ of shape $\lambda$ such that the concatenation of the weak reading
word of $U$ and $w_{\alpha}$ is a Yamanouchi word of content $\beta$, 
{\it i.e.}, $\mathrm{wread}(U)\ast w_{\alpha}$ is Yamanouchi of content $\beta$.

\begin{lemma}
\label{lemma:PPs1}
There exists a bijection 
$\chi:\mathrm{Tab}(\alpha,\lambda;\beta)\rightarrow\mathrm{Tab}'(\alpha,\lambda;\beta)$.
\end{lemma}
\begin{proof}
We construct a map 
$\chi:\mathrm{Tab}(\alpha,\lambda;\beta)\rightarrow\mathrm{Tab}'(\alpha,\lambda;\beta)$
in the following two steps.

\paragraph{Step 1}
We perform a standardization on $T\in\mathrm{Tab}(\alpha,\lambda;\beta)$ 
as follows.
Note that the region $\gamma\subseteq\beta$ is formed by contents whose reading 
word is an LRS word.
We enumerate boxes in $\gamma\subseteq\beta$ by $1,2\ldots,|\lambda|$ 
(see Step 1 for the construction of a bijection from $T$ to $\epsilon^{+}(T)$
in Section \ref{sec:SSIDTplus}).
Successively, we enumerate boxes in $\beta/\gamma$ by 
$|\lambda|+1,\ldots,|\beta|$ according to the same rule 
as in the region $\gamma$.
We denote by $Q_{\mathrm{stand}}(\beta)$ the obtained standard tableau
by the above operation.
We denote by $P_{0}(\beta)$ a unique semistandard tableau of shape $\beta$, 
weight $\beta$ and without primes.
By reversing the RSK algorithm, we obtain a word
$w:=\mathrm{RSK}^{-1}(P_{0}(\beta),Q_{\mathrm{stand}}(\beta))$.

\paragraph{Step 2}
We split the word $w$ into two words $w=w_{0}\ast w_{1}$ such that 
the length $w_{0}$ (resp. $w_{1}$) is $|\lambda|$ (resp. $|\alpha|$).
Then, we define a tableau $U$ of shape $\lambda$ by the mixed 
insertion $U:=P_{\mathrm{mix}}(w_{0})$.

We claim:
\begin{enumerate}[(S1)]
\setcounter{enumi}{2}
\item The word $w$ is a Yamanouchi word.
\item A word $w_{1}$ satisfies $P_{\mathrm{RSK}}(w_{1})=T_{\alpha}$.
\item A tableau $U=P_{\mathrm{mix}}(w_{0})$ is of shape $\lambda$ and 
a concatenation $\mathrm{wread}(U)\ast w_{\alpha}$ is a Yamanouchi word.
\end{enumerate}

Reversing the RSK insertion means that we insert an element $x$ into 
the upper row and bump out the largest element which is 
less than $x$.
From the definition of $P_{0}(\beta)$, the contents of the $i$-th row 
of $P_{0}(\beta)$ are all integer $i$.
Therefore, if we bump out the $m$-th $x$, $1\le x\le l(\beta)$, as an output, 
we have already bumped out at least $m$ $y$'s for all $1\le y\le x-1$.
This implies that the word $w$ is Yamanouchi, that is, 
the statement (S3) holds true.

Given a partition $\alpha$ of length $l$, we consider a sequence of partitions 
$\emptyset=\hat{\alpha}_{0}\subset\hat{\alpha}_{1}\subset\hat{\alpha}_{2}
\ldots\subset\hat{\alpha}_{l}=\alpha$ 
given by $\hat{\alpha}_{i}:=(\alpha_{l-i+1},\alpha_{l-i+2},\ldots,\alpha_{l})$.
Note that the skew shape $\hat{\alpha}_{i}/\hat{\alpha}_{i-1}$ is a horizontal 
strip of length $\alpha_{i}$.
We enumerate boxes in $\alpha$ by $|\alpha|,|\alpha|-1,\ldots,1$ 
in the following way.
For $1\le i\le l$, we put integers from $|\alpha|-\sum_{j=1}^{i-1}\alpha_{i}$ 
to $|\alpha|-\sum_{j=1}^{i}\alpha_{i}+1$ on the region 
$\hat{\alpha}_{i}/\hat{\alpha}_{i-1}$ from left to right.
We denote by $T_{\mathrm{rev}}$ a tableau of shape $\alpha$ constructed above. 
The tableau $T_{\mathrm{rev}}$ is standard with respect to the reversed order 
of letters.
A word $u:=u_{1}u_{2}\ldots u_{|\alpha|}$ is obtained from $T_{\mathrm{rev}}$ by 
setting $u_{i}=j$ where $j$ is the number of row of $i$ in $T_{\mathrm{rev}}$.
For example, when $\alpha=(4,2,1)$,
$T_{\mathrm{rev}}$ is given by 
\begin{eqnarray*}
T_{\mathrm{rev}}=
\ 
\begin{matrix}
7 & 5 & 2 & 1 \\
6 & 3 \\
4 \\
\end{matrix},
\end{eqnarray*}
and $u=1123121$.
By construction, the word $u$ is Yamanouchi and $P_{\mathrm{RSK}}(u)=T_{\alpha}$.
Further, the recording tableau $Q_{\mathrm{RSK}}(u)$ is a standard tableau 
numbered from $1$ to $|\alpha|$ starting with the first row and from left to right
in a row.
It is enough to show that $w_{1}=u$.
Since the reading word for the shape $\beta/\gamma$ is Yamanouchi and the 
contents of the $i$-th row in $P_{0}(\beta)$ are $i$'s, 
it is obvious that we bump out the word $u$.
This implies that the statement (S4) is true.

Let $(P_{\mathrm{RSK}}(w_{0}),Q_{\mathrm{RSK}}(w_{0}))$ be a pair of 
a semistandard tableau of content $\beta/\alpha$ and a standard tableau of 
shape $\gamma$ for the word $w_0$.
Let $\gamma'$ and $\gamma''$ be a partition satisfying 
$\gamma',\gamma''\subseteq\gamma$.
Suppose that the shape $\gamma/\gamma'$ is a horizontal strip and we enumerate 
the boxes in $\gamma/\gamma'$ by $1,2,\ldots,l$ from left to right, where 
$l$ is the number of boxes in the strip.
Starting from the box with $l$, we perform the reversed RSK insertion 
on $P_{\mathrm{RSK}}(w_0)$.
Then, it is obvious that we obtain a word which is weakly increasing from
left to right.
Similarly, suppose that the shape $\gamma/\gamma^{\prime\prime}$ is 
a vertical strip and we enumerate the boxes in $\gamma/\gamma^{\prime\prime}$ 
by $1,2,\ldots,l$ from top to bottom.
By reversing the RSK insertion, we obtain a word which is strictly decreasing.
Therefore, if we write $w_{0}=\tilde{w}_{1}\tilde{w}_{2}\ldots\tilde{w}_{l(\lambda)}$ where 
the word $\tilde{w}_{i}$ is of length $\lambda_{i}$, 
the word $\tilde{w}_{i}$, $1\le i\le l(\lambda)$, is a hook word.
The word $\tilde{w}_{i}$ is written as a concatenation of two words, 
$\tilde{w}_{i}=(\tilde{w}_{i}\downarrow)\ast (\tilde{w}_{i}\uparrow)$.
Recall that a tableau $T$ in $\mathrm{Tab}(\alpha,\lambda;\beta)$ satisfies the lattice 
property, {\it i.e.}, $\mathrm{read}(T)$ inside $\gamma$ is an LRS word.
The length of the word $(\tilde{w}_{i}\downarrow)$ (resp. $(\tilde{w}_{i}\uparrow)$) 
is the number of $i'$'s (resp. $i$'s) in $T$ plus (resp. minus) one.
The lattice property also implies the following: 
1) the rightmost $i$ in $T$ is in the same column of or left to the rightmost $i-1$, 
2) the rightmost $i$ in $T$ is placed in a lower row than that of the 
rightmost $i-1$, and 
3) the number of $i$ in $T$ is greater than or equal to that of $i+1$ in $T$.
By a similar argument to the case of a vertical strip, 1) and 2) imply that the last 
element $\tilde{w}_{i,\lambda_{i}}$ of $\tilde{w}_{i}$ is strictly decreasing with 
respect to $1\le i\le l(\lambda)$.
Also, 3) implies that the length of $(\tilde{w}_{i}\uparrow)$ is weakly decreasing
with respect to $1\le i\le l(\lambda)$.
Denote by $\tilde{l}_{i}$ the length of $(\tilde{w}_{i}\uparrow)$ for $1\le i\le l(\lambda)$.
We delete the rightmost $i$'s for $1\le i$ and obtain a new tableau $T'$.
We apply the same argument as above to $T'$.
Then, the $(\tilde{l}_{i}-j)$-th element of $(\tilde{w}_{i}\uparrow)$ is greater than 
$(\tilde{l}_{i+1}-j)$-th element of $(\tilde{w}_{i+1}\uparrow)$, {\it i.e.}, 
$(\tilde{w}_{i}\uparrow)_{\tilde{l}_{i}-j}>(\tilde{w}_{i+1}\uparrow)_{\tilde{l}_{i+1}-j}$
for $1\le j\le \tilde{l}_{i+1}$.
Let $j$ be the rightmost $(i+1)'$ or the leftmost $(i+1)$ if there is no $(i+1)'$ in $T$.
Since the rightmost $i$ and $j$ form a vertical strip of length two,  
$\tilde{w}_{i,\lambda_{i}}$ is greater than $\tilde{w}_{i+1,1}$.

We have a stronger constraint on the words $\tilde{w}_{i}$ and $\tilde{w}_{i+1}$.
Let $p$ be the position of the letter $(i+1)'$ in a tableau word $\mathrm{read}(T)$. 
Then, the lattice property for an LRS word implies the following: 
4) in the word $\mathrm{read}(T)$, the number of $i$ right to $p$ is strictly 
greater than the one of $i+1$ right to $p$.
Let $m_{j}$ be the number of elements in $(\tilde{w}_{i+1}\uparrow)$ which is 
equal to or greater than $(\tilde{w}_{i+1}\downarrow)_{j}$
for $1\le j\le l((\tilde{w}_{i+1}\downarrow))$.
The constraint 4) is rephrased:
the number of elements in $(\tilde{w}_{i}\uparrow)$ which is strictly greater than 
$(\tilde{w}_{i+1}\downarrow)_{j}$ is equal to or greater than $m_{j}$ plus one.
From these observations, we have 
that the shape of $P_{\mathrm{mix}}(w_{0})$ is $\lambda$.
To show that the statement (S5) is true, it is enough to 
show the following statement:
\begin{enumerate}
\item[(S$5'$)]
The concatenation of two words 
$\mathrm{wread}(P_{\mathrm{mix}}(x_{i}))$ and $y_{i}$ for 
$1\le i\le l(\lambda)$ is Yamanouchi 
where $x_{i}:=\tilde{w}_{1}\ldots\tilde{w}_{i}$ and 
$y_{i}:=\tilde{w}_{i+1}\ldots\tilde{w}_{l(\lambda)}w_{\alpha}$.
\end{enumerate}
We prove the above statement by induction.
When $i=1$, since the word $\tilde{w}_{1}$ is a hook word,
we have $\mathrm{wread}(P_{\mathrm{mix}}(\tilde{w}_{1}))=\tilde{w}_{1}$.
From (S3), we have 
$\mathrm{wread}(P_{\mathrm{mix}}(\tilde{w}_{1}))\ast y_{1}=\tilde{w}_{1}\ast y_{1}=w$
is Yamanouchi.
Suppose that the statement (S$5'$) is true for some $i$.
Note that the shape of $P_{\mathrm{mix}}(x_{i})$ is a shifted tableau 
$(\lambda_1,\ldots,\lambda_{i})$.
We insert the word $\tilde{w}_{i+1}$ into $P_{\mathrm{mix}}(x_{i})$ by the mixed 
insertion.
When we insert $x$ into a shifted tableau, we first bump out an smallest element $y$ which 
is strictly greater than $x$. Then, we insert $y$ into the next row (resp. column) if $y$ is
unprimed (resp. primed).
Therefore, in the weak reading word of a new tableau, 
the bumped element $y$ (or $y'$) appears left to the element $x$.
Suppose that a word $v$ is Yamanouchi and $v:=v_{2}\ast y\ast v_{1}\ast x\ast v_{0}$.
Note that $v_{0}$ is also Yamanouchi.
Then, we have that the concatenation of two words $\mathrm{wread}(P_{\mathrm{mix}}(v_{2}\ast y\ast v_{1}\ast x))$ 
and $\mathrm{read}(P_{\mathrm{RSK}}(v_{0}))$ stays Yamanouchi.
This implies that (S$5'$) is true.

The inverse of $\chi$, 
$\chi^{-1}:\mathrm{Tab}'(\alpha,\lambda;\beta)\rightarrow\mathrm{Tab}(\alpha,\lambda;\beta)$,
$U\mapsto T$, 
is given by the following two steps.

\paragraph{Step A}
Let $Q_{0}(\lambda)$ (resp. $Q_{0}(\alpha)$) be a standard tableau of 
shape $\lambda$ (resp. $\alpha$) by enumerating boxes 
in $\lambda$ (resp. $\alpha$) by $1,2\ldots,|\lambda|$ (resp. $1,2\ldots,|\alpha|$) 
from left to right in a row starting from the top row to bottom.
By reversing the insertions, we define 
$w_{0}:=\mathrm{mixRSK}^{-1}(U,Q_{0}(\lambda))$ and 
$w_{1}:=\mathrm{RSK}^{-1}(U_{\alpha},Q_{0}(\alpha))$.
We obtain $w$ by a concatenation of $w_{0}$ and $w_{1}$, that is, 
$w:=w_{0}\ast w_{1}$.

\paragraph{Step B}
Given a word $w$, we obtain a standard tableau $Q_{\mathrm{stand}}(\beta)$ of shape $\beta$
by the RSK insertion, $Q_{\mathrm{stand}}(\beta):=Q_{\mathrm{RSK}}(w)$.
We perform a destandardization on $Q_{\mathrm{stand}}(\beta)$ with respect to 
$\lambda$ and $\alpha$.
Here, destandardization is a reversed procedure of standardization. 
The boxes with contents from $1$ to $|\lambda|$ form a ordinary shape $\gamma\subseteq\beta$.
If we replace the contents from $\lambda_{i-1}+1$ to $\sum_{k=1}^{i}\lambda_{k}$ by 
$i'$ and $i$ for $1\le i\le l(\lambda)$, the reading word for the shape $\gamma$ is 
an LRS word.
Note that the destandardization is unique since the first occurrence of a letter 
$i'$ or $i$ is $i$ in an LRS word (see Section \ref{sec:tw}).
Similarly, we perform destandardization in the region $\beta/\gamma$ with 
respect to the content $\alpha$.
Then, we obtain a tableau $T\in\mathrm{Tab}(\alpha,\lambda;\beta)$.

By summarizing observations above, the map $\chi$ is a well-defined bijection. 
\end{proof}

\begin{example}
The product $s_{(3,2)}$ and $P_{(4,2,1)}$ contains $9s_{(5,4,2,1)}$, 
{\it i.e.}, $b_{(3,2)(4,2,1)}^{(5,4,2,1)}=9$.
One of them is given by the left picture in Figure \ref{fig:sPs1}.
Here, the integers $1,2$ and $3$ form an LRS word of content $(4,2,1)$
and the integers $4$ and $5$ form a Yamanouchi word of content $(3,2)$.
The right picture in Figure \ref{fig:sPs1} is the corresponding standard 
tableau $Q_{\mathrm{stand}}(\beta)$.
Then, reversing the RSK insertion, we obtain a word 
$w:=w_{0}\ast w_{1}=4233121\ast12211$.
The tableau $U$ and $T_{\alpha}$ are given by
\begin{eqnarray*}
U=P_{\mathrm{mix}}(w_{0})
=
\begin{matrix}
1 & 1 & 2' & 4' \\
  & 2 & 3' \\
  &   & 3 \\
\end{matrix}, 
\qquad\qquad
T_{\alpha}=P_{\mathrm{RSK}}(w_{1})=
\begin{matrix}
1 & 1 & 1 \\
2 & 2 \\
\end{matrix}.
\end{eqnarray*}
It is easy to check that $\mathrm{wread}(U)\ast w_{\alpha}=4233211\ast22111$
is Yamanouchi.
\begin{figure} 
\begin{eqnarray*}
\begin{matrix}
1' & 1 & 1 & 4 & 4 \\
1 & 2 & 4 & 5 \\
2 & 5 \\
3
\end{matrix}
\qquad\qquad
\begin{matrix}
1 & 3 & 4 & 9 & 10 \\
2 & 6 & 8 & 12 \\
5 & 11 \\
7
\end{matrix}
\end{eqnarray*}
\caption{A configuration for the product of $s_{(3,2)}$ and $P_{(4,2,1)}$ and its 
standardization}
\label{fig:sPs1}
\end{figure}
\end{example}

By exchanging roles of $\alpha$ and $\lambda$, we have another expression of 
Eqn.(\ref{eq:PPs21}).
Let $T$ be a shifted tableau of shape $\beta/\alpha$ and its reading word 
$\mathrm{read}(T)$ be an LRS word of content $\lambda$.
Let $U_{\alpha}$ be a tableau of ordinary shape $\alpha$ and 
$U$ be a shifted tableau of shape $\lambda$ and its weak reading word 
$\mathrm{wread}(U)$ satisfy that $\mathrm{read}(U_{\alpha})\ast\mathrm{wread}(U)$ 
is a Yamanouchi word of content $\beta$.
We construct a bijection $\chi'$ between a shifted tableau $T$ and a pair of 
$(U_{\alpha},U)$ in the following two steps.

\paragraph{Step 1}
Note that $\beta=\mathrm{shape}(T)\cup\alpha$ and $\alpha\subseteq\beta$.
We enumerate boxes in $\alpha\subseteq\beta$ by $1,2\ldots,|\alpha|$ from 
left to right in a row starting from the top row to bottom.
Then, we perform a standardization on $T$ by 
$|\alpha|+1, |\alpha|+2,\ldots,|\beta|$.
By reversing the RSK algorithm, we obtain a word $\tilde{w}$, {\it i.e.},  
$\tilde{w}:=\mathrm{RSK}^{-1}(P_{0}(\beta),Q_{\mathrm{stand}}(\beta))$.

\paragraph{Step 2}
We divide the word $\tilde{w}$ into a concatenation of two words 
$\tilde{w}=\tilde{w}_{0}\ast \tilde{w}_{1}$ such that the length 
of $\tilde{w}_0$ (resp. $\tilde{w}_{1}$) is $|\alpha|$ (resp. $|\lambda|$).
Finally, we define a shifted tableau $U$ of shape $\lambda$ by 
$U:=P_{\mathrm{mix}}(\tilde{w}_{1})$ and a tableau $U_{\alpha}$ of 
shape $\alpha$ by $U_{\alpha}:=P_{\mathrm{RSK}}(\tilde{w}_{0})$.

By summarizing the above discussion and a similar argument to 
Lemma \ref{lemma:PPs1}, we have 
\begin{lemma}
The map $\chi':T\mapsto(U_{\alpha},U)$ is a bijection.
\end{lemma}

\begin{proof}[Proof of Theorem \ref{thrm:PPs2}]
If we set $\alpha=\emptyset$ in Lemma \ref{lemma:PPs1}, we 
obtain $e_{\emptyset\mu}^{\alpha}=b_{\emptyset\mu}^{\alpha}$, 
which is equivalent 
to $e_{\emptyset\mu}^{\alpha}=|\mathrm{Tab}^{\prime}(\emptyset,\mu;\alpha)|$.
From the definition of $\mathrm{Tab}^{\prime}(\emptyset,\mu;\alpha)$, a shifted 
tableau of shape $\mu$ has the content $\alpha$ and its weak reading word 
is a Yamanouchi word of content $\alpha$.
We have 
\begin{eqnarray*}
e_{\lambda\mu}^{\alpha}=\sum_{\gamma}b_{\gamma\lambda}^{\alpha}e_{\emptyset\mu}^{\gamma}.
\end{eqnarray*}
By applying Lemma \ref{lemma:PPs1} to $b_{\gamma\lambda}^{\alpha}$, we have 
\begin{eqnarray*}
e_{\lambda\mu}^{\alpha}=\sum_{\gamma}
|\mathrm{Tab}^{\prime}(\gamma,\lambda;\alpha)|\cdot|\mathrm{Tab}^{\prime}(\emptyset,\mu;\gamma)|,
\end{eqnarray*}
which implies Theorem \ref{thrm:PPs2} is true.
\end{proof}

\begin{theorem}
\label{thrm:PSP}
We have
\begin{eqnarray*}
f_{\lambda\alpha}^{\mu}=\#\{T\in T''(\mu/\lambda;\alpha)\ \vert\  
\mathrm{wread}(T) \text{ is a Yamanouchi word}\}.
\end{eqnarray*}
\end{theorem}

\begin{proof}
Since $\hat{S}_{\alpha}$ has a determinant expression Eqn.(\ref{eqn:Sdet}), 
the coefficient $f_{\lambda\alpha}^{\mu}$ is rewritten as 
\begin{eqnarray*}
f_{\lambda\alpha}^{\mu}&=&\left[P_{\lambda}\hat{S}_{\alpha},Q_{\mu}\right] \\
&=&\left[\hat{S}_{\alpha},Q_{\mu/\lambda}\right] \\
&=&\sum_{\pi\in\mathcal{S}_{l}}\mathrm{sgn}(\pi)
\left[q_{\pi\ast\alpha},Q_{\mu/\lambda}\right].
\end{eqnarray*}
When $\beta=(n)$, by applying Theorem \ref{thrm:StemPP}, we have 
\begin{eqnarray*}
\left[q_{(n)},Q_{\mu/\lambda}\right]
=\#\{T\in T^{\prime\prime}(\mu/\lambda)\ \vert\ \mathrm{wread}(T)
\text{ is a Yamanouchi word of content } (n)\}.
\end{eqnarray*}
Thus, $f_{\lambda\alpha}^{\mu}$ is given by 
\begin{eqnarray*}
f_{\lambda\alpha}^{\mu}=\sum_{(\pi,T)\in\mathcal{M}}
\mathrm{sgn}(\pi),
\end{eqnarray*}
where 
\begin{eqnarray*}
\mathcal{M}
:=\{(\pi,T) \vert 
\pi\in\mathcal{S}_{l}, T\in T^{\prime\prime}(\mu/\lambda;\pi\ast\alpha)\}.
\end{eqnarray*}
We define the subset $\mathcal{M}_{0}$ of $\mathcal{M}$ by 
\begin{eqnarray*}
\mathcal{M}_{0}:=
\{(\pi,T)\in\mathcal{M} \vert \mathrm{wread}(T) \text{ is a Yamanouchi 
word of content } \alpha\}.
\end{eqnarray*}
We want to show that $f_{\lambda\alpha}^{\mu}=|\mathcal{M}_{0}|$. 
By replacing $\gamma/\alpha$ and $\mathcal{L}$ in the proof of Theorem
\ref{thrm:sS} by $\mu/\lambda$ and $\mathcal{M}$, 
we have a proof of Theorem.
The difference between Theorem \ref{thrm:sS} and Theorem \ref{thrm:PSP} 
is that we may have a skew tableau which is not detached.
We can construct an involution for non-detached shape by the following 
procedure developed by Stembridge (see Section 6 in \cite{Ste90}).
Suppose that a skew shape $T$ is a tableau formed by two letters 
and $T$ is not detached.
Without loss of generality, a tableau $T$ has $k-1$ main diagonals of 
length two and the $k$-th main diagonal is of length one.
The entries on the first main diagonal can be primed or unprimed.
We delete the first $k-1$ diagonals and let $S$ be a new tableau
and $a$ be the unique entry on the main diagonal in $S$.
After deletion of the diagonals, the entry $a$ can be either primed 
or unprimed.
Then, we choose one of the entries in the first main diagonal 
of $T$ and let $b$ be this entry.
We set $a$ primed (resp. unprimed) if $b$ is primed (resp. unprimed).
We have a two-to-one map $\vartheta:T\mapsto S$. 
One can perform involutions $\psi,\omega_{1,2}$ and $\omega_{1',2'}$
on $S$.
We have an one-to-two map by reversing the map $\vartheta$. 
Thus, we have a two-to-two involution on $T$.  
This completes the proof.
\end{proof}

\begin{theorem}
\label{thrm:SSs1}
Let $A\in T''(\alpha;\ast)$ and $B\in T''(\beta;\ast)$ be a shifted tableaux.
We denote $w(A):=\mathrm{wread}(A)$ and $w(B):=\mathrm{wread}(B)$. 
We have 
\begin{eqnarray*}
g_{\alpha\beta}^{\gamma}
=\#\{(w(A),w(B))\ \vert\  w(B)\ast w(A) \text{ is a Ymanouchi word of content }\gamma\}.
\end{eqnarray*}
\end{theorem}

Let $\alpha,\beta$ and $\gamma$ be ordinary partitions. 
We define $w_{\alpha}$ as the reading word $\mathrm{read}(T_{\alpha})$ where 
$T_{\alpha}$ is a unique tableau whose shape and weight are $\alpha$.
\begin{lemma}
\label{lemma:SSs1}
Let $G_{\alpha\beta}^{\gamma}$ be an expansion coefficient of $s_{\alpha}\hat{S}_{\beta}$
in terms of $s_{\gamma}$, {\it i.e.}, 
$s_{\alpha}\hat{S}_{\beta}=\sum_{\gamma}G_{\alpha\beta}^{\gamma}s_{\gamma}$.
Then, we have 
\begin{eqnarray}
\label{eq:sSs1}
G_{\alpha\beta}^{\gamma}
=\#\{T\in T''(\beta;\gamma/\alpha)\ \vert \ \mathrm{wread}(T)\ast w_{\alpha} 
\text{ is a Yamanouchi word of content }\gamma\}.
\end{eqnarray}
\end{lemma}
\begin{proof}
We have 
\begin{eqnarray*}
G_{\alpha\beta}^{\gamma}&=&\langle s_{\alpha}\hat{S}_{\beta},s_{\gamma}\rangle \\
&=&\langle \hat{S}_{\beta},s_{\gamma/\alpha}\rangle  \\
&=&\sum_{\mu}f_{\emptyset\beta}^{\mu}\langle P_{\mu},s_{\gamma/\alpha}\rangle \\
&=&\sum_{\mu}f_{\emptyset\beta}^{\mu}\left[ P_{\mu},\hat{S}_{\gamma/\alpha}\right] \\
&=&\left[\hat{S}_{\beta},\hat{S}_{\gamma/\alpha}\right]  \\
&=&\langle s_{\beta},\hat{S}_{\gamma/\alpha}\rangle.
\end{eqnarray*}
The skew $\hat{S}_{\gamma/\alpha}$ has a determinant expression:
\begin{eqnarray*}
\hat{S}_{\gamma/\alpha}=\det[q_{\gamma_{i}-\alpha_{j}-i+j}].
\end{eqnarray*}
By a similar argument to a proof of Theorem \ref{thrm:sS}, we obtain 
Eqn.(\ref{eq:sSs1}).

\end{proof}

\begin{proof}[Proof of Theorem \ref{thrm:SSs1}]
We want to expand the product $\hat{S}_{\alpha}\hat{S}_{\beta}$ in 
terms of Schur functions $s_{\gamma}$.
First, by setting $\alpha=\emptyset$ in Lemma \ref{lemma:SSs1},
we obtain $\hat{S}_{\alpha}=\sum_{\zeta}G_{\emptyset\alpha}^{\zeta}s_{\zeta}$.
Given $\zeta$, we obtain  from Lemma \ref{lemma:SSs1} that 
$s_{\zeta}\hat{S}_{\beta}=\sum_{\gamma}G_{\zeta\beta}^{\gamma}s_{\gamma}$.
Thus, the coefficient $g_{\alpha\beta}^{\gamma}$ is given by 
\begin{eqnarray}
\label{eqn:SSs1}
g_{\alpha\beta}^{\gamma}
=
\sum_{\zeta}G_{\emptyset\alpha}^{\zeta}
G_{\zeta\beta}^{\gamma},
\end{eqnarray}
where the sum is taken over all ordinary partitions $\zeta$'s.
Let $w$ be a Yamanouchi word of content $\gamma$. 
If we write $w$ as a concatenation of two word $w_{0}\ast w_{1}$,
then the word $w_{1}$ is again a Yamanouchi word of content $\zeta$ 
for some $\zeta\subseteq\gamma$.
Therefore, Eqn.(\ref{eqn:SSs1}) and Lemma \ref{lemma:SSs1} imply 
Theorem \ref{thrm:SSs1}.
\end{proof}

\begin{theorem}
\label{thrm:SSP}
Let $A\in T''(\alpha;\ast)$ and $B\in T''(\beta;\ast)$ be a shifted tableaux.
We have 
\begin{eqnarray*}
h_{\alpha\beta}^{\lambda}
=2^{l(\lambda)}\cdot\#\{(A,B)\ \vert \ \mathrm{read}(A)*\mathrm{read}(B) 
\text{ is an LRS word of content }\lambda\}.
\end{eqnarray*}
\end{theorem}
Let $\mathrm{Tab}^{\gamma}(\lambda)$ be the set of semistandard shifted tableaux
of shape $\gamma$ such that alphabets from $1$ to $l(\lambda)$ form an LRS word 
of content $\lambda$.
Similarly, let $\mathrm{Tab}^{\gamma}(\alpha,\beta)$ be the set of 
semistandard shifted tableaux  of shape $\gamma$ without primed letters such that 
alphabets from $1$ to $l(\alpha)$ (resp. from $l(\alpha)+1$ to 
$l(\alpha)+l(\beta)$) form a Yamanouchi word of content $\alpha$ (resp. $\beta$).
Let $\mathrm{Tab}^{\prime\prime}(\alpha,\beta;\gamma)$ be the set of pairs of 
shifted tableaux:
\begin{eqnarray*}
\mathrm{Tab}^{\prime\prime}(\alpha,\beta;\lambda)
:=\{(A,B)\in T^{\prime\prime}(\alpha;\ast)\times T^{\prime\prime}(\beta;\ast)\ \vert \  
\mathrm{read}(A;B) \text{ is an LRS word of content }\lambda\}.
\end{eqnarray*}
where $\mathrm{read}(A;B):=\mathrm{read}(A)\ast\mathrm{read}(B)$.

We will construct a bijection between 
a pair of tableaux $(T,U)$ where $T\in\mathrm{Tab}^{\gamma}(\lambda)$ and 
$U\in\mathrm{Tab}^{\gamma}(\alpha,\beta)$ and a pair of tableaux $(A,B)$ where 
the concatenation of two words $\mathrm{read}(A)\ast \mathrm{read}(B)$ is 
an LRS word.
Namely, we define a map 
$\theta:\mathrm{Tab}^{\gamma}(\lambda)\times\mathrm{Tab}^{\gamma}(\alpha,\beta)
\rightarrow \mathrm{Tab}^{\prime\prime}(\alpha,\beta;\lambda), (T,U)\mapsto(A,B)$,  
as the following two steps.

\paragraph{Step 1}
We perform a standardization on $T$ and $U$, and obtain a pair of 
standard tableaux of shape $\gamma$, $(\mathrm{stand}(T),\mathrm{stand}(U))$.
Then, by reversing the RSK insertion, we obtain a word 
$w:=\mathrm{RSK}^{-1}(\mathrm{stand}(T),\mathrm{stand}(U))$.

\paragraph{Step 2}
We split the word $w$ into a concatenation of two words $w=w_{1}\ast w_{2}$
where $w_{1}$ (resp. $w_{2}$) is a word of length $|\alpha|$ (resp. $|\beta|$).
By the RSK insertion, we obtain a pair of standard tableaux of shape $\alpha$ 
and $\beta$, denote by 
$(\mathrm{stand}(A),\mathrm{stand}(B))=(P_{\mathrm{RSK}}(w_{1}),P_{\mathrm{RSK}}(w_{2}))$.
Finally, we perform a destandardization on the pair $(\mathrm{stand}(A),\mathrm{stand}(B))$,
and we obtain the pair $(A,B)$ of semistandard tableaux of shape $\alpha$ and $\beta$.
Here, destandardization means that if the box $b$ has a content $j$ in $\mathrm{stand}(T)$ 
and $b$ in $T$ has a letter $x=i$ or $i'$, we put $x$ on a box in $A$ or $B$ which has 
a content $j$ in $\mathrm{stand}(A)$ or $\mathrm{stand}(B)$. 

\begin{lemma}
The map $\theta$ satisfies
\begin{enumerate}
\item Two tableaux $\mathrm{stand}(A)$ and $\mathrm{stand}(B)$ are of shape $\alpha$ and $\beta$.
\item Two tableaux $A$ and $B$ are well-defined. 
In other words, $A$ and $B$ are semistandard shifted tableaux with respect to the 
marked alphabet $X'$.
\item The word $\mathrm{read}(A)\ast\mathrm{read}(B)$ is an LRS word of content $\lambda$.
\end{enumerate}
\end{lemma}
\begin{proof}
For (1), we will show that $\mathrm{stand}(A)$ is of shape $\alpha$ since the case for 
$\mathrm{stand}(B)$ can be shown in a similar way.
Note that the letters from $\sum_{j=1}^{i-1}\alpha_{j}+1$ to $\sum_{j=1}^{i}\alpha_{j}$
in a tableau $U$ form a horizontal strip.
By reversing the RKS insertion for a horizontal strip, 
the words $w_{1}$ is written as $w_{1}=\tilde{w}_{1}\ldots\tilde{w}_{l(\alpha)}$ 
where $\tilde{w}_{i}$ is a weakly increasing 
sequence of length $\alpha_{i}$.
In $\mathrm{stand}(U)$, the position of the letter $\sum_{j=1}^{i-1}\alpha_{j}$ is upper than that of 
the letter $\sum_{j=1}^{i}\alpha_{j}$, since the alphabets from $1$ to $l(\alpha)$ form 
a Yamanouchi word of content $\alpha$.
By reversing the RSK insertion, this condition is rephrased as the following condition:
$\tilde{w}_{i-1,\alpha_{i-1}}>\tilde{w}_{i,\alpha_{i}}$ for $2\le i\le l(\alpha)$.
We delete the boxes with letters  $\sum_{j=1}^{i-1}\alpha_{j}$ and $\sum_{j=1}^{i}\alpha_{j}$ from 
$\mathrm{stand}(U)$. 
This means that we delete the right most boxes with integer $i-1$ and $i$ in $U$.
In the obtained tableau, letters $i-1$ and $i$ form a Yamanouchi word of content 
$(\alpha_{i}-1,\alpha_{i-1}-1)$.
By applying the same argument, we obtain the condition 
$\tilde{w}_{i-1,\alpha_{i-1}-j+1}>\tilde{w}_{i,\alpha_{i}-j+1}$ 
for $2\le j\le \alpha_{i}$.
These observations imply that $\mathrm{stand}(A)=P_{\mathrm{RSK}}(w_{1})$ is of 
shape $\alpha$.

For (2), since a tableau $\mathrm{stand}(A)$ is obtained by the RSK insertion, it is obvious 
that unprimed letters in the tableau $A$ satisfy the semistandard property, {\it i.e.},
unprimed letters appear at most once in a column of $A$.
Boxes corresponding to a primed letter $i'$ in $T$ form a vertical strip, 
which means that contents corresponding to $i'$ in $\mathrm{stand}(T)$ 
appear as a decreasing sequence in the word $w$.
This property holds true after the RSK insertion of $w_{1}$ or $w_{2}$.
In other words, boxes for a primed letter $i'$ form a vertical strip 
in both $A$ and $B$.
These imply that the tableau $A$ is a semistandard shifted tableau 
and so is $B$.

For (3), note that the destandardization of the word $w=w_{1}\ast w_{2}$ 
is an LRS word of content $\lambda$.
The destandardization of the partial word $w_{2}$ is also an LRS word.
Thus, the tableau word $\mathrm{read}(B)$ is an LRS word.
A tableau $A$ is given by a destandardization of $P_{\mathrm{RSK}}(w_{1})$. 
Suppose that $w_{1}$ is written as $w_{1}:=v_{1}\ldots v_{l}$ where $l=l(\alpha)$.
We show the statement (3) by induction.
The word $\mathrm{read}(P_{\mathrm{RSK}}(v_{1}))=v_{1}$ and a destandardization 
of the concatenation words $v_{1}\ast v_{2}\ldots v_{l}\ast w_{2}$ is an LRS word.
We insert $v_{i}$ into $P_{\mathrm{RSK}}(v_{1}\ldots v_{i-1})$. 
By induction assumption, the destandardization of 
$\mathrm{read}(P_{\mathrm{RSK}}(v_{1}\ldots v_{i-1}))\ast v_{i}\ldots v_{l}\ast w_{2}$
is an LRS word.
We first bump out smallest $x$ which is greater than $v_{i}$, and bumped element 
$x$ moves to the next row. 
Thus, in the word $\mathrm{read}(P_{\mathrm{RSK}}(v_{1}\ldots v_{i}))$, $x$ is 
left to $v_{i}$.
Since a letter $v_{i}$ stays in the first row of $P_{\mathrm{RSK}}(v_{1}\ldots v_{i})$ 
and the elements right to $v_{i}$ are strictly greater than $v_{i}$,
the lattice property holds after the insertion and destandardization.
In general, elements bumped out by the RSK insertion are inserted into the next 
row in $P_{\mathrm{RSK}}(v_{1}\ldots v_{i})$. 
After destandardization, the lattice property holds true.
As a result, the concatenation of three words 
$\mathrm{read}(P_{\mathrm{RSK}}(v_{1}\ldots v_{i}))$, $v_{i+1}\ldots v_{l}$ 
and $w_{2}$ is an LRS word.
By induction, we obtain that $\mathrm{read}(P_{\mathrm{RSK}}(w_{1}))\ast w_{2}$ is an 
LRS word.
Combining these observations with the fact that $\mathrm{read}(B)$ is an LRS
word, we obtain that $\mathrm{read}(A)\ast \mathrm{read}(B)$ is an LRS word of 
content $\lambda$.
\end{proof}

The inverse 
$\theta^{-1}:\mathrm{Tab}^{\prime\prime}(\alpha,\beta;\gamma)\rightarrow
\mathrm{Tab}^{\gamma}(\lambda)\times\mathrm{Tab}^{\gamma}(\alpha,\beta),
(A,B)\mapsto(T,U)$,
can be given by the following two steps.
\paragraph{Step A}
We perform a standardization on $(A,B)$ and obtain a pair of standard 
tableaux of shape $\alpha$ and $\beta$.
We enumerate the boxes in tableaux $A$ and $B$ with the letter $1'$ 
by $1,\ldots,\#(1')$ where $\#(1')$ is the total number of $1'$ in 
tableaux $A$ and $B$.
We start from the top row of the tableau $B$ to the bottom row of $B$, 
moving to the tableau $A$, and continue to enumerate boxes from the 
top row to the bottom row of $A$.
Then, we continue to enumerate the boxes with letter $1$ by 
$\#(1')+1,\ldots,\#(1')+\#(1)$ starting from the leftmost column of a tableau
$A$ to the rightmost column of $A$, moving to the tableau $B$, and successively
enumerate from the leftmost column to the rightmost column in $B$.
We similarly enumerate boxes with $i'$ or $i$ as above.
Note that the order of scanning boxes in tableaux $A$ and $B$ for $i$ is reversed 
compared to the case of $i'$. 
We denote by $\mathrm{stand}(A)$ a tableau on which standardization is performed.
Let $Q_0(\alpha)$ be a standard tableau of shape $\alpha$ obtained by enumerating 
boxes in $\alpha$ by $1,2,\ldots,|\alpha|$ from left to right in a row starting 
from the top row to bottom.
By reversing the RSK insertion, we obtain two words $w_{1}$ and $w_{2}$ as 
$w_{1}:=\mathrm{RSK}^{-1}(\mathrm{stand}(A),Q_{0}(\alpha))$ and 
$w_{2}:=\mathrm{RSK}^{-1}(\mathrm{stand}(B),Q_{0}(\beta))$.
The word $w$ is defined as a concatenation of $w_{1}$ and $w_{2}$, 
{\it i.e.}, $w:=w_{1}\ast w_{2}$.

\paragraph{Step B}
Given the word $w$, we obtain a pair of standard tableaux of shape $\gamma$ 
by the RSK insertion. 
we define $\mathrm{stand}(T):=P_{\mathrm{RSK}}(w)$ and 
$\mathrm{stand}(U):=Q_{\mathrm{RSK}}(w)$. 
Then, we perform a destandardization on $\mathrm{stand}(T)$ and $\mathrm{stand}(U)$
such that $T$ is formed by an LRS word of content $\lambda$ and $U$ is formed by
two Yamanouchi words of content $\alpha$ and $\beta$.
The pair of tableaux obtained by the above procedure is $(T,U)$.

\begin{lemma}
The map $\theta^{-1}$ satisfies
\begin{enumerate}
\item A tableau $T$ is well-defined, {\it i.e.}, $T$ is a semistandard 
shifted tableau of shape $\gamma$ of content $\lambda$. 
\item A tableau $U$ is well-defined, {\it i.e.}, $U$ is a semistandard 
tableau of shape $\gamma$ which is formed by a tableau $\alpha$ of content $\alpha$ and 
and a skew tableau $\gamma/\alpha$ of content $\beta$. 
\end{enumerate}
\end{lemma}
\begin{proof}
For (1), we observe that the positions of letters from 
$\sum_{k=1}^{i-1}\lambda_{k}+1$ to $\sum_{k=1}^{i}\lambda_{k}$ form 
a hook word of length $\lambda_{i}$ in the word $w$.
Let $\mathrm{pos}(j)$ be the position (from left end) of $\sum_{k=1}^{i-1}\lambda_{k}+i$ 
for $1\le i\le\lambda_{i}$ in $w$. 
Since $\mathrm{pos}(j)$'s form a hook word, there exists a unique $i_{0}$
such that 
$\mathrm{pos}(1)>\mathrm{pos}(2)>\cdots>\mathrm{pos}(i_{0})<\mathrm{pos}(i_{0}+1)<\cdots
<\mathrm{pos}(\lambda_{i})$.
Recall that a pair $(A,B)$ satisfies that the concatenation of two 
words $\mathrm{read}(A)\ast\mathrm{read}(B)$ is an LRS word of 
content $\lambda$. 
Two words $w_{1}$ and $w_{2}$ are given by the reversed RSK insertion, 
and destandardization of the word $w$ is an LRS word of content $\lambda$.
Since $\mathrm{pos}(j)$ with $1\le j\le i_{0}$ form a decreasing sequence 
in $w$, the letters from $\sum_{k=1}^{i-1}\lambda_{k}+1$ to 
$\sum_{k=1}^{i-1}\lambda_{k}+i_{0}$ form a vertical strip by the RSK insertion.
Similarly, since $\mathrm{pos}(j)$ with $i_{0}\le j\le \lambda_{i}$ form
a increasing sequence in $w$, the letters from $\sum_{j=1}^{i-1}\lambda_{j}+i_{0}$
to $\sum_{k=1}^{i}\lambda_{k}$ form a horizontal strip by the RSK insertion.
From these observations, a tableau $T$ is a semistandard shifted tableau
of content $\lambda$.
The shape $\gamma$ is nothing but the shape of $P_{\mathrm{RSK}}(w)$.

For (2), observe that $w=w_{1}\ast w_{2}$ where the length of 
$w_{1}$ (resp. $w_{2}$) is $|\alpha|$ (resp. $|\beta|$). 
The word $w_{1}$ is written as $w_{1}=u_{1}\ldots u_{l(\alpha)}$ where 
$u_{i}$ is a strictly increasing sequence of length $\alpha_{i}$ 
for $1\le i\le l(\alpha)$. 
Since the word $w_{1}$ is given by the reversed RSK insertion, 
we have $u_{i,\alpha_{i}}>u_{i+1,\alpha_{i+1}}$.
More generally, we have $u_{i,\alpha_{i}-j}>u_{i+1,\alpha_{i+1}-j}$ 
for $0\le j\le \alpha_{i+1}$.
From these conditions, the shape of insertion tableau $P_{\mathrm{RSK}}(w_{1})$ 
is $\alpha$. 
The tableau $U$ is obtained by destandardization of the recording 
tableau $Q_{\mathrm{RSK}}(w_{1})$ and its shape is the same as the one
of $P_{\mathrm{RSK}}(w_{1})$. 
We insert $w_{2}$ into $P_{\mathrm{RSK}}(w_{1})$. 
Let $\gamma$ be the shape of $P_{\mathrm{RSK}}(w)$. 
The word $w_{2}$ is also rewritten as a concatenation of increasing 
sequences. 
By a similar argument to the case of (1), $U$ has the region $\gamma/\alpha$
with content $\beta$.
This completes the proof.
\end{proof}

Summarizing the above observations, we have
\begin{lemma}
\label{lemma:SSP1}
There exists a bijection $\theta:\mathrm{Tab}^{\gamma}(\lambda)\times\mathrm{Tab}^{\gamma}(\alpha,\beta)
\rightarrow \mathrm{Tab}^{\prime\prime}(\alpha,\beta;\lambda)$.
\end{lemma}

\begin{example}
The product $\hat{S}_{(4,2,1)}$ and $\hat{S}_{(3,2,1)}$ contains 
$464P_{(5,4,3,1)}$. Thus, we have $464/2^{4}=29$ pairs which satisfy 
the lattice property.
Suppose $\gamma=(4,4,2,2,1)$.
We have two pairs of 
$(T,U)\in \mathrm{Tab}^{\gamma}(\lambda)\times\mathrm{Tab}^{\gamma}(\alpha,\beta)$.
One of them is the left picture of Figure \ref{fig:SSs1}.
\begin{figure}[ht]
\begin{eqnarray*}
\left(
\begin{matrix}
1' & 1 & 1 & 1 \\
1 & 2' & 2 & 2 \\
2 & 3' \\
3 & 3 \\
4
\end{matrix}
\quad,\quad  
\begin{matrix}
1 & 1 & 1 & 1 \\
2 & 2 & 4 & 4 \\
3 & 5 \\
4 & 6 \\
5
\end{matrix}
\right)
\qquad\qquad
\left(
\begin{matrix}
1 & 3 & 4 & 5 \\
2 & 6 & 8 & 9 \\
7 & 10 \\
11 & 12 \\
13
\end{matrix}\quad,\quad  
\begin{matrix}
1 & 2 & 3 & 4 \\
5 & 6 & 9 & 10 \\
7 & 12 \\
8 & 13 \\
11
\end{matrix}
\right)
\end{eqnarray*}
\caption{A pair $(T,U)$ and its standardization}
\label{fig:SSs1}
\end{figure}
We have an LRS word of content $(5,4,3,1)$ in $T$ and two Yamanouchi 
words of contents $(4,2,1)$ and $(3,2,1)$ in $U$.  
The right picture of Figure \ref{fig:SSs1} is its standardization.
By reversing the RSK insertion, 
we have the word 
$w=\mathrm{RSK}^{-1}(\mathrm{stand}(T),\mathrm{stand}(U))=(2,7,11,13,1,12,8,3,6,10,4,9,5)$,
which is a concatenation of two words $w_{1}\ast w_{2}=(2,7,11,13,1,12,8)\ast(3,6,10,4,9,5)$.
We obtain a pair of standard tableaux 
$(\mathrm{stand}(A),\mathrm{stand}(B)):=(P_{\mathrm{RSK}}(w_{1}),P_{\mathrm{RSK}}(w_{2}))$.
By destandardization, we have
\begin{eqnarray*}
(A,B)=\theta(T,U)=
\left(
\begin{matrix}
1' & 2 & 2 & 3 \\
1 & 3 \\
4
\end{matrix}\quad,\quad
\begin{matrix}
1 & 1 & 1 \\
2' & 2 \\
3'
\end{matrix}
\right).
\end{eqnarray*}
The tableau word $\mathrm{read}(A)\ast\mathrm{read}(B)$ is an LRS word 
of content $(5,4,3,1)$.

\end{example}

\begin{proof}[Proof of Theorem \ref{thrm:SSP}]
The coefficient $h_{\alpha\beta}^{\lambda}$ can be expressed as 
\begin{eqnarray}
\label{eqn:SSP1}
\begin{aligned}
h_{\alpha\beta}^{\lambda}&=\left[\hat{S}_{\alpha}\hat{S}_{\beta},Q_{\lambda}\right] \\
&=\sum_{\gamma}a_{\alpha\beta}^{\gamma}\left[S_{\gamma},Q_{\lambda}\right] \\
&=\sum_{\gamma}2^{l(\lambda)}a_{\alpha\beta}^{\gamma}b_{\lambda}^{\gamma}.
\end{aligned}
\end{eqnarray}
The number $a_{\alpha\beta}^{\gamma}b_{\lambda}^{\gamma}$ in Eqn.(\ref{eqn:SSP1}) 
is equal to 
$|\mathrm{Tab}^{\gamma}(\lambda)|\cdot|\mathrm{Tab}^{\gamma}(\alpha,\beta)|$.
From Lemma \ref{lemma:SSP1}, we have a bijection between 
$\mathrm{Tab}^{\gamma}(\lambda)\times\mathrm{Tab}^{\gamma}(\alpha,\beta)$ and 
$\mathrm{Tab}^{\prime\prime}(\alpha,\beta;\lambda)$.
Thus, 
$\sum_{\gamma}a_{\alpha\beta}^{\gamma}b_{\lambda}^{\gamma}
=|\mathrm{Tab}^{\prime\prime}(\alpha,\beta;\lambda)|$,
which implies Theorem is true.

\end{proof}

\subsection{Littlewood--Richardson--Stembridge coefficients revisited}
We will show three different expressions of coefficient $d_{\lambda\mu}^{\nu}$. 

Let $\lambda,\mu$ and $\nu$ be strict partitions.
We consider a standard shifted tableau $S$ of shape $\nu/\mu$.
Let $b_{1},\ldots,b_{\lambda_{i}}$ be boxes in $S$ labeled 
by an integer in the interval 
$[\sum_{j=1}^{i-1}\lambda_{j}+1,\sum_{j=1}^{i}\lambda_{j}]$.
If a box $b_{j}$ is in the $r$-th row in $\nu$, 
we denote $\mathrm{ht}(b_{j}):=r$.

The proposition below is a shifted analogue of the 
Remmel--Whitney algorithm~\cite{RemWhi84}.
It directly follows from Lemma 8.4 and Lemma 8.6 in \cite{Ste89}.

\begin{prop}
The number $d_{\lambda\mu}^{\nu}$ is the number of standard shifted tableaux 
of $S$ of shape $\nu/\mu$ satisfying the following:
\begin{enumerate}
\item 
$\mathrm{ht}(b_{1})<\mathrm{ht}(b_{2})<\ldots<\mathrm{ht}(b_{r})
\ge\mathrm{ht}(b_{r+1})\ge\ldots \ge\mathrm{ht}(b_{\lambda_{i}})$
for $1\le i\le l(\lambda)$.
\item the shape of $P_{\mathrm{shift}}(\mathrm{read}(S))$ is $\lambda$. 
\end{enumerate}
\end{prop}

\begin{theorem}
We have
\label{thrm:PPP1}
\begin{eqnarray*}
d_{\lambda\mu}^{\nu}
=\sum_{\alpha}b_{\lambda}^{\alpha}\cdot 
\#\{T\in T(\nu/\mu;\alpha)\ \vert\  T \text{ is a Yamanouchi word}\}.
\end{eqnarray*}
\end{theorem}

\begin{remark}
Since $b_{\lambda}^{\alpha}=e^{\alpha}_{\lambda}$, we have 
\begin{eqnarray*}
d_{\lambda\mu}^{\nu}=\sum_{\alpha}e^{\alpha}_{\lambda}\cdot 
\#\{T\in T(\nu/\mu;\alpha) \ \vert \ T \text{ is a Yamanouchi word}\}
\end{eqnarray*}
from Theorem \ref{thrm:PPP1}.
From Theorem \ref{thrm:PPs2}, $e^{\alpha}_{\lambda}$ is expressed in terms of 
Yamanouchi words.
Thus, we have a description of $d_{\lambda\mu}^{\nu}$ in terms of Yamanouchi words
rather than LRS words.
The coefficient $e^{\alpha}_{\lambda}$ is expressed in terms of weak reading words.
Thus, the complexity of an LRS word is replaced with a weak reading word.
\end{remark}

For a proof of Theorem \ref{thrm:PPP1}, we introduce the following lemma.
Let $\alpha$ be an ordinary partition and $\lambda,\mu$ and $\nu$ be 
strict partitions.
Let $\mathrm{Tab}^{\alpha}(\lambda)$ (resp. $\mathrm{Tab}^{\nu/\mu}(\lambda)$)
be the set of shifted tableaux of shape $\alpha$ (resp. $\nu/\mu$) which have 
an LRS word of content $\lambda$, and  
$\mathrm{Tab}^{\nu/\mu}(\alpha)$ be the set of skew tableaux
of shape $\nu/\mu$ which have a Yamanouchi word of content $\alpha$.

\begin{lemma}
\label{lemma:PPP1}
There exists a bijection $\zeta:\mathrm{Tab}^{\nu/\mu}(\lambda)\rightarrow
\mathrm{Tab}^{\alpha}(\lambda)\times\mathrm{Tab}^{\nu/\mu}(\alpha)$.
\end{lemma}
\begin{proof}
We construct a map  
$\zeta:\mathrm{Tab}^{\nu/\mu}(\lambda)\rightarrow
\mathrm{Tab}^{\alpha}(\lambda)\times\mathrm{Tab}^{\nu/\mu}(\alpha),
T\mapsto(A,B)$, as follows.
Since $\mathrm{read}(T), T\in\mathrm{Tab}^{\nu/\mu}(\lambda)$, is an LRS 
word of $\lambda$, we first perform a standardization of $T$ and define a word 
$w:=\mathrm{read}(\mathrm{stand}(T))$.
By the RSK insertion, we obtain a pair of standard tableaux 
$(P_{\mathrm{RSK}}(w),Q_{\mathrm{RSK}}(w))$ of shape $\alpha$.
A tableau $A\in\mathrm{Tab}^{\alpha}(\lambda)$ is obtained by destandardization 
of $P_{\mathrm{RSK}}(w)$ with respect to the content $\lambda$.
We construct a tableau $B\in\mathrm{Tab}^{\nu/\mu}(\alpha)$ from $Q_{\mathrm{RSK}}(w)$
as follows.
Let $P_{0}(\alpha)$ be a standard tableau enumerated by $1,\ldots,|\alpha|$
from left to right in row starting from the top row to bottom.
By reversing the RSK insertion, we obtain a word 
$w':=\mathrm{RSK}^{-1}(P_{0}(\alpha),Q_{\mathrm{RSK}}(w))$.
A standard tableau $\mathrm{stand}(B)$ of shape $\nu/\mu$ is given from $w'$ such that 
the reading word $\mathrm{stand}(B)$ is $w'$.
Finally, we perform a destandardization on $\mathrm{stand}(B)$ 
with respect to the content $\alpha$.

We claim: 
\begin{enumerate}
\item[(S6)] A standard tableau $\mathrm{stand}(B)$ is well-defined. 
In other words, the word $w'$ is compatible with the skew shape $\nu/\mu$.
\item[(S7)] A semistandard tableau $B$ is compatible with the skew shape $\nu/\mu$ and 
the word $\mathrm{read}(B)$ is a Yamanouchi word of content $\alpha$.
\end{enumerate}

We show that $\mathrm{stand}(B)$ is a well-defined standard tableau.
Let $\xi\in\mathbb{Z}^{l(\nu)}$ be a sequence of non-negative integers of length $l(\nu)$ 
defined by $\xi:=\nu-\mu$ where we set $\mu_{i}=0$ for $l(\mu)<i\le l(\nu)$.
In the skew shape $\nu/\mu$, we have $\xi_{i}$ boxes in the $i$-th row.
Given $T\in\mathrm{Tab}^{\nu/\mu}(\lambda)$, the $\xi_{i}$ letters in $i$-th row
are strictly increasing in $\mathrm{stand}(T)$.
By the RSK insertion of $w$, the letters from $\sum_{j=1}^{i-1}\xi_{j}+1$ to 
$\sum_{j=1}^{i}\xi_{i}$ form a horizontal strip in the recording tableau 
$Q_{\mathrm{RSK}}(w)$.
We obtain $w'$ by reversing the RSK insertion starting from the pair 
$(P_{0}(\alpha),Q_{\mathrm{RSK}}(w))$.
The above property of $Q_{\mathrm{RSK}}(w)$ implies that 
$w'$ is written as a concatenation of words 
$w':=w'_{l(\nu)}\ldots w'_{2}w'_{1}$ where $w'_{i}$ is 
an increasing sequence of length $\xi_{i}$ for $1\le i\le l(\nu)$.
This means that $\mathrm{stand}(B)$ is standard in a row.

Suppose that the skew shape $\nu/\mu$ has $k$ boxes in the main 
diagonal.
We enumerate rows of $\nu/\mu$ by $1,2\ldots,l(\nu)$ from bottom to top.
The $k$ rows from bottom form a shifted tableau 
$\tau:=(\tau_{1},\tau_{2},\ldots,\tau_{k})$ and 
rows from $k$-th to $l(\nu)$-th form a ordinary skew shape 
$\beta/\gamma$.
Note that both $\tau$ and $\beta/\gamma$ contain the $k$-th 
row of $\nu/\mu$.
We consider a sequence of ordinary partitions obtained from 
a shifted partition $\tau$ by $\emptyset=\hat{\tau}_{0}\subseteq\hat{\tau}_{1}
\subseteq\ldots\subseteq\hat{\tau}_{k}=\tau$ 
where $\hat{\tau}_{i}:=(\tau_{k-i+1},\tau_{k-i+2},\ldots,\tau_{k})$.
The skew shape $\tau_{i}/\tau_{i-1}$ forms a horizontal strip of 
length $\tau_{k-i+1}$.
Then, we put letters from $\sum_{l=1}^{i-1}\tau_{k-l+1}+1$ to 
$\sum_{l=1}^{i}\tau_{k-l+1}$ on the boxes of $\tau_{i}/\tau_{i-1}$ 
from left to right.
We obtain a standard tableau of ordinary shape $\tau$ and denote 
it by $Q(\tau)$.
By the RSK insertion of $w$, the recording tableau $Q_{\mathrm{RSK}}(w)$
contains the tableau $Q(\tau)$. 
We enumerate boxes in $\nu/\mu$ by $1,2,\ldots,|\nu/\mu|$ from 
left to right in a row starting from the bottom rows to top.
When the boxes labelled by letters in $[i,j]$ forms a row in $\nu/\mu$, 
the letters $[i,j]$ form a horizontal strip in $Q_{\mathrm{RSK}}(w)$.
Suppose that the box $b_{j}$ labelled as $j$ is in the same column as the box $b_{i}$ 
labelled as $i$ with $i<j$ in $\nu/\mu$.  
When the boxes $b_{i}$ and $b_{j}$ are in $\beta/\gamma$, the letter $j$ 
is below the letter $i$ in $Q_{\mathrm{RSK}}(w)$.
Similarly, when the boxes $b_{i}$ and $b_{j}$ are in $\tau$,
the letter $j$ is below the letter $i$ or next to each other in
the same row in $Q_{\mathrm{RSK}}(w)$.
Since 
$(P_{\mathrm{RSK}}(\rho),Q_{\mathrm{RSK}}(\rho))
=(Q_{\mathrm{RSK}}(\rho^{-1}),P_{\mathrm{RSK}}(\rho^{-1}))$ for 
a permutation $\rho$,
we consider the inverse of $w'^{-1}$ where 
$w'^{-1}=\mathrm{RSK}^{-1}(Q_{\mathrm{RSK}}(w),P_{0}(\alpha))$.
By reversing the RSK insertion, 
if letters $i_{1}<i_{2}<\ldots<i_{r}$ are in the same column 
in $\nu/\mu$, 
$i_{r},i_{r-1},\ldots,i_{1}$ appear in $w'^{-1}$ as a decreasing 
sequence.
By considering the inverse of $w'^{-1}$, the letters in boxes 
labelled $i_{1},i_{2},\ldots,i_{r}$ are decreasing from bottom 
to top.
This implies that $\mathrm{stand}(B)$ is column standard.
Since $\mathrm{stand}(B)$ is standard in both rows and columns, 
the statement (S6) is true.

Let $\tilde{P}_{0}(\alpha)$ be a unique semistandard tableau of shape 
$\alpha$ and content $\alpha$.
The statement (S7) is equivalent to show that the word 
$\tilde{w}:=\mathrm{RSK}^{-1}(\tilde{P}_{0}(\alpha),Q_{\mathrm{RSK}}(w))$
fits to the skew shape $\nu/\mu$ and is Yamanouchi word 
of content $\alpha$.
Since $P_{0}(\alpha)$ is the standardization of $\tilde{P}_{0}(\alpha)$,
the word $\tilde{w}$ is compatible with the shape $\nu/\mu$ 
by a similar argument to the case of (S6). 
Since the contents of $i$-th row of $\tilde{P}_{0}(\alpha)$ are all $i$'s, 
it is obvious that $\tilde{w}$ is a Yamanouchi.
These observations imply that (S7) is true.

From the construction, it is obvious that the map $\zeta$ has an inverse 
and well-defined by a similar argument to a proof of (S6) and (S7).
Thus, $\zeta$ is a bijection.

\end{proof}

\begin{example}
Let $\lambda:=(6,4,2), \mu:=(5,2)$ and $\nu:=(7,5,4,2,1)$. 
\begin{figure}[ht]
\begin{eqnarray*}
T=
\begin{matrix}
* & * & * & * & * & 1' & 1\\
  & * & * & 1 & 1 & 1 \\
  &   & 1 & 2' & 2 & 2 \\
  &   &   & 2 & 3' \\
  &   &   &   & 3
\end{matrix}\qquad\qquad
\begin{matrix}
* & * & * & * & * & 1 & 6\\
  & * & * & 3 & 4 & 5 \\
  &   & 2 & 7 & 9 & 10 \\
  &   &   & 8 & 11 \\
  &   &   &   & 12
\end{matrix}
\end{eqnarray*}
\caption{A tableau $T$ and its standardization}
\label{fig:PPP1}
\end{figure}
The product $P_{\lambda}P_{\mu}$ contains $10P_{\nu}$.
A tableau $T$ in $\mathrm{Tab}^{\nu/\mu}(\lambda)$ is given by 
the left picture of Figure \ref{fig:PPP1} and its standardization 
is in the right picture.
The reading word  of $\mathrm{stand}(T)$ is 
$w=\mathrm{read}(\mathrm{stand}(T))=(12,8,11,2,7,9,10,3,4,5,1,6)$.
By the RSK insertion, 
we have
\begin{eqnarray*}
(P_{\mathrm{RSK}}(w),Q_{\mathrm{RSK}}(w))
=\left(
\begin{matrix}
1 & 3 & 4 & 5 & 6 \\
2 & 9 & 10 \\
7 & 11 \\
8 \\
12
\end{matrix}\quad,\quad 
\begin{matrix}
1 & 3 & 6 & 7 & 12 \\
2 & 5 & 10 \\
4 & 9 \\
8 \\
11
\end{matrix}
\right).
\end{eqnarray*}
By reversing the RSK insertion, we have 
a word 
\begin{eqnarray*}
w'=\mathrm{RSK}^{-1}(P_{0}(\alpha),Q_{\mathrm{RSK}}(w))
=(12,9,11,1,6,7,
10,2,3,8,4,5).
\end{eqnarray*}
A pair of standard tableaux $(\mathrm{stand}(A),\mathrm{stand}(B))$ 
is given by 
\begin{eqnarray*}
(\mathrm{stand}(A),\mathrm{stand}(B))=
\left(
\begin{matrix}
1 & 3 & 4 & 5 & 6 \\
2 & 9 & 10 \\
7 & 11 \\
8 \\
12
\end{matrix}\quad,\quad
\begin{matrix}
* & * & * & * & * & 4 & 5\\
  & * & * & 2 & 3 & 8 \\
  &   & 1 & 6 & 7 & 10 \\
  &   &   & 9 & 11 \\
  &   &   &   & 12
\end{matrix}
\right),
\end{eqnarray*}
which yields a pair of tableaux $(A,B)$:
\begin{eqnarray*}
(A,B)=\left(
\begin{matrix}
1' & 1 & 1 & 1 & 1 \\
1 & 2 & 2 \\
2' & 3' \\
2 \\
3
\end{matrix}\quad,\quad
\begin{matrix}
* & * & * & * & * & 1 & 1\\
  & * & * & 1 & 1 & 2 \\
  &   & 1 & 2 & 2 & 3 \\
  &   &   & 3 & 4 \\
  &   &   &   & 5
\end{matrix}
\right)
\in\mathrm{Tab}^{\alpha}(\lambda)\times\mathrm{Tab}^{\nu/\mu}(\alpha).
\end{eqnarray*} 
\end{example}

\begin{proof}[Proof of Theorem \ref{thrm:PPP1}]
From Theorem \ref{thrm:StemPP}, we have 
$d_{\lambda\mu}^{\nu}=|\mathrm{Tab}^{\nu/\mu}(\lambda)|$. 
From Lemma \ref{lemma:PPP1}, we have 
$|\mathrm{Tab}^{\nu/\mu}(\lambda)|
=\sum_{\alpha}|\mathrm{Tab}^{\alpha}(\lambda)|\cdot|\mathrm{Tab}^{\nu/\mu}(\alpha)|$,
which implies Theorem \ref{thrm:PPP1}.
\end{proof}

The coefficient $d_{\lambda\mu}^{\nu}$ can be expressed in terms of 
semistandard increasing decomposition tableaux 
$\epsilon^{+}(\lambda)$ and $\epsilon^{+}(\mu)$ as follows.

Let $T$ and $U$ be a semistandard increasing decomposition tableaux of 
shape $\epsilon^{+}(\lambda)$ and $\epsilon^{+}(\mu)$. 
We denote by $w_{\lambda}:=\mathrm{read}(\epsilon^{+}(T))$, 
$w_{\mu}:=\mathrm{read}(\epsilon^{+}(U))$.
Let $\mathrm{Word}(\mu,\lambda)$ be the set 
\begin{eqnarray*}
\mathrm{Word}(\mu,\lambda):=
\left\{
(\epsilon^{+}(T),\epsilon^{+}(U))\ \vert \ 
w_{\mu}\ast w_{\lambda} \text{ is a shifted Yamanouchi word of content } \nu 
\right\}.
\end{eqnarray*}

\begin{theorem}
\label{thrm:PPPword}
We have 
\begin{eqnarray*}
d_{\lambda\mu}^{\nu}=|\mathrm{Word}(\mu,\lambda)|.
\end{eqnarray*}
\end{theorem}

We introduce a lemma used in a proof of Theorem \ref{thrm:PPPword}.
Recall that $\mathrm{Tab}^{\nu/\mu}(\lambda)$ is the set of 
shifted tableaux of skew shape $\nu/\mu$ which have 
an LRS word of content $\lambda$.
\begin{lemma}
\label{lemma:PPPword}
There exists a bijection 
$\kappa:\mathrm{Tab}^{\nu/\mu}(\lambda)\rightarrow 
\mathrm{Word}(\mu,\lambda), T\mapsto(w_{1},w_{2})$.
\end{lemma}
\begin{proof}
A map $\kappa:\mathrm{Tab}^{\nu/\mu}(\lambda)\rightarrow\mathrm{Word}(\mu,\lambda)$ 
is given by the following three steps.
\paragraph{Step 1}
We perform a standardization on a tableau $T$.
Then, we obtain a word $v_{1}$ as the inverse permutation of 
the reading word of $\mathrm{stand}(T)$, {\it i.e.},
$v_{1}:=(\mathrm{read}(\mathrm{stand}(T)))^{-1}$.

\paragraph{Step 2}
Let $\epsilon_{0}^{+}(\lambda)$ be a tableau of shape $\epsilon^{+}(\lambda)$ 
such that the boxes are enumerated by $1,2,\ldots,|\lambda|$ from left to 
right in row starting from the bottom row to top.
We slide down the columns of $\epsilon_{0}^{+}(\lambda)$ such that 
the lowest box in a column is placed in the same hight as the unique 
box in the leftmost column of $\epsilon_{0}^{+}(\lambda)$. 
We denote by $Q$ the obtained tableau.
By turning $Q$ upside down and moving its rows to fit the shifted 
tableau $\lambda$,
we obtain a standard tableau $Q_{0}(\lambda)$.
For example, $Q_{0}(\lambda)$ for $\lambda=(6,3,1)$ is given by
\begin{eqnarray*}
\epsilon^{+}_{0}(\lambda)=
\begin{matrix}
 &  &  &  &  & 10 \\
 &  &  &  & 8 & 9 \\
 &  &  & 5 & 6 & 7 \\
 &  & 3 & 4 \\
 & 2 \\
1 
\end{matrix}
\quad\rightarrow\quad 
\begin{matrix}
  & & & & & 10 \\
 & & & 5 & 8 & 9 \\
1 & 2 & 3 & 4 & 6 & 7
\end{matrix}
\quad\rightarrow\quad
Q_{0}(\lambda)=
\begin{matrix}
1 & 2 & 3 & 4 & 6 & 7 \\
  & 5 & 8 & 9 \\
  &   & 10
\end{matrix}.
\end{eqnarray*}
From $v_{1}$ and $Q_{0}(\lambda)$, we have a pair of 
shifted tableaux $(P_{\mathrm{mix}}(v_{1}),Q_{0}(\lambda))$.
The tableau $P_{\mathrm{mix}}(v_{1})$ may have primed letters 
but $Q_{0}(\lambda)$ does not.
By reversing the mixed insertion, we obtain a word 
$v'_{2}:=\mathrm{mixRSK}^{-1}(P_{\mathrm{mix}}(v_{1}),Q_{0}(\lambda))$.
We define $v_{2}$ as the inverse permutation of $v'_{2}$.

\paragraph{Step 3}
We put letters in a tableau $U'$ of shape $\nu/\mu$ such that the reading 
word $\mathrm{read}(U')$ is equal to $v_{2}$.
We construct a tableau $U$ of shape $\nu$ from $U'$ by enumerating 
boxes in $\mu\subseteq\nu$ by $1,2,\ldots,|\mu|$ as $Q_{0}(\mu)$ 
and boxes in $\nu/\mu$ by 
$|\mu|+1,\ldots,|\nu|$ according to the letters in $U'$, {\it i.e.},
a letter in the region $\nu/\mu$ in $U$ is a letter in $U'$ plus $|\mu|$.
Let $P_{0}(\nu)$ be a unique tableau of shape $\nu$ and of content $\nu$.
Then, by reversing the mixed insertion, we obtain a word 
$w:=\mathrm{mixRSK}^{-1}(P_{0}(\nu),U)$.
The word is written as a concatenation of two words $w=w_{1}\ast w_{2}$ 
where $w_{1}$ (resp. $w_{2}$) is of length $|\mu|$ (resp. $|\lambda|$).
This gives a map $\kappa:T\mapsto(w_{1},w_{2})$.

We claim the following with respect to the map $\kappa$.
\begin{enumerate}
\item[(S8)] The word $v_{2}$ is compatible with the shape $\nu/\mu$.
\item[(S9)] The words $w_{1}$ and $w_{2}$ are compatible with the shapes 
$\epsilon^{+}(\mu)$ and $\epsilon^{+}(\lambda)$.
\end{enumerate}

For (S8), let $\xi\in\mathbb{Z}^{l(\nu)}$ be a sequence of non-negative integers 
defined by $\xi:=\nu-\mu$ where $\mu_{i}=0$ for $l(\mu)<i\le l(\nu)$.
We enumerate the boxes in $\nu/\mu$ by $1,2\ldots,|\lambda|$ from left 
to right in a row starting from the bottom row to top. 
We denote by $b_{i}$ the box in $\nu/\mu$ labelled $i$. 
The $i$-th row (from top) of the shape $\nu/\mu$ has $\xi_{i}$ boxes.
Since $\mathrm{stand}(T)$ is a standard tableau,
the word $v_{1}$ has the following properties: 1) the letters from 
$\sum_{j=i+1}^{l(\nu)}\xi_{j}+1$ to $\sum_{j=i}^{l(\nu)}\xi_{j}$ 
for $1\le i\le l(\nu)$ form an increasing sequence in $v_{1}$, and 
2) if the boxes $b_{i}$ and $b_{j}$ ($i<j$) are in the same column 
in $\nu/\mu$, the letter $j$ is left to the letter $i$ in $v_{1}$.
To show that $v_{2}$ is compatible with the shape $\nu/\mu$, 
it is enough to show that $v'_{2}$ satisfies the properties 
1) and 2) (replace $v_{1}$ with $v'_{2}$). 
We first show that $v'_{2}$ satisfies the property 2).
By definition, two words $v_{1}$ and $v_{2}$ have the same 
mixed insertion tableau.
From Theorem \ref{thrm:splactic}, we have $v_{1}\sim v_{2}$.
Suppose that the box $b_{i+1}$ is just above the box
$b_{i}$ in $\nu/\mu$.
The letter $i+1$ is left to $i$ in the word $v_{1}$.
In the plactic relations, there is no relation which 
exchanges the letters $i$ and $i+1$. 
Therefore, the letter $i+1$ is left to the letter $i$ 
even in the word $v'_{2}$. 
Suppose that the boxes $b_{i+l}, b_{i+l+1},\ldots,b_{i+2l-1}$
are just above the boxes $b_{i},b_{i+1},\ldots,b_{i+l-1}$ 
in $\nu/\mu$.
The letter $i+j+l$ is left to the letter $i+j$ for $0\le j\le l-1$
in the word $v_{1}$.
If the letter $i+j$ is left to the letter $i+j+l$ in $v'_{2}$, 
we have to exchange $i+j+l$ and $i+j$ in $v_{1}$ by a plactic
relation.
Since a plactic relation is applied to a word of length four, it is enough 
to consider a word $\tilde{w}$ of length four including $i+j+l$ and $i+j$ 
such that $i+j$ and $i+j+l$ are next to each other in $\tilde{w}$.
We have four cases:
\begin{enumerate}[(a)]
\item $j<l-1$ and $\tilde{w}$ contains $i+j+1$,
\item $j<l-1$ and $\tilde{w}$ does not contain $i+j+1$, 
\item $j=l-1$ and $\tilde{w}$ contains $i+j+l-1$,
\item $j=l-1$ and $\tilde{w}$ does not contain $i+j+l-1$.
\end{enumerate}
For case (a), since $i+j+l+1$ is right to $i+j+l$ and left to $i+j+1$,
$\tilde{w}=(i+j+l)(i+j)(i+j+l+1)(i+j+1)$. The word $\tilde{w}$ is 
locally equivalent to a word $cadb$ with $a<b<c<d$. 
None of the plactic relations exchanges $a$ and $c$. 
For case (b), we have two cases: (i) $\widetilde{w}$ contains 
$i+j+l-1$ and (ii) $\widetilde{w}$ does not contain $i+j+l-1$.
For case (b-i), since the letter $i+j-1$ is right to $i+j+l-1$ and 
left to $i+j$, $\tilde{w}=(i+j+l-1)(i+j-1)(i+j+l)(i+j)$.
Locally, $\tilde{w}$ is equivalent to $cadb$ with $a<b<c<d$. 
There is no plactic relation which exchanges $b$ and $d$.
For case (b-ii), observe that the word $\tilde{w}$ neither contains 
$i+j+1$ nor $i+j+l-1$. If $\tilde{w}$ is formed by four letters 
$a,b,c$ and $d$ with $a<b<c<d$, the letters $i+j$ and $i+j+l$ 
form a partial word $ab$, $bc$ or $cd$ in $\tilde{w}$.
However, there is no plactic relation which exchanges $ab$, $bc$ or 
$cd$.
By a similar argument, one can show that there is no plactic relation 
which exchanges $i+j$ and $i+j+l$ for cases (c) and (d).
Summarizing above observations, the letter $i+j+l$ is left to 
the letter $i+j$ in the word $v'_{2}$.  
This implies that $v'_{2}$ satisfies the property 2).
By a similar argument, one can show that $v'_{2}$ satisfies 
the property 1). Thus, the statement (S8) is true.

For the statement (S9), observe that if two letters $i$ and $i+1$ 
are in the same row in $\epsilon^{+}_{0}(\lambda)$, 
$i$ and $i+1$ form a horizontal strip in $Q_{0}(\lambda)$.
Similarly, if the letters $i_{1}<i_{2}<\ldots<i_{r}$ are in the same column 
in $\epsilon^{+}_{0}(\lambda)$, these letters form a vertical strip in
$Q_{0}(\lambda)$.
Since the reading word of the tableau $U'$ is $v_{2}$, these properties 
are hold by $U'$.
The tableau $U$ is divided into two regions $\mu$ and $\nu/\mu$.
Both regions have the same properties as $Q_{0}(\mu)$ and $U'$
(or equivalently $Q_{0}(\lambda)$). 
Finally, since the word $w$ is obtained by reversing the mixed insertion with 
the recording tableau $U$, the word $w_{1}$ and $w_{2}$ fit to the 
shapes $\epsilon^{+}(\mu)$ and $\epsilon^{+}(\lambda)$.
Thus, the statement (S9) is true.

From the construction, the map $\kappa$ has the inverse. 
By a similar argument discussed above, it is easy to show that 
the map $\kappa^{-1}$ is well-defined.
This completes the proof.

\end{proof}

\begin{example}
Let $\lambda=(5,3)$, $\mu=(3,1)$ and $\nu=(6,4,2)$.
The product of $P_{\lambda}P_{\mu}$ contains $4P_{\nu}$.
An example of $T$ in $\mathrm{Tab}^{\nu/\mu}(\lambda)$ 
is given by 
\begin{eqnarray*}
T=\begin{matrix}
\ast & \ast & \ast & 1' & 1 & 1 \\
     & \ast & 1 & 1 & 2' \\
     &      & 2 & 2
\end{matrix}.
\end{eqnarray*}
Since the word $\mathrm{read}(\mathrm{stand}(T))$ is $78236145$,
we have $v_{1}=63478512$.
Thus, we have a pair of tableaux 
\begin{eqnarray*}
(P_{\mathrm{mix}}(v_1),Q_{0}(\lambda))
:=\left(
\begin{matrix}
1 & 2 & 3' & 6' & 8 \\
  & 4 & 5 & 7' 
\end{matrix}\quad,\quad
\begin{matrix}
1 & 2 & 3 & 5 & 7 \\
  & 4 & 6 & 8
\end{matrix}
\right).
\end{eqnarray*}
By reversing the mixed insertion, we obtain words $v'_{2}$ and $v_{2}$ 
as $v'_{2}=63745182$ and $v_{2}=68245137$.
The tableau $U'$ and $U$ is given from $v_{2}$ by
\begin{eqnarray*}
U'=\begin{matrix}
* & * & * & 1 & 3 & 7 \\
  & * & 2 & 4 & 5 \\
  &   & 6 & 8
\end{matrix},\qquad\qquad
U=\begin{matrix}
1 & 2 & 3 & 5 & 7 & 11 \\
  & 4 & 6 & 8 & 9 \\
  &   & 10 & 12
\end{matrix}.
\end{eqnarray*}
The word $w=w_{1}\ast w_{2}$ is obtained from $P_{0}(\lambda)$ and $U$:
\begin{eqnarray*}
w=\mathrm{mixRSK}^{-1}\left(
\begin{matrix}
1 & 1 & 1 & 1 & 1 & 1 \\
  & 2 & 2 & 2 & 2 \\
  &   & 3 & 3
\end{matrix}\quad,\quad
\begin{matrix}
1 & 2 & 3 & 5 & 7 & 11 \\
  & 4 & 6 & 8 & 9 \\
  &   & 10 & 12
\end{matrix}
\right)
=1221\ast 31312121.
\end{eqnarray*}
Note that $w$ is a shifted Yamanouchi word and the word $w_1$ (resp. $w_2$)  
is compatible with the shape $\epsilon^{+}(\mu)$ (resp. $\epsilon^{+}(\lambda)$). 
\end{example}

\begin{proof}[Proof of Theorem \ref{thrm:PPPword}]
From Theorem \ref{thrm:StemPP}, we have 
$d_{\lambda\mu}^{\nu}=|\mathrm{Tab}^{\nu/\mu}(\lambda)|$. 
From Lemma \ref{lemma:PPPword}, we have 
$|\mathrm{Tab}^{\nu/\mu}(\lambda)|=|\mathrm{Word}(\lambda,\mu)|$,
which implies Theorem is true.
\end{proof}

\section{\texorpdfstring{$\hat{S}$}{S}-functions and 
products of \texorpdfstring{$P$}{P}-functions}

\subsection{\texorpdfstring{$\hat{S}$}{S}-function in 
terms of products of \texorpdfstring{$P$}{P}-functions}
\label{sec:SinPP}

Let $\alpha$ be an ordinary partition and $\lambda,\mu$ be strict partitions 
such that $\alpha=\lambda\otimes\mu$.
We regard a strict partition as a set of positive integers.
Let $A=(A_1,A_2,\ldots):=(\lambda\cup\mu)\setminus(\lambda\cap\mu)$ be a decreasing integer
sequence.
We assign a sign to an element of $A$: $\mathrm{sign}(A_{i}):=(-1)^{i-1}$.

We define a pair of strict partitions $(\lambda',\mu')$ from $(\lambda,\mu)$ as follows.
As sets of positive integers, $\lambda'$ and $\mu'$ satisfy 
$\lambda'\subseteq\lambda$, $\mu'\supseteq\mu$, 
$\lambda'\cup\mu'=\lambda\cup\mu$ and $\lambda'\cap\mu'=\lambda\cap\mu$.
We denote by $S_{1}(\lambda,\mu)$ the set of pairs of strict partitions $(\lambda',\mu')$.
We define the sets 
\begin{eqnarray*}
S^{+}_{1}(\lambda,\mu)
&:=&\{(\lambda',\mu')\in S_{1}(\lambda,\mu) \ |\  |\mu'\setminus\mu|\equiv0\pmod{2}\}, \\
S^{-}_{1}(\lambda,\mu)&:=&S_{1}(\lambda,\mu)\setminus S^{+}_{1}(\lambda,\mu).
\end{eqnarray*}
Similarly, let $S_{2}(\lambda,\mu)$ be the set of pairs of strict partitions $(\lambda'',\mu'')$
such that $\lambda''\supseteq\lambda$, $\mu''\subseteq\mu$, 
$\lambda''\cup\mu''=\lambda\cup\mu$ and $\lambda''\cap\mu''=\lambda\cap\mu$.

When $(\lambda',\mu')\in S_{1}(\lambda,\mu)$ or $(\lambda'',\mu'')\in S_{2}(\lambda,\mu)$, 
we define two signs $\mathrm{sign}_{1}(\lambda',\mu';\lambda,\mu)$ and 
$\mathrm{sign}_{2}(\lambda'',\mu'';\lambda,\mu)$ as 
\begin{eqnarray*}
\mathrm{sign}_{1}(\lambda',\mu';\lambda,\mu)
&:=&\prod_{a\in A\cap(\lambda\setminus\lambda')}\mathrm{sign}(a), \\
\mathrm{sign}_{2}(\lambda'',\mu'';\lambda,\mu)&:=&(-1)^{|\mu\setminus\mu''|}
\prod_{a\in A\cap(\mu\setminus\mu'')}\mathrm{sign}(a).
\end{eqnarray*}

\begin{theorem}
\label{thrm:SinPP}
We have 
\begin{eqnarray}
\label{eqn:SinPP1}
\hat{S}_{\alpha}=2^{l(\mu)}\sum_{(\lambda',\mu')\in S_{1}(\lambda,\mu)}
\mathrm{sign}_{1}(\lambda', \mu';\lambda,\mu)P_{\lambda'}P_{\mu'},
\end{eqnarray}
for $l(\lambda)-l(\mu)=0$ or $1$  and 
\begin{eqnarray}
\label{eqn:SinPP2}
\begin{aligned}
\hat{S}_{\alpha}&=&2^{l(\lambda)}\sum_{(\lambda',\mu')\in S^{+}_{1}(\lambda,\mu)}
\mathrm{sign}_{1}(\lambda', \mu';\lambda,\mu)P_{\lambda'}P_{\mu'} \\
&=&2^{l(\lambda)}
\sum_{(\lambda',\mu')\in S^{-}_{1}(\lambda,\mu)}
\mathrm{sign}_{1}(\lambda', \mu';\lambda,\mu)P_{\lambda'}P_{\mu'}
\end{aligned}
\end{eqnarray}
for $l(\lambda)=l(\mu)+1$.
We also have
\begin{eqnarray}
\label{eqn:SinPP3}
\hat{S}_{\alpha}=2^{l(\lambda)}\sum_{(\lambda'',\mu'')\in S_{2}(\lambda,\mu)}
\mathrm{sign}_{2}(\lambda'', \mu'';\lambda,\mu)P_{\lambda''}P_{\mu''}.
\end{eqnarray}
for $l(\lambda)-l(\mu)=0$ or $1$.
\end{theorem}

\begin{example}
Let $\lambda:=(5,4,2)$ and $\mu:=(3,1)$. 
Then, we have $\alpha=\lambda\otimes\mu=(4,3,3,3,2)$.
The $\hat{S}_{\alpha}$ is written in terms of $P$-functions in three ways as  
\begin{eqnarray*}
\hat{S}_{\alpha}
&=&2^{3}(P_{(5,4,2)}P_{(3,1)}+P_{(5)}P_{(4,3,2,1)}-P_{(4)}P_{(5,3,2,1)}
-P_{(2)}P_{(5,4,3,1)}) \\
&=&2^{3}(P_{(4,2)}P_{(5,3,1)}-P_{(5,2)}P_{(4,3,1)}-P_{(5,4)}P_{(3,2,1)}+P_{(5,4,3,2,1)}) \\
&=&
2^{3}(P_{(5,4,2)}P_{(3,1)}-P_{(5,4,3,2)}P_{(1)}-P_{(5,4,2,1)}P_{(3)}+P_{(5,4,3,2,1)}) 
\end{eqnarray*}
The first and second expressions are Eqn. (\ref{eqn:SinPP2}) and they imply
Eqn. (\ref{eqn:SinPP1}).
The third expression comes from Eqn. (\ref{eqn:SinPP3}).
\end{example}

We will prove Theorem \ref{thrm:SinPP} by induction.
Before we move to a proof of Theorem \ref{thrm:SinPP}, we introduce two 
lemmas needed later.

\begin{lemma}
\label{lemma:SinPP1}
Suppose that $\alpha:=\lambda\otimes\mu$ with $l(\lambda)=l(\mu)+1$ and 
$\beta:=\sigma\otimes\rho$ with $l(\sigma)=l(\rho)$. 
The function $\hat{S}_{\alpha}$ satisfies Eqn.(\ref{eqn:SinPP2}) if 
Eqn.(\ref{eqn:SinPP1}) is true for all $\beta$'s satisfying $|\beta|<|\alpha|$.
\end{lemma}
\begin{proof}

Let $\widehat{\lambda}_{i}$ be a strict partition obtained by 
deleting $\lambda_{i}$ from $\lambda$.
Recall that $\alpha^{T}$ is the conjugate partition of $\alpha$.
A $\hat{S}$-function satisfies $\hat{S}_{\alpha}=\hat{S}_{\alpha^{T}}$  
and $S_{\alpha}^{T}$ has a determinant expression (\ref{Sdet1}).
By expanding a determinant with respect to the first column,
we have 
\begin{eqnarray}
\label{eqn:SinPP5}
\begin{aligned}
\hat{S}_{\alpha}&=\hat{S}_{\alpha^{T}} \\
&=\sum_{i=1}^{l(\lambda)}(-1)^{i-1}
q_{\lambda_{i}}\cdot\hat{S}_{\mu\otimes\widehat{\lambda}_{i}} \\
&=2^{l(\lambda)}
\sum_{i=1}^{l(\lambda)}\sum_{(\lambda',\mu')\in S_{1}(\mu,\widehat{\lambda}_{i})}
(-1)^{i-1}\mathrm{sign}_{1}(\lambda',\mu';\mu,\widehat{\lambda}_{i})
P_{\lambda_{i}}P_{\lambda'}P_{\mu'}
\end{aligned}
\end{eqnarray}
Since $P_{\lambda}$ has a expression by a Pfaffian (see Eqn. (\ref{eqn:Ppf})), we have 
\begin{eqnarray}
\label{eqn:SinPP6}
\sum_{i=1}^{l(\lambda)}(-1)^{i-1}P_{\lambda_{i}}P_{\widehat{\lambda}_{i}}
=
\left\{
\begin{matrix}
0, & \text {for } l(\lambda)\equiv0\pmod{2}\\
P_{\lambda}, & \text{ for } l(\lambda)\equiv1\pmod{2}
\end{matrix}
\right..
\end{eqnarray}
By Substituting Eqn.(\ref{eqn:SinPP6}) to a product $P_{\lambda_{i}}P_{\mu'}$ in 
Eqn.(\ref{eqn:SinPP5}) and rearranging the terms,
we obtain Eqn.(\ref{eqn:SinPP2}).
\end{proof}

\begin{lemma}
\label{lemma:SinPP2}
Suppose that $\alpha:=\lambda\otimes\mu$ with $l(\lambda)=l(\mu)$ and 
$\beta=\sigma\otimes\rho$ with $l(\sigma)=l(\rho)$.
The function $\hat{S}_{\alpha}$ satisfies Eqn.(\ref{eqn:SinPP1}) 
if Eqn.(\ref{eqn:SinPP1}) is true for all $\beta$ such that 
$|\beta|<|\alpha|$ or $\mu_{l(\mu)}>\rho_{l(\rho)}$ for 
$|\alpha|=|\beta|$ with $l(\mu)=l(\rho)$. 
\end{lemma}
\begin{proof}
Let $\gamma:=\lambda'\otimes\mu'$ with $l(\lambda')=l(\mu')+1$.
Since Eqn. (\ref{eqn:SinPP1}) is true for all $\beta$ satisfying 
$|\beta|<|\alpha|$, by Lemma \ref{lemma:SinPP1}, 
Eqn. (\ref{eqn:SinPP2}) is true for all 
$\gamma$ with $|\gamma|\le|\alpha|$.
We have two cases: (a) $\mu_{l(\mu)}=1$ and (b) $\mu_{l(\mu)}\ge2$. 

For (a), suppose that $\mu_{l(\mu)}=1$.
We denote by $\widehat{\mu}$ a strict partition obtained from $\mu$ 
by deleting $\mu_{l(\mu)}$ and 
by $\widehat{\alpha}:=\lambda\otimes\widehat{\mu}$.
Since $|\widehat{\alpha}|=|\alpha|-1$ and $l(\lambda)=l(\widehat{\mu})+1$,
the function $\hat{S}_{\widehat{\alpha}}$ satisfies Eqn. (\ref{eqn:SinPP2}).
We consider the product of $\hat{S}_{\widehat{\alpha}}$ and $\hat{S}_{(1)}$.
A product of two $\hat{S}$-functions is expressed in terms of the 
Littlewood--Richardson rule (see Eqn. (\ref{eqn:LRSS})), 
\begin{eqnarray}
\label{eqn:SinPP7}
\hat{S}_{\widehat{\alpha}}\hat{S}_{(1)}
=\hat{S}_{\alpha}+\sum_{\alpha'\neq\alpha}\hat{S}_{\alpha'}
\end{eqnarray}
where the sum is taken over all $\alpha'$'s such that $\alpha'/\widehat{\alpha}$ is 
a single box and $\alpha'\neq\alpha$.
Note that multiplicity of a $\hat{S}$-function is one.
We have $\hat{S}_{(1)}=2P_{(1)}$
On the other hand, a $\hat{S}$-function is expressed as a sum of 
products of $P$-functions.
Recall that a product $P_{\lambda}P_{(1)}$ is expressed by the 
Littlewood--Richardson--Stembridge rule (see Theorem \ref{thrm:StemPP}), 
namely 
\begin{eqnarray*}
P_{\lambda}P_{(1)}=\sum_{\lambda'}P_{\lambda'}
\end{eqnarray*}
where the sum is taken over all $\lambda'$'s such that $\lambda'/\lambda$ 
is a single box.
Thus, a product of $P_{\lambda}P_{\mu}$ and $\hat{S}_{(1)}$ is expressed 
as 
\begin{eqnarray}
\label{eqn:SPPinP}
\begin{aligned}
\hat{S}_{(1)}P_{\lambda}P_{\mu}&=2P_{(1)}P_{\lambda}P_{\mu} \\
&=(P_{(1)}P_{\lambda})P_{\mu}+P_{\lambda}(P_{(1)}P_{\mu}) \\
&=\sum_{\lambda'}P_{\lambda'}P_{\mu}+\sum_{\mu'}P_{\lambda}P_{\mu'},
\end{aligned}
\end{eqnarray}
where the skew shapes $\lambda'/\lambda$ and $\mu'/\mu$ are a single box.
In Eqn. (\ref{eqn:SinPP7}), $\hat{S}_{\alpha'}$ with $\alpha'\neq\alpha$
satisfies Eqn. (\ref{eqn:SinPP1}) or Eqn. (\ref{eqn:SinPP2}) by 
the assumption.
We multiply $\hat{S}_{(1)}$ by both sides of Eqn. (\ref{eqn:SinPP2}) 
for $\widehat{\alpha}$ and rearrange the terms for $\hat{S}_{\alpha'}$.
By a direct calculation, one can show that $\hat{S}_{\alpha}$ satisfies 
Eqn. (\ref{eqn:SinPP1}).

For (b), we have $\mu_{l(\mu)}\ge2$.
We denote by $\widehat{\mu}$ a strict partition obtained from $\mu$ 
by replacing $\mu_{l(\mu)}$ with $\mu_{l(\mu)}-1$ and 
define $\widehat{\alpha}:=\lambda\otimes\widehat{\mu}$.
We consider the product $\hat{S}_{\widehat{\alpha}}\hat{S}_{(1)}$.
By a similar argument to case (a), 
one can show that $\hat{S}_{\alpha}$ satisfies Eqn. (\ref{eqn:SinPP1}).
\end{proof}

Given an ordinary partition $\alpha$, we have $\alpha=\lambda\otimes\mu$.
Let $l:=l(\lambda)$ if $l(\lambda)=l(\mu)+1$ and $l:=l(\lambda)+1$ if
$l(\lambda)=l(\mu)$.
If $l(\lambda)=l(\mu)+1$, we regard $\mu$ as a strict partition of length
$l$ by adding zero to $\mu$.
Similarly, if $l(\lambda)=l(\mu)$, we regard $\lambda$ and $\mu$ as strict 
partitions of length $l$ by adding zero to $\lambda$ and $\mu$.

We place integers in $\lambda\cup\mu$ in a decreasing order from left to right.
When an integer $n$ is in $\lambda\cap\mu$, we put two $n$'s next to each other 
and denote by $n_{\lambda}$ (resp. $n_{\mu}$) the left (resp. right) $n$. 
The index of $n_{\lambda}$ stands for $n\in\lambda$.
We construct a perfect matching of length $2l$ by connecting two integers
via an arc.
A perfect matching satisfies the following conditions.
\begin{enumerate}
\item We connect two integers $n$ and $m$ if and only if $n\in\lambda$ 
and $m\in\mu$.
\item We do not connect the same integers $n_{\lambda}$ and $n_{\mu}$ for $n\ge1$. 
\end{enumerate}
Note that two integers $0_{\lambda}$ and $0_{\mu}$ can be connected by an arc.
Thus, a perfect matching characterizes a permutation $\pi\in\mathcal{S}_{l}$ 
since an arc connects two integers $\lambda_{i}$ and $\mu_{\pi(i)}$.
We denote by $\mathrm{cr}(\lambda,\mu;\pi)$ the number of crossings 
in a perfect matching $\pi$. 
We assign a sign for an element of the sequence $s:=s_{1}\ldots s_{2l}$ 
of integers $\lambda\cup\mu$. 
Here, the multiplicity of $n$ with $n\in\lambda\cap\mu$ is two 
in the sequence $s$.
The sign $\mathrm{sign}'(s_{i})$ is defined as $(-1)^{i-1}$.
We define a sign of a perfect matching $\pi$ as
\begin{eqnarray*}
\mathrm{sign}_{3}(\lambda,\mu;\pi):=
(-1)^{\mathrm{cr}(\lambda,\mu;\pi)}
\cdot\prod_{s_{i}\in\lambda}\mathrm{sign}'(s_{i})
\end{eqnarray*}

\begin{example}
\label{example:pm}
Let $\lambda=(5,4,3)$ and $\mu=(4,2)$.
We have four perfect matchings:
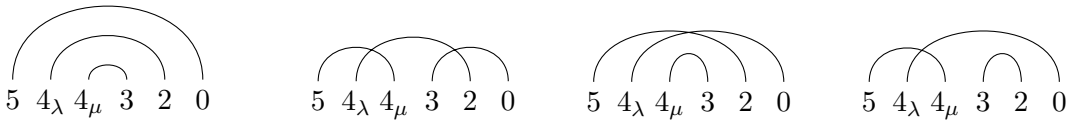
\begin{eqnarray*}
\tikzpic{-0.5}{
\draw(0,0)node[anchor=north]{$5$}..controls(0,1.3)and(2.5,1.3)..(2.5,0)node[anchor=north]{$0$}
     (0.5,0)node[anchor=north]{$4_{\lambda}$}
           ..controls(0.5,0.78)and(2,0.78)..(2,0)node[anchor=north]{$2$}
     (1,0)node[anchor=north]{$4_{\mu}$}
           ..controls(1,0.26)and(1.5,0.26)..(1.5,0)node[anchor=north]{$3$};
}\qquad
\tikzpic{-0.7}{
\draw(0,0)node[anchor=north]{$5$}
           ..controls(0,0.6)and(1,0.6)..(1,0)node[anchor=north]{$4_{\mu}$}
     (0.5,0)node[anchor=north]{$4_{\lambda}$}
           ..controls(0.5,0.78)and(2,0.78)..(2,0)node[anchor=north]{$2$}
     (1.5,0)node[anchor=north]{$3$}
           ..controls(1.5,0.6)and(2.5,0.6)..(2.5,0)node[anchor=north]{$0$};
}\quad
\tikzpic{-0.65}{
\draw(0,0)node[anchor=north]{$5$}
           ..controls(0,0.9)and(2,0.9)..(2,0)node[anchor=north]{$2$}
     (0.5,0)node[anchor=north]{$4_{\lambda}$}
           ..controls(0.5,0.9)and(2.5,0.9)..(2.5,0)node[anchor=north]{$0$}
     (1,0)node[anchor=north]{$4_{\mu}$}
           ..controls(1,0.5)and(1.5,0.5)..(1.5,0)node[anchor=north]{$3$};
}\quad
\tikzpic{-0.65}{
\draw(0,0)node[anchor=north]{$5$}
           ..controls(0,0.6)and(1,0.6)..(1,0)node[anchor=north]{$4_{\mu}$}
     (0.5,0)node[anchor=north]{$4_{\lambda}$}
           ..controls(0.5,0.9)and(2.5,0.9)..(2.5,0)node[anchor=north]{$0$}
     (1.5,0)node[anchor=north]{$3$}
           ..controls(1.5,0.5)and(2,0.5)..(2,0)node[anchor=north]{$2$};
} 
\end{eqnarray*}
The numbers of crossings are $0,2,1$ and $1$ from left to right. 
The first two perfect matchings have sign plus and the last two have minus. 
\end{example}

We denote by $P[m,n]$ the $P$-function $P_{(m,n)}$ for simplicity.
We define $P[m,n]:=P[n,m]$ for $m<n$.
By definition, note that $P[n,n]=0$ for $n\ge1$ and $P[0,0]=1$.
\begin{theorem}
\label{thrm:SinPP2}
Suppose that $\hat{S}_{\alpha}$ satisfies Eqn.(\ref{eqn:SinPP1}) or 
Eqn.(\ref{eqn:SinPP2}).
Then, a $\hat{S}$-function is expressed in terms of $P$-functions as 
\begin{eqnarray}
\label{eqn:SinPPP}
\hat{S}_{\alpha}
=2^{l(\lambda)}
\sum_{\pi\in\mathcal{S}_{l}}
\mathrm{sign}_{3}(\lambda,\mu;\pi)
\prod_{i=1}^{l} P[\lambda_{i},\mu_{\pi(i)}]
\end{eqnarray}
where $\mathcal{S}_{l}$ is the symmetric group of order $l$.
\end{theorem}
\begin{proof}
From equations (\ref{eqn:SinPP1}) and (\ref{eqn:SinPP2}), 
we have an expression of $\hat{S}_{\alpha}$ in terms of 
products of two $P$-functions.
Since a $P$-function is written in terms of a Pfaffian, 
$\hat{S}_{\alpha}$ is expressed as a sum of products of 
$P$-functions $P_{\lambda}$'s where $\lambda$ is of length 
two.
A term in $\hat{S}_{\alpha}$ is the from $P_{\lambda'}P_{\mu'}$ and 
$\mu'\supseteq\mu$.
Note that $\lambda=\lambda'\cup(\mu'\setminus\mu)$.
We expand $P_{\mu'}$ in terms of $P$-functions of length two or one 
according to the length of $\mu'$.
By a direct calculation using an expansion formula for a Pfaffian,
one can easily show that the coefficient of $P[\mu_{i},\mu_{j}]$ is zero.
Thus, $\hat{S}_{\alpha}$ is a sum of products of $P$-functions 
$P[\lambda_{i},\mu_{\pi(i)}]$ for some $\pi\in\mathcal{S}_{l}$.
By an expansion of a Pfaffian, we have the coefficient of 
$\prod_{i=1}^{l}P[\lambda_{i},\mu_{\pi(i)}]$ is one except for 
the overall factor. 
By a direct calculation, one can show that the sign 
of $\prod_{i=1}^{l}P[\lambda_{i},\mu_{\pi(i)}]$ is 
given by $\mathrm{sign}_{3}(\lambda,\mu;\pi)$.
\end{proof}

\begin{example}
Let $\lambda=(5,4,3)$ and $\mu=(4,2)$.
Then, $\alpha:=\lambda\otimes\mu=(5,4,3,3,3)$.
\begin{eqnarray}
\label{eqn:SinPP8}
\begin{aligned}
\hat{S}_{\alpha}&=2^{3}(P_{(5,4,3)}P_{(4,2)}-P_{(4)}P_{(5,4,3,2)}) \\
&=2^{3}(P_{(5)}P_{(4,2)}P_{(4,3)}+P_{(5,4)}P_{(4,2)}P_{(3)}
-P_{(5,2)}P_{(4)}P_{(4,3)}-P_{(5,4)}P_{(4)}P_{(3,2)}).
\end{aligned}
\end{eqnarray}
Note that the terms in the last line in Eqn. (\ref{eqn:SinPP8}) 
correspond to the perfect matchings in Example \ref{example:pm}.

\end{example}

\begin{proof}[Proof of Theorem \ref{thrm:SinPP}]
We prove Theorem by induction.
For $\alpha=(1)$, it is obvious that $\hat{S}_{\alpha}=2P_{(1)}$ and 
Eqn. (\ref{eqn:SinPP1}) is true.

Let $\beta:=\sigma\otimes\rho$.
Suppose that Equations (\ref{eqn:SinPP1}) and (\ref{eqn:SinPP2}) are 
true for all $\beta$'s such that $|\beta|<|\alpha|$ or 
$\mu_{l(\mu)}>\rho_{l(\rho)}$ for $|\alpha|=|\beta|$ with $l(\mu)=l(\rho)$.
From Lemmas \ref{lemma:SinPP1} and \ref{lemma:SinPP2}, 
the function $\hat{S}_{\alpha}$ satisfies either Eqn. (\ref{eqn:SinPP1}) 
or Eqn. (\ref{eqn:SinPP2}).

If we expand $P$-functions in Eqn. (\ref{eqn:SinPP3}) by using an expansion 
formula for a Pfaffian, we obtain the same expression 
as Theorem \ref{thrm:SinPP2}.
This implies that Eqn. (\ref{eqn:SinPP3}) is true. 
\end{proof}

From Theorem \ref{thrm:SinPP}, we have the following corollary.
\begin{cor}[J\'ozefiak and Pragacz \cite{JozPra91}]
Let $\alpha$ be shift-symmetric, {\it i.e.}, $\alpha:=\lambda\otimes\lambda$.
Then, $\hat{S}_{\alpha}$ is given by the square of $P_{\lambda}$:
\begin{eqnarray*}
\hat{S}_{\alpha}=2^{l(\lambda)}P_{\lambda}^{2}.
\end{eqnarray*}
\end{cor}

\subsection{A product of \texorpdfstring{$P$}{P}-functions 
in terms of \texorpdfstring{$\hat{S}$}{S}-functions}

Theorem \ref{thrm:SinPP} shows that a $\hat{S}$-function can be expressed 
as a sum of products of two Schur $P$-functions.
By solving the relation reversely, one can show that a product of 
Schur $P$-function can be expressed in terms of a sum of Schur $\hat{S}$-function.
This expression does not have positivity, {\it i.e.}, a sign of the coefficient 
of a $\hat{S}$-function can be minus.
However, the expression is multiplicity free except for the overall factor ,{\it i.e.}, 
the coefficient of a $\hat{S}$-function is either $1$ or $-1$.

Fix strict partitions $\lambda$ and $\mu$.
Let $\lambda_{0}$ and $\mu_{0}$ be sets of positive integers  
such that $\lambda_0:=\lambda\cup\mu$ and $\mu_{0}:=\lambda\cap\mu$.
We define a set $S_{3}(\lambda,\mu)$ by
\begin{eqnarray*}
S_{3}(\lambda,\mu):=\{
\nu\  |\  \nu\subseteq\lambda_{0}\setminus\mu_{0}
\text{ and } l(\nu)=\lfloor(|\lambda_{0}|-|\mu_{0}|)/2 \rfloor
\}.
\end{eqnarray*}
We define a sign of $\nu\subseteq\lambda_{0}\setminus\mu_{0}$
by 
\begin{eqnarray*}
\mathrm{sign}_{4}(\nu):=
\prod_{m\in\nu}(-1)^{d(m)}
\end{eqnarray*}
where
\begin{eqnarray*}
d(m):=\#\{ l\in\mathbb{N} \ | \ 
l\in\lambda_{0}\setminus\mu_{0} \text{ and } l\ge m \}.
\end{eqnarray*}
We construct an ordinary partition $\alpha(\nu)$ from 
$\nu\in S_{3}(\lambda,\mu)$ by defining 
\begin{eqnarray*}
\alpha(\nu):=(\lambda_{0}\setminus\nu)\otimes(\mu_{0}\cup\nu).
\end{eqnarray*}
Suppose that $\lambda$ and $\mu$ are written by 
$\chi\subseteq\lambda_{0}\setminus\mu_{0}$ as 
$\lambda=\lambda_{0}\setminus\chi$ and $\mu=\mu_{0}\cup\chi$.
Finally, we define a sign with respect to $\lambda,\mu$ and $\alpha(\nu)$
by 
\begin{eqnarray*}
\mathrm{sign}_{5}(\lambda,\mu;\alpha(\nu))
:=(-1)^{l(\chi)+l(\chi\cap\nu)}\cdot\mathrm{sign}_{4}(\nu)
\cdot\mathrm{sign}_{4}(\chi).
\end{eqnarray*}

\begin{theorem}
We have 
\begin{eqnarray}
\label{eqn:PPinS}
P_{\lambda}P_{\mu}
=2^{-(l(\lambda)+l(\mu)-l(\lambda\cap\mu))}
\sum_{\nu\in S_{3}(\lambda,\mu)}
\mathrm{sign}_{5}(\lambda,\mu;\alpha(\nu))
\cdot \hat{S}_{\alpha(\nu)}.
\end{eqnarray}
\end{theorem}
\begin{proof}
Suppose that $\lambda$ and $\mu$ is written as 
$\lambda=\lambda_{0}\setminus\chi$ and $\mu=\mu_{0}\cup\chi$ 
where $\lambda_{0}=\lambda\cup\mu$ and $\mu_{0}=\lambda\cap\mu$.
Then, $l(\lambda)+l(\mu)-l(\lambda\cap\mu)=l(\lambda_{0})$.
By substituting Eqn. (\ref{eqn:SinPPP}) into the right hand side 
of Eqn. (\ref{eqn:PPinS}), we obtain 
\begin{eqnarray}
\label{eqn:PPinS2}
\qquad
2^{-l(\nu)}\sum_{\nu\in S_{3}(\lambda,\mu)}\sum_{\pi\in\mathcal{S}_{l}}
\mathrm{sign}_{5}(\lambda,\mu;\alpha(\nu))
\mathrm{sign}_{3}(\lambda_{0}\setminus\nu,\mu_{0}\cup\nu;\pi)
\prod_{i=1}^{l}P[(\lambda_{0}\setminus\nu)_{i},(\mu_{0}\cup\nu)_{\pi(i)}].
\end{eqnarray}
Note that in the expression of 
$\mathrm{sign}_{3}(\lambda_{0}\setminus\nu,\mu_{0}\cup\nu;\pi)$ we have 
\begin{eqnarray}
\prod_{s_{i}\in\lambda_{0}\setminus\nu}\mathrm{sign}'(s_{i})
=
\prod_{s_{i}\in\lambda_{0}}\mathrm{sign}'(s_{i})
\prod_{s_{i}\in\nu}\mathrm{sign}'(s_{i}).
\end{eqnarray}
We denote 
\begin{eqnarray*}
\mathrm{Sign}(\lambda,\mu;\nu)
=
\mathrm{sign}_{3}(\lambda_{0}\setminus\nu,\mu_{0}\cup\nu;\pi)
\mathrm{sign}_{5}(\lambda,\mu;\alpha(\nu)).
\end{eqnarray*}

We will show that Eqn. (\ref{eqn:PPinS2}) is equal to $P_{\lambda}P_{\mu}$.
We first show that the coefficient of $P[\lambda_{i},\mu_{j}]$ in
Eqn. (\ref{eqn:PPinS2}) is zero.

Suppose that $\lambda_{i},\mu_{j}\notin\mu_{0}$, 
$\lambda_{i}\in\lambda_{0}\setminus\nu$ and 
$\mu_{j}\in\mu_{0}\cup\nu$.
Then, we have $\lambda_{i}\notin\nu$, $\mu_{j}\in\nu$, 
$\lambda_{i}\notin\chi$ and $\mu_{j}\in\chi$.
There exists $\nu'$ such that $l(\nu\cap\nu')=l(\nu)-1$, 
$\lambda_{i}\in\nu'$ and 
$\mu_{j}\notin\nu'$.
We have $\mu_{j}\subset(\chi\cap\nu)$, $\mu_{j}\not\subset(\chi\cap\nu')$
and $((\chi\cap\nu)\setminus\{\mu_{j}\})=\chi\cap\nu'$.
This implies that $l(\chi\cap\nu)=l(\chi\cap\nu')+1$.
If $\mathrm{sign}'(\lambda_{i})=\mathrm{sign}'(\mu_{j})$, 
we have $\mathrm{sign}_{4}(\nu)=\mathrm{sign}_{4}(\nu')$, 
$\prod_{s_{i}\in\nu}\mathrm{sign}'(s_{i})=\prod_{s_{i}\in\nu'}\mathrm{sign}'(s_{i})$
and 
$\mathrm{cr}(\lambda_{0}\setminus\nu,\mu_{0}\cup\nu;\pi)
=\mathrm{cr}(\lambda_{0}\setminus\nu',\mu_{0}\cup\nu';\pi)$.
Here, $\lambda_{i}$ and $\mu_{j}$ is connected by an arc in 
the perfect matching $\pi$.
Thus we have 
$\mathrm{Sign}(\lambda,\mu;\nu)=-\mathrm{Sign}(\lambda,\mu;\nu')$.
If $\mathrm{sign}'(\lambda_{i})\neq\mathrm{sign}'(\mu_{j})$,
we have $\mathrm{sign}_{4}(\nu)=-\mathrm{sign}_{4}(\nu')$. 
We also have 
$\prod_{s_{i}\in\nu}\mathrm{sign}'(s_{i})
=-\prod_{s_{i}\in\nu'}\mathrm{sign}'(s_{i})$
and 
$\mathrm{cr}(\lambda_{0}\setminus\nu,\mu_{0}\cup\nu;\pi)
=\mathrm{cr}(\lambda_{0}\setminus\nu',\mu_{0}\cup\nu';\pi)$.
Thus we have $\mathrm{Sign}(\lambda,\mu;\nu)=-\mathrm{Sign}(\lambda,\mu;\nu')$.
The coefficient of $P[\lambda_{i},\mu_{j}]$ in Eqn. (\ref{eqn:PPinS2}) is 
zero since the contributions from $\nu$ and $\nu'$ cancel each other. 

The observation above implies that 
Eqn. (\ref{eqn:PPinS2}) contains only the terms $P[\lambda_{i},\lambda_{j}]$
and $P[\mu_{i},\mu_{j}]$.
We will show that there exits $\nu'$ such that 
$\prod_{i=1}^{l}P[(\lambda\setminus\nu)_{i},(\mu_{0}\cup\nu)_{\pi(i)}]
=\prod_{i=1}^{l}P[(\lambda\setminus\nu')_{i},(\mu_{0}\cup\nu')_{\pi(i)}]$
and $l(\nu\cap\nu')=l(\nu)-1$ and 
$\mathrm{Sign}(\lambda,\mu;\nu)=\mathrm{Sign}(\lambda,\mu;\nu')$ 
if $\lambda_{i}$ or $\lambda_{j}$ is in $\nu$.
We consider the coefficient of $P[\lambda_i,\lambda_j]$ since one can apply the same 
argument to $P[\mu_{i},\mu_{j}]$ by the symmetry between 
$\lambda$ and $\mu$.

Suppose that $\lambda_{i},\lambda_{j}\notin\mu_{0}$,
$\lambda_{i}\in\lambda_{0}\setminus\nu$ and $\lambda_{j}\in\mu_{0}\cup\nu$.
Then, we have $\lambda_{i}\notin\nu$, $\lambda_{j}\in\nu$ and 
$\lambda_{i},\lambda_{j}\notin\chi$.
There exists a unique $\nu'$ such that $l(\nu\cap\nu')=l(\nu)-1$, 
$\lambda_{i}\in\nu'$ and $\lambda_{j}\notin\nu'$.
Since $\lambda_{i},\lambda_{j}\notin\chi$, we have 
$l(\chi\cap\nu)=l(\chi\cap\nu')$.
Since $\lambda_{i}$ and $\lambda_{j}$ are connected by an arc 
in perfect matchings corresponding to $\nu$ and $\nu'$,
we have $\mathrm{cr}(\lambda_{0}\setminus\nu,\mu_{0}\cup\nu;\pi)
=\mathrm{cr}(\lambda_{0}\setminus\nu',\mu_{0}\cup\nu';\pi)$.
If $\mathrm{sign}'(\lambda_{i})=\mathrm{sign}'(\lambda_{j})$, 
we have $\mathrm{sign}_{4}(\nu)=\mathrm{sign}_{4}(\nu')$
and $\prod_{s_{i}\in\nu}\mathrm{sign}'(s_{i})
=\prod_{s_{i}\in\nu'}\mathrm{sign}'(s_{i})$.
Thus we have $\mathrm{Sign}(\lambda,\mu;\nu)=\mathrm{Sign}(\lambda,\mu;\nu')$.
Similarly, if $\mathrm{sign}'(\lambda_{i})\neq\mathrm{sign}'(\lambda_{j})$,
we have $\mathrm{sign}_{4}(\nu)=-\mathrm{sign}_{4}(\nu')$
and $\prod_{s_{i}\in\nu}\mathrm{sign}'(s_{i})
=-\prod_{s_{i}\in\nu'}\mathrm{sign}'(s_{i})$,
which implies 
$\mathrm{Sign}(\lambda,\mu;\nu)=\mathrm{Sign}(\lambda,\mu;\nu')$.
The sets $\nu$ and $\nu'$ give the same contribution in Eqn. (\ref{eqn:PPinS2}),
which gives the coefficient two as an overall factor.

Suppose that $\lambda_{i}\in\mu_{0}$ and $\lambda_{j}\notin\mu_{0}$.
Then, $\lambda_{i}\notin\nu$ and $\lambda_{i},\lambda_{j}\notin\chi$.
We first consider the case where $\lambda_{i}\in\lambda_{0}\setminus\nu$ 
and $\lambda_{j}\in\mu_{0}\cup\nu$.
Then, $\lambda_{j}\in\nu$.
We define a sequence of a pair of integers 
$\alpha_{a}:=(\alpha_{a,1},\alpha_{a,2})$, $1\le a\le p$ for 
some positive integer $p$, such that it satisfies the following 
five conditions:
\begin{enumerate}
\item $\alpha_{1}:=(\lambda_{i},\lambda_{j})$, 
\item $\alpha_{a+1,2}=\alpha_{a,1}\in\mu_{0}$ for $1\le a\le p-1$, 
\item $\alpha_{a,2}\neq \alpha_{b,2}$ if $b\neq a$,
\item $\alpha_{p,1}\notin\mu_{0}$,
\item $\alpha_{a,1}\in\lambda_{0}\setminus\nu$ and 
$\alpha_{a,2}\in\mu_{0}\cup\nu$ for $1\le a\le p$.
\end{enumerate}
From condition (2), we have $\alpha_{a,1}\notin\nu$ for $1\le a\le p$.
From (4) and (5), we have $\alpha_{p,1}\notin\nu$.
There exists $\nu'$ such that $\alpha_{p,1}\in\nu'$, $\lambda_{j}\notin\nu'$
and $l(\nu\cap\nu')=l(\nu)-1$.
This $\nu'$ is characterized by a sequence of pairs of positive 
integers $\beta_{a}$, $1\le a\le p$, where 
$\beta_{a}=(\alpha_{a,2},\alpha_{a,1})$.
We have two cases: (a) $\alpha_{p,1}=\lambda_{k}$ and 
(b) $\alpha_{p,1}=\mu_{k}$ for some integer $k$.
For (a), observe that $\lambda_{j},\lambda_{k}\notin\chi$ and 
$\lambda_{k}\notin\mu_{0}$.
By a similar argument above, we have 
$\mathrm{Sign}(\lambda,\mu;\nu)=\mathrm{Sign}(\lambda,\mu;\nu')$.
For (b), we have $\mu_{k}\notin\mu_{0}$ and $\mu_{k}\in\chi$.
If $\mathrm{sign}'(\lambda_{j})=\mathrm{sign}'(\mu_{k})$, 
we have $l(\chi\cap\nu)=l(\chi\cap\nu')-1$, 
$\mathrm{cr}(\lambda_{0}\setminus\nu,\mu_{0}\cup\nu;\pi)=
\mathrm{cr}(\lambda_{0}\setminus\nu',\mu_{0}\cup\nu';\pi)-1$,
and $\mathrm{sign}_{4}(\nu)=\mathrm{sign}_{4}(\nu')$.
These imply that 
$\mathrm{Sign}(\lambda,\mu;\nu)=\mathrm{Sign}(\lambda,\mu;\nu')$.
If $\mathrm{sign}'(\lambda_{j})\neq\mathrm{sign}'(\mu_{k})$,
we have $l(\chi\cap\nu)=l(\chi\cap\nu')-1$, 
$\mathrm{cr}(\lambda_{0}\setminus\nu,\mu_{0}\cup\nu;\pi)=
\mathrm{cr}(\lambda_{0}\setminus\nu',\mu_{0}\cup\nu';\pi)-1$,
$\mathrm{sign}_{4}(\nu)=-\mathrm{sign}_{4}(\nu')$ 
and $\prod_{s_{i}\in\nu}\mathrm{sign}'(s_{i})
=-\prod_{s_{i}\in\nu'}\mathrm{sign}'(s_{i})$.
Thus we have $\mathrm{Sign}(\lambda,\mu;\nu)=\mathrm{Sign}(\lambda,\mu;\nu')$.
In case of $\lambda_{j}\in\lambda_{0}\setminus\nu$ and 
$\lambda_{i}\in\mu_{0}\cup\nu$, one can similarly define a sequence $\alpha_{a}$
and show that there exists $\nu'$ such that 
$\mathrm{Sign}(\lambda,\mu;\nu)=\mathrm{Sign}(\lambda,\mu;\nu')$ and 
$\prod_{i=1}^{l}P[(\lambda\setminus\nu)_{i},(\mu_{0}\cup\nu)_{\pi(i)}]
=\prod_{i=1}^{l}P[(\lambda\setminus\nu')_{i},(\mu_{0}\cup\nu')_{\pi(i)}]$.
Therefore, the sets $\nu$ and $\nu'$ give the overall factor two.

Suppose that $\lambda_{i},\lambda_{j}\in\mu_{0}$. 
Then, $\lambda_{i},\lambda_{j}\notin\nu$ and 
$\lambda_{i},\lambda_{j}\notin\chi$.
We define a sequence of a pair of integers $\alpha_{a}:=(\alpha_{a,1},\alpha_{a,2})$, 
$1\le a\le 2p$ for some integer $p$, such that it satisfies the following four conditions:
\begin{enumerate} 
\item $\alpha_{1}:=(\lambda_{i},\lambda_{j})$, 
\item $\alpha_{a,1},\alpha_{a,2}\in\mu_{0}$, 
\item $\alpha_{a+1,1}=\alpha_{a,2}$ for $1\le a\le p$ with $\alpha_{p+1,1}:=\alpha_{1,1}$,
\item $\alpha_{a,1}\neq\alpha_{b,1}$ if $a\neq b$.
\end{enumerate}
We define $\beta_{a}$, $1\le a\le 2p$, by $\beta_{a}:=(\alpha_{a,2},\alpha_{a,1})$.
A sequence $\alpha_{a}$, $1\le a\le 2p$, corresponds to 
a product $\prod_{i=1}^{2p}P[a_{i,1},a_{i,2}]$ and so does $\beta$.
However, we have a freedom to choose $a_{i,1}\in\lambda$ or $\mu$.
Thus we have two-to-two bijection the choice of $\alpha$ or $\beta$ 
and the choice of $\lambda$ or $\mu$.
Note that these give the same contribution in Eqn. (\ref{eqn:PPinS2}) but 
the corresponding perfect matchings are different.

By summarizing the observations above and taking care of the overall factor, 
Eqn. (\ref{eqn:PPinS2}) is rewritten as 
\begin{eqnarray*}
\sum_{\rho\in\mathcal{F}_{n},\sigma\in\mathcal{F}_{m}}
\mathrm{Sign}'(\rho,\sigma)
\prod_{i=1}^{n}P[\lambda_{\rho(2i-1)},\lambda_{\rho(2i)}]
\prod_{i=1}^{m}P[\mu_{\sigma(2i-1)},\mu_{\sigma(2i)}],
\end{eqnarray*}
where $n=\lfloor(l(\lambda)+1)/2\rfloor$ and 
$m=\lfloor(l(\mu)+1)/2\rfloor$.
We will show that $\mathrm{Sign}'(\rho,\sigma)=\epsilon(\rho)\epsilon(\sigma)$
where $\epsilon(\rho)$ (resp. $\epsilon(\sigma)$) is a sign for the permutation 
$\rho$ (resp. $\sigma$).
Since we have a symmetry between $\lambda$ and $\mu$ in $P_{\lambda}P_{\mu}$,
we set $\lambda_{1}$ is the largest integer in $\lambda\cup\mu$ without 
loss of generality.
Given $\lambda$ and $\mu$, we define 
$\overline{\lambda}:=\lambda\setminus\{\lambda_{1},\lambda_{r}\}$, 
$\overline{\lambda}_{0}:=\overline{\lambda}\cup\mu$ 
and $\overline{\nu}:=\nu\setminus\{\lambda_{1},\lambda_{r}\}$
for some $r$. 
This corresponds to considering the perfect matchings such that  
$\lambda_{1}$ and $\lambda_{r}$ are connected by an arc.
We define $\overline{\chi}$ such that 
$\overline{\lambda}=\overline{\lambda}_{0}\setminus\overline{\chi}$
and $\mu=\overline{\lambda}_{0}\cup\overline{\chi}$.
Since we consider a perfect matching $\pi$ which contains an arc connecting 
$\lambda_{1}$ and $\lambda_{r}$, we denote $\overline{\pi}$ by a perfect matching
obtained from $\pi$ by deleting this arc.
For simplicity, we write 
$\mathrm{cr}(\nu;\pi)
:=\mathrm{cr}(\lambda_{0}\setminus\nu,\mu_{0}\cup\nu;\pi)$. 
Since $P_{\lambda}$ has an expression in terms of Pfaffian, 
to show $\mathrm{Sign}'(\rho,\sigma)=\epsilon(\rho)\epsilon(\sigma)$ 
is equivalent to show   
$\mathrm{Sign}(\lambda,\mu;\nu)
=(-1)^{r}\mathrm{Sign}(\overline{\lambda},\overline{\mu};\overline{\nu})$.

Suppose that 
\begin{eqnarray*}
\lambda_{1}\ge\mu_{1,1}>\cdots>\mu_{1,k_{1}}>\lambda_{2}\ge
\mu_{2,1}>\cdots>\mu_{1,k_{2}}>\lambda_{3}\ge\ldots 
>\mu_{r-1,k_{r-1}}>\lambda_{r}\ge\mu_{r,1}.
\end{eqnarray*}
We assume that $a$ $\lambda_{i}$'s ($2\le i<r$) are in $\mu_{0}$,  
$b$ $\mu_{i,j}$'s are in $\nu$ and $c$ $\lambda_{k}$'s ($2\le k<r$) 
are in $\nu$.
We define $k:=\sum_{i=1}^{r-1}k_{i}$.
We have four cases: (1) $\lambda_{1},\lambda_{r}\notin\mu_{0}$,
(2) $\lambda_{1}\in\mu_{0}$ and $\lambda_{r}\notin\mu_{0}$,
(3) $\lambda_{1}\notin\mu_{0}$ and $\lambda_{r}\in\mu_{0}$,
and (4) $\lambda_{1},\lambda_{r}\in\mu_{0}$.
We prove $\mathrm{Sign}(\lambda,\mu;\nu)
=(-1)^{r}\mathrm{Sign}(\overline{\lambda},\overline{\mu};\overline{\nu})$ 
for only case (1), since other cases can be proven in a similar way.

For case (1), we have two cases: 
(a) $\lambda_{1}\notin\nu$ and $\lambda_{r}\in\nu$, and 
(b) $\lambda_{1}\in\nu$ and $\lambda_{r}\notin\nu$.
For case (a), we have 
\begin{eqnarray}
\mathrm{sign}_{4}(\nu)&=&(-1)^{b}(-1)^{c}(-1)^{k+r}\mathrm{sign}_{4}(\overline{\nu}), \\
\label{eqn:signchi}
\mathrm{sign}_{4}(\chi)&=&(-1)^{k-a}\mathrm{sign}_{4}(\overline{\chi}), \\
(-1)^{\mathrm{cr}(\nu;\pi)}
&=&(-1)^{k+r-2}
(-1)^{\mathrm{cr}(\overline{\nu};\overline{\pi})}, \\
\label{eqn:signlam0}
\prod_{s_{i}\in\lambda_{0}}\mathrm{sign}'(s_{i})
&=&(-1)^{k+r-2a-2}(-1)^{a}(-1)^{k+r-1}
\prod_{s_{i}\in\overline{\lambda}_{0}}\mathrm{sign}'(s_{i}), \\
\prod_{s_{i}\in\nu}\mathrm{sign}'(s_{i})
&=&(-1)^{b}(-1)^{c}(-1)^{k+r-1}
\prod_{s_{i}\in\overline{\nu}}\mathrm{sign}'(s_{i}), \\
\label{eqn:lchi}
l(\chi)&=&l(\overline{\chi}) \\
\label{eqn:lnuchi}
l(\nu\cap\chi)&=&l(\overline{\nu}\cap\overline{\chi}).
\end{eqnarray}
Therefore, we obtain 
$\mathrm{Sign}(\lambda,\mu;\nu)=
(-1)^{r}\mathrm{Sign}(\overline{\lambda},\overline{\mu};\overline{\nu})$.
For case (b), we have 
\begin{eqnarray*}
\mathrm{sign}_{4}(\nu)&=&(-1)^{b}(-1)^{c}(-1)\mathrm{sign}_{4}(\overline{\nu}), \\
\prod_{s_{i}\in\nu}\mathrm{sign}'(s_{i})
&=&(-1)^{b}(-1)^{c}
\prod_{s_{i}\in\overline{\nu}}\mathrm{sign}'(s_{i})
\end{eqnarray*}
and the same Equations from (\ref{eqn:signchi}) to (\ref{eqn:signlam0}), 
(\ref{eqn:lchi}) and (\ref{eqn:lnuchi}).
Thus, we have $\mathrm{Sign}(\lambda,\mu;\nu)=
(-1)^{r}\mathrm{Sign}(\overline{\lambda},\overline{\mu};\overline{\nu})$.

Note that we have $l(\overline{\chi})=l(\chi)+1$ in case 
of (2) and (3) and $l(\overline{\chi})=l(\chi)+2$ in case of (4).
This completes the proof.
\end{proof}

\begin{example}
Let $\lambda:=(4,1)$ and $\mu:=(3,2)$.
We have 
\begin{eqnarray*}
P_{(4,1)}P_{(3,2)}
&=&2^{-4}(-\hat{S}_{(4,3)\otimes(2,1)}+\hat{S}_{(4,2)\otimes(3,1)}
+\hat{S}_{(4,1)\otimes(3,2)} \\
&&
+\hat{S}_{(3,2)\otimes(4,1)}+\hat{S}_{(3,1)\otimes(4,2)}-\hat{S}_{(2,1)\otimes(4,3)}) \\
&=&2^{-4}(
-\hat{S}_{(3,3,2,2)}+\hat{S}_{(4,3,2,1)}+\hat{S}_{(4,4,1,1)}+\hat{S}_{(5,3,2)}
+\hat{S}_{(5,4,1)}-\hat{S}_{(5,5)}
).
\end{eqnarray*}
\end{example}

\subsection{Giambelli formula}
\label{Giambelli}
A Schur function $s_{\alpha}$ has a Giambelli formula with determinant, and 
$Q$-function has a Giambelli formula in terms of a Pfaffian.
We show two types of Giambelli formula for $\hat{S}$-function: one is 
a determinant and the other is a Pfaffian.

Let $n$ be the length of $\lambda$. 
When $l(\mu)=n-1$, we regard $\mu$ as a strict partition of length $n$ 
by $\mu_{n}=0$.
When $l(\mu)=n$, we regard $\lambda,\mu$ as a strict partition of 
length $n+1$ by $\lambda_{n+1}=\mu_{n+1}=0$.
We define a matrix $\tilde{P}(\lambda,\mu)$ in terms of Schur $P$-functions as 
\begin{eqnarray*}
\tilde{P}(\lambda,\mu)_{i,j}
:=
\begin{cases}
P_{(\lambda_i,\mu_j)}, & \lambda_{i}>\mu_{j}>0, \\
-P_{(\mu_j,\lambda_i)}, & \mu_{j}>\lambda_{i}>0, \\
P_{(\lambda_i)}, & \lambda_{i}\neq0, \mu_{j}=0,\\
-P_{(\mu_j)}, & \lambda_{i}=0, \mu_{j}\neq0,\\
\delta_{\lambda_{i},0}, & \lambda_{i}=\mu_{j}.
\end{cases}
\end{eqnarray*}

\begin{theorem}
Let $\alpha=\lambda\otimes\mu$.
Then,  
\begin{eqnarray*}
\hat{S}_{\alpha}=
2^{n}\det\left[\tilde{P}(\lambda,\mu)_{i,j} \right]_{1\le i,j\le m}.
\end{eqnarray*}
where $m=n$ for $l(\lambda)=l(\mu)+1=n$ and $m=n+1$ for $l(\lambda)=l(\mu)=n$.
\end{theorem}
\begin{proof}
From the definition of $\tilde{P}(\lambda)$, we have 
\begin{eqnarray*}
\det\left[
\tilde{P}(\lambda,\mu)_{i,j}
\right]
=\sum_{\pi\in\mathcal{S}_{m}}\epsilon(\pi)
\prod_{i=1}^{m}(-1)^{\delta(\mu_{\pi(i)}>\lambda_{i})}
\prod_{i=1}^{m}P[\lambda_{i},\mu_{\pi(i)}].
\end{eqnarray*}
where $\delta(S)=1$ if the statement $S$ is true and zero otherwise.
From Theorem \ref{thrm:SinPP2}, it is enough to show that 
\begin{eqnarray}
\label{eqn:sign3ep}
\mathrm{sign}_{3}(\lambda,\mu;\pi)=
\epsilon(\pi)\prod_{i=1}^{m}(-1)^{\delta(\mu_{\pi(i)}>\lambda_{i})}.
\end{eqnarray}
First, we show Eqn. (\ref{eqn:sign3ep}) is true for $\pi=\mathrm{id}$ by induction.
It is easy to see that Eqn. (\ref{eqn:sign3ep}) is true for $l(\lambda)=1$.
Let $\lambda'=\lambda\setminus\{\lambda_{1}\}$ and 
$\mu'=\mu\setminus\{\mu_{1}\}$. 
If $\lambda_1=\mu_1$, then $P[\lambda_1,\mu_1]=0$ by definition.
Therefore, we have two cases: (1) $\lambda_1>\mu_{1}$ and (2) $\lambda_{1}<\mu_{1}$.
For (1), suppose that $\lambda_{1}>\lambda_{2}>\ldots>\lambda_{k+1}>\mu_{1}$.
Then, 
\begin{eqnarray*}
(-1)^{\mathrm{cr}(\lambda,\mu;\mathrm{id})}
&=&(-1)^{k}\cdot(-1)^{\mathrm{cr}(\lambda',\mu';\mathrm{id})} \\
\prod_{s_{i}\lambda}\mathrm{sign}'(s_{i})
&=&(-1)^{k}\prod_{s_{i}\lambda}\mathrm{sign}'(s_{i}).
\end{eqnarray*}
Thus, we have 
$\mathrm{sign}_{3}(\lambda,\mu;\mathrm{id})
=\mathrm{sign}_{3}(\lambda',\mu';\mathrm{id})$
and the right hand side of Eqn. (\ref{eqn:sign3ep}) is equal to 
$\prod_{i=2}^{m}(-1)^{\delta(\mu_{i}>\lambda_{i})}$.
By induction assumption, we have 
$\mathrm{sign}_{3}(\lambda',\mu';\mathrm{id})
=\prod_{i=2}^{m}(-1)^{\delta(\mu_{i}>\lambda_{i})}$. 
Thus, Eqn. (\ref{eqn:sign3ep}) is true for general $\lambda$ and $\mu$.
For case (2), by a similar argument, 
we have 
$\mathrm{sign}_{3}(\lambda,\mu;\mathrm{id})
=-\mathrm{sign}_{3}(\lambda',\mu';\mathrm{id})$ 
which is compatible with the right hand side of Eqn. (\ref{eqn:sign3ep}).
We successively apply the argument above to $\lambda$ and $\mu$, 
which implies that Eqn. (\ref{eqn:sign3ep}) is true for $\pi=\mathrm{id}$.

Given $\pi$, we denote $\pi'=(i,i+1)\circ\pi$ for some $i$ where $(i,i+1)$ 
is a transposition in $\mathcal{S}_{l}$.
Suppose that we have arcs connecting $\lambda_{i}$ with $\mu_{\pi(i)}$
and $\lambda_{i+1}$ with $\mu_{\pi(i+1)}$ in a perfect matching.
Applying the transposition on the perfect matching means 
that we connect $\lambda_{i}$ with $\mu_{\pi(i+1)}$ and 
$\lambda_{i+1}$ with $\mu_{\pi(i)}$.
Let $j:=\min\{\pi(i),\pi(i+1)\}$ and $k:=\max\{\pi(i),\pi(i+1)\}$.
We have the cases where
$\lambda_{i}$ and $\lambda_{i+1}$ are distinct and so do $\mu_{j}$ and $\mu_{k}$.
We have six local configurations for 
$\lambda_{i},\lambda_{i+1}, \mu_{j}$ and $\mu_{k}$: 
1) $\lambda_{i}>\lambda_{i+1}\ge\mu_{j}>\mu_{k}$,
2) $\lambda_{i}\ge\mu_{j}\ge\lambda_{i+1}\ge\mu_{k}$,
3) $\lambda_{i}\ge\mu_{j}>\mu_{k}\ge\lambda_{i+1}$, 
4) $\mu_{j}\ge\lambda_{i}>\lambda_{i+1}\ge\mu_{k}$,
5) $\mu_{j}\ge\lambda_{i}\ge\mu_{k}\ge\lambda_{i+1}$, and 
6) $\mu_{j}>\mu_{k}\ge\lambda_{i}>\lambda_{i+1}$.
For case 1), the reconnection of arcs is locally given by
\begin{eqnarray*}
\tikzpic{-0.5}{
\draw(0,0)node[anchor=north]{$\lambda_{i+1}$}
      ..controls(0,0.4)and(0.6,0.4)..(0.6,0)node[anchor=north]{$\mu_{j}$};
\draw(-0.6,0)node[anchor=north]{$\lambda_{i}$}
      ..controls(-0.6,0.8)and(1.2,0.8)..(1.2,0)node[anchor=north]{$\mu_{k}$};
}\leftrightarrow
\tikzpic{-0.5}{
\draw(0,0)node[anchor=north]{$\lambda_{i+1}$}
      ..controls(0,0.6)and(1.2,0.6)..(1.2,0)node[anchor=north]{$\mu_{j}$};
\draw(-0.6,0)node[anchor=north]{$\lambda_{i}$}
      ..controls(-0.6,0.6)and(0.6,0.6)..(0.6,0)node[anchor=north]{$\mu_{k}$};
}.
\end{eqnarray*}
It is easy to see that the both sides of Eqn. (\ref{eqn:sign3ep}) gives 
$(-1)$ by the reconnection.
Similarly, we have a factor $(-1)$ for the cases 3), 4) and 6) and 
$1$ for the cases 2) and 5).
Starting from $\pi=\mathrm{id}$, by successively applying the procedure above,
we obtain that Eqn. (\ref{eqn:sign3ep}) is true for any $\pi$.
This completes the proof.
\end{proof}

\begin{example}
Let $\alpha:=(4,3,3,3)=(4,3,2)\otimes(3,1)$. 
The matrix $\tilde{P}:=\tilde{P}((4,3,2),(3,1))$ is given by 
\begin{eqnarray*}
\tilde{P}=
\begin{bmatrix}
P[(4,3)] & P[(4,1)] & P[(4)] \\
0 & P[(3,1)] & P[(3)] \\
-4P[(3,2)]  & P[(2,1)] & P[(2)]
\end{bmatrix}.
\end{eqnarray*}
We have $\hat{S}_{\alpha}=2^{3}\cdot \det[\tilde{P}]$.
\end{example}

Let $n$ be the length of $\lambda$. 
When $l(\mu)=n-1$, we regard $\mu$ as a strict partition of length $n$ 
by $\mu_{n}=0$.
When $l(\mu)=n$, we regard $\lambda,\mu$ as a strict partition of 
length $n+1$ by defining $\lambda_{n+1}=\mu_{n+1}=0$.
Let $m=n$ when $l(\lambda)=l(\mu)+1$ and $m=n+1$ when $l(\lambda)=l(\mu)$.
Let $A:=\lambda\cup\mu$ be a weakly decreasing integer sequence with 
repeated entries such that  an integer in $\lambda\cap\mu$ 
(resp. $(\lambda\cup\mu)\setminus(\lambda\cap\mu)$ ) appears twice 
(resp. once) in $A$.
We write $A:=(a_{1}\ge a_{2}\ge\ldots\ge a_{2m})$.
For $i\in A$ and $i\notin\lambda\cap\mu$, we define 
the sign of $i$ as $\mathrm{sign}(i)=+$ if $i\in\lambda$ and 
$\mathrm{sign}(i)=-$ if $i\in\mu$.
We denote the sign of $A$ by 
$\mathrm{sign}(A):=(\mathrm{sign}(a_1),\mathrm{sign}(a_2),\ldots,\mathrm{sign}(a_{2m}))$. 
Let $\mathrm{sign}_{0}(A)$ be a sequence of alternating signs, {\it i.e.}, 
$\mathrm{sign}_{0}(A):=(+,-,+,\ldots,-)$.
We denote by $d(A)$ the number of transpositions to obtain $\mathrm{sign}(A)$ from 
$\mathrm{sign}_{0}(A)$.

For $i\in A$ and $i\in\lambda\cap\mu$, we assign $\mathrm{sign}(i)=+$
for the first $i$ and $\mathrm{sign}(i)=-$ for the second $i$.
We define a skew symmetric matrix $P_{i,j}(A)$, $1\le i,j\le 2m$, in terms of  
Schur $P$-functions: 
\begin{eqnarray*}
P_{i,j}(A):=
\begin{cases}
0, &\mathrm{sign}(a_i)=\mathrm{sign}(a_j) \text{ or } a_{i}=a_{j}\neq0, \\
P_{(a_{i},a_{j})}, & a_i>a_j>0, \\
P_{(a_{i})}, & a_i>a_j=0, \\
1, & a_i=a_j=0.
\end{cases}
\end{eqnarray*}
for $1\le i\le j\le 2m$.

\begin{theorem}
Let $A$, $d(A)$ and $P_{i,j}(A)$ as above.
\begin{eqnarray}
\label{eqn:SinPf1}
\hat{S}_{\alpha}=(-1)^{d(A)}\cdot 2^{n}\cdot \mathrm{pf}\left[P_{i,j}(A)\right]_{1\le i,j\le 2m}.
\end{eqnarray}
\end{theorem}
\begin{proof}
Recall that a Pfaffian is defined by Eqn. (\ref{eqn:defPf}) and 
a $\hat{S}$-function has an expression in terms of perfect matchings
as Theorem \ref{thrm:SinPP2}.
We place integers in $[1,2m]$ from left to right.
Then, for a permutation $\pi\in\mathcal{F}_{m}$, we consider 
a perfect matching such that $\pi(2i-1)$ and $\pi(2i)$ 
for $1\le i\le m$ are connected by an arc.
It is obvious that the sign $\epsilon(\pi)$ is equal to 
$(-1)^{\mathrm{cr}}$ where $\mathrm{cr}$ is the number of 
crossings in the perfect matching.
We have 
\begin{eqnarray*}
\mathrm{pf}[P_{i,j}(A)]
=\sum_{\pi\in\mathcal{F}_{m}}(-1)^{\mathrm{cr}(\lambda,\mu;\pi)}
\prod_{i=1}^{m}P[\lambda_{i},\mu_{\pi(i)}],
\end{eqnarray*}
where $\alpha=\lambda\otimes\mu$.
It is easy to see 
\begin{eqnarray*}
\prod_{s_{i}\in\lambda}\mathrm{sign}'(s_{i})=(-1)^{d(A)}.
\end{eqnarray*}
Finally, $n=l(\lambda)$ by definition.
Combining these observations with Theorem \ref{thrm:SinPP2}, 
we obtain Eqn. (\ref{eqn:SinPf1}).
This completes the proof.
\end{proof}

\begin{example}
Let $\alpha:=(4,3,3,3)=(4,3,2)\otimes (3,1)$.
Then, $A=(4,3,3,2,1,0)$ and $\mathrm{sign}(A)=(+,+,-,+,-,-)$.
The skew symmetric matrix $P_{i,j}(A)$ is given by
\begin{eqnarray*}
P_{i,j}(A)=
\begin{bmatrix}
0 & 0 & P[(4,3)] & 0 & P[(4,1)] & P[(4)] \\
0 & 0 & 0 & 0 & P[(3,1)] & P[(3)] \\
-P[(4,3)] & 0 & 0 & P[(3,2)] & 0 & 0 \\
0 & 0 & -P[(3,2)] & 0 & P[(2,1)] & P[(2)] \\
-P[(4,1)] & -P[(3,1)] & 0 & -P[(2,1)] & 0 & 0 \\
-P[(4)] & -P[(3)] & 0 & -P[(2)] & 0 & 0 
\end{bmatrix}.
\end{eqnarray*}
Then, $\hat{S}_{\alpha}=2^{3}\cdot\mathrm{pf}[P_{i,j}(A)]$.
\end{example}

\subsection{Skew \texorpdfstring{$\hat{S}$}{S}-functions}
Given two strict partitions $\lambda$ and $\mu$, 
we define the sets of strict partitions $S_{p}(\lambda,\mu)$, $p=1,2$, and 
signs $\mathrm{sign}_{p}(\lambda',\mu';\lambda,\mu)$, $p=1,2$, as 
in Section \ref{sec:SinPP}.
Let $\alpha$ and $\beta$ be ordinary partitions such that 
$\beta\subseteq\alpha$, $\alpha=\lambda\otimes\mu$ and $\beta=\nu\otimes\xi$.
We define two sets $S(\lambda,\mu,\nu,\xi)$ and $\tilde{S}(\lambda,\mu,\nu,\xi)$ as 
\begin{eqnarray*}
S(\lambda,\mu,\nu,\xi):=
\{(\lambda',\mu',\nu',\xi')\ \vert\  (\lambda',\mu')\in S_{1}(\lambda,\mu) 
\text{ and }(\nu',\xi')\in S_{2}(\nu,\xi) \}, \\
\tilde{S}(\lambda,\mu,\nu,\xi):=
\{(\lambda',\mu',\nu',\xi')\ \vert\ (\lambda',\mu')\in S_{2}(\lambda,\mu) 
\text{ and }(\nu',\xi')\in S_{1}(\nu,\xi) \}.
\end{eqnarray*}
We define $m:=l(\nu)$ (resp. $m:=l(\nu)+1$) if $l(\nu)=l(\xi)+1$ 
(resp. $l(\nu)=l(\xi)$) and $n:=l(\lambda)$ (resp. $n:=l(\lambda)+1$) 
if $l(\lambda)=l(\mu)+1$ (resp. $l(\lambda)=l(\mu)$).

\begin{theorem}
\label{thrm:skewSinQQ}
A $\hat{S}$-function can be expressed as a sum of products of Schur $Q$-functions:
\begin{eqnarray}
\label{eqn:skewSinQQ1}
\hat{S}_{\alpha/\beta}
=2^{l(\nu)-l(\lambda)}\sum_{(\lambda',\mu',\nu',\xi')\in S(\lambda,\mu,\nu,\xi)} 
\mathrm{sign}_{1}(\lambda',\mu';\lambda,\mu)\mathrm{sign}_{2}(\nu',\xi';\nu,\xi)
Q_{\lambda'/\nu'}Q_{\mu'/\xi'},
\end{eqnarray}
and 
\begin{eqnarray}
\label{eqn:skewSinQQ2}
\hat{S}_{\alpha/\beta}
=2^{m-n}\sum_{(\lambda',\mu',\nu',\xi')\in \tilde{S}(\lambda,\mu,\nu,\xi)} 
\mathrm{sign}_{2}(\lambda',\mu';\lambda,\mu)\mathrm{sign}_{1}(\nu',\xi';\nu,\xi)
Q_{\lambda'/\nu'}Q_{\mu'/\xi'}.
\end{eqnarray}
\end{theorem}

\begin{example}
Let $\alpha=(4,3,2,1)$ and $\beta=(2,2)$. Then,
$\alpha=(4,2)\otimes(3,1)$ and $\beta=(2,1)\otimes(1)$, 
$m=2$ and $n=3$.
Thus, we have 
\begin{eqnarray*}
\hat{S}_{\alpha/\beta}
&=&
Q_{(4,2)/(2,1)}Q_{(3,1)/(1)}\\
&=&2^{-1}\left(Q_{(4,2)/(2,1)}Q_{(3,1)/(1)}+Q_{(4,2)/(1)}Q_{(3,1)/(2,1)} \right.\\
&&+Q_{(4,3,2)/(2,1)}Q_{(1)/(1)}+Q_{(4,2,1)/(2,1)}Q_{(3)/(1)}
\left.\right).
\end{eqnarray*}
\end{example}

We denote $Q[\lambda/\mu]:=Q_{\lambda/\mu}$ and $\hat{S}[\alpha/\beta]:=\hat{S}_{\alpha/\beta}$ 
for simplicity.
Given a strict partition $\lambda$, let $\widehat{\lambda}_{i}$ be 
a strict partition obtained from $\lambda$ by deleting $\lambda_{i}$.
\begin{lemma}
\label{lemma:QQ1}
Let $\lambda/\mu$ be a skew partition. Then, 
a skew Schur $Q$-function satisfies
\begin{eqnarray}
\label{eqn:QQ1}
Q[\lambda/\mu]=\sum_{i=1}^{l(\lambda)}
(-1)^{i-1}
Q[(\lambda_{i}-\mu_{1})]
Q[\widehat{\lambda}_{i}/\widehat{\mu}_{1}]
\end{eqnarray}
and 
\begin{eqnarray}
\label{eqn:QQ2}
\sum_{i=1}^{l(\lambda)}
(-1)^{i-1}
Q[(\lambda_{i}-\mu_{1})]
Q[\widehat{\lambda}_{i}/\mu]=0.
\end{eqnarray}
\end{lemma}
\begin{proof}
For Eqn. (\ref{eqn:QQ1}), recall that we have a Pfaffian expression 
of $Q[\lambda/\mu]$ (see Eqn. (\ref{eqn:skewPpf})).
By expanding the Pfaffian with respect to the column corresponding 
to $\mu_{1}$, we obtain Eqn. (\ref{eqn:QQ1}).

For Eqn. (\ref{eqn:QQ2}), we expand $Q[\hat{\lambda}_{i}/\mu]$ by using 
Eqn. (\ref{eqn:QQ1}).
If we regard $\mu$ as a set of integers, we have $\mu_{1}\in\mu$.
Then, the coefficient of $Q[(\lambda_{i}-\mu_{1})]Q[(\lambda_{j}-\mu_{1})]$
is zero for any pair of $i$ and $j$, which implies Eqn. (\ref{eqn:QQ2}).
 
\end{proof}

\begin{proof}[Proof of Theorem \ref{thrm:skewSinQQ}]
From Theorem \ref{thrm:SinPP}, Theorem \ref{thrm:skewSinQQ} holds true 
for $\beta=\emptyset$.
We assume that Eqn. (\ref{eqn:skewSinQQ1}) is true for all $\alpha'$
with $|\alpha'|<|\alpha|$. 
 
Recall that $\hat{S}[\alpha/\beta]$ has a determinant expression and 
$\hat{S}[\alpha/\beta]=\hat{S}[\alpha^T/\beta^T]$ by the conjugate partitions.
Then, we have 
\begin{eqnarray*}
\hat{S}_{\alpha/\beta}&=&\hat{S}_{\alpha^T/\beta^T} \\
&=&\sum_{i=1}^{l(\lambda)}
(-1)^{i-1}Q[\lambda_{i}/\nu_{1}]\cdot
\hat{S}[\mu\otimes\widehat{\lambda}_{i}/\xi\otimes\widehat{\nu}_{1}].
\end{eqnarray*}
Note that if $(\mu',\lambda')\in S_{1}(\mu,\lambda)$, then 
$(\lambda',\mu')\in S_{2}(\lambda,\mu)$.
We also have 
\begin{eqnarray*}
\mathrm{sign}_{1}(\lambda',\mu';\lambda,\mu)
=(-1)^{|\lambda\setminus\lambda'|}
\mathrm{sign}_{2}(\mu',\lambda';\mu,\lambda).
\end{eqnarray*}
We have four cases for the power of $2$ in Eqn. (\ref{eqn:skewSinQQ1}):
(1) $l(\lambda)=l(\mu)+1$ and $l(\nu)=l(\xi)+1$, 
(2) $l(\lambda)=l(\mu)+1$ and $l(\nu)=l(\xi)$,
(3) $l(\lambda)=l(\mu)$ and $l(\nu)=l(\xi)+1$,
and 
(4) $l(\lambda)=l(\mu)+1$ and $l(\nu)=l(\xi)$.
In all cases, we have $l(\xi')-l(\mu')=n-m$, which is the 
same power as Eqn. (\ref{eqn:skewSinQQ2}). 
By substituting Eqn. (\ref{eqn:skewSinQQ1}) for 
$\hat{S}[\mu\otimes\widehat{\lambda}_{i}/\xi\otimes\widehat{\nu}_{1}]$
into the equation above and successively applying Lemma \ref{lemma:QQ1},
we obtain Eqn. (\ref{eqn:skewSinQQ2}).

If we assume Eqn. (\ref{eqn:skewSinQQ2}) for $\hat{S}[\alpha'/\beta]$ 
with $|\alpha'|<|\alpha|$, 
we obtain Eqn. (\ref{eqn:skewSinQQ1}) in the same way as above.
\end{proof}

\begin{cor}[J\'ozefiak and Pragacz \cite{JozPra91}]
Let $\alpha$ and $\beta$ be shift-symmetric, {\it i.e.},
$\alpha:=\lambda\otimes\lambda$ and $\beta:=\mu\otimes\mu$.
Then, $\hat{S}_{\alpha/\beta}$ is given by the square of $Q_{\lambda/\mu}$:
\begin{eqnarray*}
S_{\alpha/\beta}=2^{l(\mu)-l(\lambda)}Q_{\lambda/\mu}^{2}.
\end{eqnarray*}
\end{cor}

Let $\alpha:=\lambda\otimes\mu$ and $\beta:=\nu\otimes\xi$ with 
$\beta\subseteq\alpha$.
We define $n:=l(\lambda)+1$ if $l(\lambda)=l(\mu)$ and $l(\nu)=l(\mu)$, 
and $n:=l(\lambda)$ otherwise.
We also define $m:=l(\nu)$. 
We regard $\mu$ (resp. $\xi$) as a partition  
of length $n$ (resp. $m$) by adding a zero to $\mu$ (resp. $\xi$)
if $l(\mu)=l(\lambda)-1$ (resp. $l(\xi)=l(\nu)-1$).
If $l(\lambda)=l(\mu)$ and $l(\nu)=l(\xi)$, we regard $\lambda$ 
and $\mu$ as a partition of length $n$ by adding a zero to $\lambda$ 
and $\mu$.
Note that in the last case we do not add zeros to $\nu$ and $\xi$.

We place the integers in $\lambda\cup\mu$ in a decreasing order 
from left to right and successively the integers in $\nu\cup\xi$
in a decreasing order from left to right.
When an integer $p$ is in $\lambda\cap\mu$ or $\nu\cap\xi$,
we place two $p$'s next to each other. 
We denote by $p_{\lambda}$ or $p_{\nu}$ (resp. $p_{\mu}$ or $p_{\xi}$) 
the left (resp. right) $p$.
We construct a perfect matching of length $2(n+m)$ by connecting 
two integers via an arc or a dashed arc.
A perfect matching satisfies:
\begin{enumerate}
\item We connect two integers $p$ and $q$ via an arc if and only if 
$p\in\lambda$ and $q\in\mu$.
\item We connect two integers $p$ and $q$ via a dashed arc 
if and only if $p\in\lambda$ and $q\in\nu$, or $p\in\mu$ and $q\in\xi$.
\item We do not connect the same integers $p_{\lambda}$ and $p_{\mu}$ 
for $p\ge1$. 
\item We do not connect two integers $p$ and $q$ if 
both $p$ and $q$ are in $\nu$ or $\xi$.
\end{enumerate}
Let $\mathrm{PM}(\alpha/\beta)$ be 
the set of perfect matchings satisfying conditions above.
Let $\pi\in\mathrm{PM}(\alpha/\beta)$. 
Then, we denote by $\mathrm{cr}(\lambda,\mu,\nu,\xi;\pi)$ 
the number of crossings in $\pi$ by ignoring whether an arc 
is dashed or not. 
Let $s:=s_{1}\ldots s_{2(m+n)}$ be the sequence of integers in 
$\lambda\cup\mu$ and $\nu\cup\xi$ with repeated entries.
We define a sign $\mathrm{sign}'(s_{i})$ by $\mathrm{sign}'(s_{i})=(-1)^{i-1}$
if $s_{i}\in\lambda\cup\mu$ and by $\mathrm{sign}'(s_{i})=(-1)^{i}$
if $s_{i}\in\nu\cup\xi$.
Then, we define a sign of a perfect matching as 
\begin{eqnarray*}
\mathrm{sign}_{6}(\lambda,\mu,\nu,\xi;\pi):=
(-1)^{\mathrm{cr}(\lambda,\mu,\nu,\xi;\pi)}
\prod_{s_{i}\in\lambda}\mathrm{sign}'(s_{i})\cdot 
\prod_{s_{i}\in\nu}\mathrm{sign}'(s_{i}).
\end{eqnarray*}

Given a perfect matching $\pi$,
we denote by $\mathrm{Arc}(\pi)$ (resp. $\mathrm{dArc}(\pi)$) 
the set of arcs (resp. dashed arcs).
When $s_{i}$ and $s_{j}$, $i<j$, is connected via an arc (resp. a dashed arc), 
we denote $(i,j)\in\mathrm{Arc}(\pi)$ (resp. $(i,j)\in\mathrm{dArc}(\pi)$).
Note that $s_{i}\ge s_{j}$ if $i<j$.

\begin{theorem}
\label{thrm:skewSinQQ2}
Denote $Q[\lambda]:=Q_{\lambda}$. Then, 
\begin{eqnarray*}
\hat{S}_{\alpha/\beta}:=2^{\mathrm{deg}(\alpha,\beta)}
\sum_{\pi\in\mathrm{PM}(\alpha/\beta)}
\mathrm{sign}_{6}(\lambda,\mu,\nu,\xi;\pi)
\prod_{(i,j)\in\mathrm{Arc}(\pi)}Q[(s_{i},s_{j})]
\prod_{(i,j)\in\mathrm{dArc}(\pi)}Q[(s_{i}-s_{j})],
\end{eqnarray*}
where 
$\mathrm{deg}(\alpha,\beta):=l(\nu)-l(\mu)$ if $l(\lambda)=l(\mu)$ 
and $\mathrm{deg}(\alpha,\beta):=l(\xi)-l(\mu)$ if $l(\lambda)=l(\mu)+1$.
\end{theorem}

\begin{example}
Let $\alpha=(5,4,3,2,1)$ and $\beta=(2,1)$. We have 
$\alpha=(5,3,1)\otimes(4,2)$ and $\beta=(2)\otimes(1)$.
The signs are given by 
$\prod_{s_{i}\in\lambda}\mathrm{sign}'(s_{i})=+1$ and 
$\prod_{s_{i}\in\nu}\mathrm{sign}'(s_{i})=-1$. 
Then, we have
\begin{eqnarray}
\label{eqn:skewSinQQQ}
\begin{aligned}
\hat{S}_{\alpha/\beta}
=&2^{-1}(
Q_{(3)}^3 Q_{(2, 1)} - Q_{(1)}Q_{(3)}Q_{(5)}Q_{(2, 1)} 
+Q_{(1)}Q_{(3)}^2 Q_{(3, 2)} \\
&-Q_{(1)}Q_{(3)}^2 Q_{(4, 1)}+Q_{(1)}^2 Q_{(5)}Q_{(4, 1)}
+Q_{(1)}^2 Q_{(3)}Q_{(4, 3)} \\ 
&-Q_{(1)}^2 Q_{(3)}Q_{(5, 2)} +Q_{(1)}^3 Q_{(5, 4)}).
\end{aligned}
\end{eqnarray}
In Figure \ref{fig:skewSinQQ}, the left (resp. right) picture correspond to 
the fourth (resp. seventh) term in Eqn. (\ref{eqn:skewSinQQQ}). 
\begin{figure}[ht]
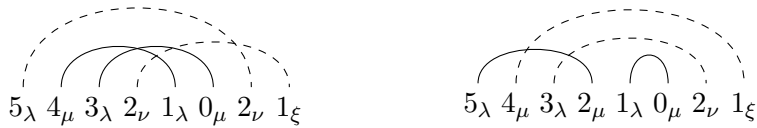

\begin{eqnarray*}
\tikzpic{-0.5}{
\draw[dashed](0,0)node[anchor=north]{$5_{\lambda}$}
  ..controls(0,1.4)and(3.0,1.4)..(3.0,0)node[anchor=north]{$2_{\nu}$};
\draw[dashed](1.5,0)node[anchor=north]{$2_{\nu}$}
  ..controls(1.5,0.8)and(3.5,0.8)..(3.5,0)node[anchor=north]{$1_{\xi}$};
\draw(0.5,0)node[anchor=north]{$4_{\mu}$}
  ..controls(0.5,0.72)and(2,0.72)..(2,0)node[anchor=north]{$1_{\lambda}$};
\draw(1,0)node[anchor=north]{$3_{\lambda}$}
  ..controls(1,0.72)and(2.5,0.72)..(2.5,0)node[anchor=north]{$0_{\mu}$};
}
\qquad\qquad
\tikzpic{-0.5}{
\draw(0,0)node[anchor=north]{$5_{\lambda}$}
   ..controls(0,0.6)and(1.5,0.6)..(1.5,0)node[anchor=north]{$2_{\mu}$};
\draw[dashed](0.5,0)node[anchor=north]{$4_{\mu}$}
   ..controls(0.5,1.3)and(3.5,1.3)..(3.5,0)node[anchor=north]{$1_{\xi}$};
\draw[dashed](1,0)node[anchor=north]{$3_{\lambda}$}
   ..controls(1,0.8)and(3,0.8)..(3,0)node[anchor=north]{$2_{\nu}$};
\draw(2,0)node[anchor=north]{$1_{\lambda}$}
   ..controls(2,0.5)and(2.5,0.5)..(2.5,0)node[anchor=north]{$0_{\mu}$};
}
\end{eqnarray*}
\caption{Examples of perfect matchings.}
\label{fig:skewSinQQ}
\end{figure}
\end{example}

\begin{proof}[Proof of Theorem \ref{thrm:skewSinQQ2}]
From Theorem \ref{thrm:skewSinQQ}, a $\hat{S}$-function can be 
expressed as a sum of products of two (skew) $Q$-functions.
We expand $Q$-functions in terms of skew $Q$-functions of length 
one and two.
By a direct calculation, it is easy to show that
the coefficient of $Q[(\lambda_{i},\lambda_{j})]$ with 
$\lambda_{i},\lambda_{j}\neq\lambda\cap\mu$ is zero.
Similarly, the coefficient of $Q[(\lambda_{i}-\xi_{j})]$
with $\lambda_{i}\neq\lambda\cap\mu$ and 
$\xi_{j}\neq\nu\cap\xi$ is zero.
Thus, we have
\begin{eqnarray*}
\hat{S}_{\alpha/\beta}:=2^{\mathrm{deg'}(\alpha,\beta)}
\sum_{\pi\in\mathrm{PM}(\alpha/\beta)}
\mathrm{sign}'_{6}(\lambda,\mu,\nu,\xi;\pi)
\prod_{(i,j)\in\mathrm{Arc}(\pi)}Q[(s_{i},s_{j})]
\prod_{(i,j)\in\mathrm{dArc}(\pi)}Q[(s_{i}-s_{j})],
\end{eqnarray*}
with some $\mathrm{deg}'(\alpha,\beta)$ and 
$\mathrm{sign}'_{6}(\lambda,\mu,\nu,\xi;\pi)$.
We first show that $\mathrm{deg}'(\alpha,\beta)=\mathrm{deg}(\alpha,\beta)$.  
When $l(\lambda)=l(\mu)$, we have a single term corresponding to a 
perfect matching $\pi$.
Thus, from Theorem \ref{thrm:skewSinQQ}, we have 
$\mathrm{deg}(\alpha,\beta)=l(\nu)-l(\lambda)=l(\nu)-l(\mu)$.
By the same reason, we have 
$\mathrm{deg}(\alpha,\beta)=l(\nu)-l(\lambda)=l(\xi)-l(\mu)$
if $l(\lambda)=l(\mu)+1$ and $l(\nu)=l(\xi)+1$.
However, in case of $l(\lambda)=l(\mu)+1$ and $l(\nu)=l(\xi)$,
we have two terms coming from the right hand side of 
Eqn. (\ref{eqn:skewSinQQ1}), which corresponds to a perfect matching 
$\pi$.
To see this, suppose that the right hand side of Eqn. (\ref{eqn:skewSinQQ1}) 
contains a term $Q[\lambda'/\nu']Q[\mu'/\xi']$ with 
$(\lambda',\mu',\nu',\xi')\in S(\lambda,\mu,\nu,\xi)$. 
Since $l(\lambda)=l(\mu)+1$ and $l(\nu)=l(\xi)$, we have 
$l(\lambda)-l(\nu)=l(\mu)-l(\xi)+1$. 
Assume that $\lambda'$ contains $\lambda_{i}\notin\lambda\cap\mu$. 
We consider the term containing $Q[\lambda_{i}]$.
Let $\widehat{\lambda'}:=\lambda'\setminus\{\lambda_i\}$ and 
$\widehat{\mu'}:=\mu'\cup\{\lambda_{i}\}$ as sets.
Then, an expansion of the term 
$Q[\widehat{\lambda'}/\nu']Q[\widehat{\mu'}/\xi']$ 
also contains $Q[\lambda_{i}]$. 
By taking the sign into account, we can show that 
these two terms contains the same sign. 
This implies that we have a factor $2$ in the expansion.
Thus, the power of $2$ is $l(\nu)-l(\lambda)+1=l(\xi)-l(\mu)$.

We will show that 
$\mathrm{sign}'_{6}(\lambda,\mu,\nu,\xi;\pi)
=\mathrm{sign}_{6}(\lambda,\mu,\nu,\xi;\pi)$.
Recall that we expand a $Q$-function in terms of $Q$-functions 
of length one and two.
It is easy to see that the number of crossings in a perfect matching 
is compatible with the sign arising from the expansion of a $Q$-function.
Thus, we have $\mathrm{sign}'_{6}(\lambda,\mu,\nu,\xi;\pi)$ is equal to
$\mathrm{sign}_{6}(\lambda,\mu,\nu,\xi;\pi)$ up to the overall factor.
By comparing the signs 
$\prod_{s_{i}\in\lambda\cup\nu}\mathrm{sign}'(s_{i})$ 
with the number of crossings, we conclude that the overall factor is one.
This completes the proof.
\end{proof}

Suppose that two ordinary partitions $\alpha$ and $\beta\subseteq\alpha$ are written as 
$\alpha:=\lambda'\otimes\mu'$ and $\beta:=\nu'\otimes\xi'$.
When $l(\lambda')=l(\mu')+1$, we define strict partitions of length $l(\lambda')$ by 
$\lambda=\lambda'$ and $\mu_{i}=\mu'_{i}$ for $1\le i\le l(\mu')$ and $\mu_{l(\mu')+1}=0$.
When $l(\lambda')=l(\mu')$, we define strict partitions of length $l(\lambda')+1$ by 
$\lambda_{i}=\lambda'_{i}$ for $1\le i\le l(\lambda')$ and $\lambda_{l(\lambda')+1}=0$ and 
$\mu_{i}=\mu'_{i}$ for $1\le i\le l(\mu')$ and $\mu_{l(\lambda')+1}=0$.
We also define $\nu$ and $\xi$ from $\nu'$ and $\xi'$ as follows.
If $l(\nu')=l(\xi')$, we define $\nu'=\nu$ and $\xi'=\xi$.
If $l(\nu')=l(\xi')+1$, we define $\nu:=\nu'$ and define $\xi$ 
by adding a zero to $\xi'$.

We define a matrix $\tilde{Q}(\lambda,\mu;\nu,\xi):=(\tilde{Q}_{i,j})_{1\le i,j\le l(\lambda)+l(\nu)}$
in terms of Schur $Q$-functions as follows.
We define 
\begin{eqnarray*}
\tilde{Q}_{i,j}:=
\begin{cases}
Q_{(\lambda_i,\mu_j)}, & \lambda_i>\mu_j>0, \\
-Q_{(\mu_j,\lambda_i)}, & 0<\lambda_i<\mu_j, \\
Q_{(\lambda_i)}, & \lambda_i>\mu_j=0, \\
-Q_{(\mu_j)}, & 0=\lambda_i<\mu_j, \\
0, & \lambda_i=\mu_j\neq0, \\
1 & \lambda_i=\mu_j=0,
\end{cases}
\end{eqnarray*}
for $1\le i,j\le l(\lambda)$, 
\begin{eqnarray*}
\tilde{Q}_{i,j+l(\lambda)}:=
\begin{cases}
Q_{(\lambda_i-\nu_j)}, & \lambda_i>\nu_j, \\
1, & \lambda_i=\nu_j, \\
0, & \lambda_i<\nu_j,
\end{cases}
\end{eqnarray*}
for $1\le i\le l(\lambda)$ and $1\le j\le l(\nu)$, 
\begin{eqnarray*}
\tilde{Q}_{i+l(\lambda),j}:=
\begin{cases}
-Q_{(\mu_j-\xi_i)}, & \mu_j>\xi_i, \\
-1, & \mu_j=\xi_i, \\
0, & \mu_j<\xi_i,
\end{cases}
\end{eqnarray*}
for $1\le i\le l(\nu)$ and $1\le j\le l(\lambda)$, and 
$\tilde{Q}_{i+l(\lambda),j+l(\lambda)}:=0$ for $1\le i,j\le l(\nu)$.

\begin{theorem}
\label{thrm:skewSindet}
In the above notation, a function $\hat{S}_{\alpha/\beta}$ can be expressed 
in terms of a determinant: 
\begin{eqnarray*}
\hat{S}_{\alpha/\beta}=2^{l(\xi')-l(\mu')}\det\left[\tilde{Q}(\lambda,\mu;\nu,\xi)\right].
\end{eqnarray*}
\end{theorem}
\begin{proof}
From Theorem \ref{thrm:skewSinQQ2}, we have an expression of $\hat{S}$-function 
in terms of perfect matchings.
Comparing the definition of the determinant $\det[\tilde{Q}]$ with 
a term in Theorem \ref{thrm:skewSinQQ2},
we have that $\hat{S}_{\alpha/\beta}$ is proportional to the determinant.
We have three cases: (1) $l(\lambda')=l(\mu')+1$, 
(2) $l(\lambda')=l(\mu')$ and $l(\nu')=l(\xi')$, 
and (3) $l(\lambda')=l(\mu')$ and $l(\nu')=l(\xi')+1$.
In case of (1) and (2), the overall factor is $2^{l(\xi')-l(\mu')}$.
In case of (3), note that the last elements in $\lambda,\mu$ and 
$\xi$ are zero. 
In Theorem \ref{thrm:skewSinQQ2}, we have no double zeros in the 
sequence $s$.
However, to connect the determinant to a perfect matching $\pi$ considered in
Theorem \ref{thrm:skewSinQQ2}, we introduce a perfect matching with 
double zeros.
We denote such perfect matching by $\pi'$.
We show that we have a one-to-two bijection between $\pi$ and $\pi'$.
Suppose that $p_{\mu}$ is connected to $0_{\xi}$ by a dashed arc.
Then, in $\pi'$, we connect $p_{\mu}$ and $0_{\lambda}$ by an arc 
and $0_{\mu}$ and $0_{\xi}$ by a dashed arc, or 
connect $p_{\mu}$ and $0_{\xi}$ by a dashed arc and $0_{\lambda}$ and 
$0_{\mu}$ by an arc.
Note that two $\pi'$'s above give the same products of $Q$-functions
and the sign is the same.
Thus, the power of two becomes $l(\nu')-l(\lambda')-1=l(\xi')-l(\mu')$.
This completes the proof.
\end{proof}

\begin{example}
Let $\alpha=(4,3,2,1)$ and $\beta=(1)$. 
Then, $\alpha=(4,2)\otimes(3,1)$ and $\beta=(1)\otimes\emptyset$.
The matrix $\tilde{Q}:=\tilde{Q}(\lambda,\mu;\nu,\xi)$ is given by 
\begin{eqnarray*}
\tilde{Q}=
\begin{bmatrix}
Q_{(4,3)} & Q_{(4,1)} & Q_{(4)} & Q_{(3)} \\
-Q_{(3,2)} & Q_{(2,1)} & Q_{(2)} & Q_{(1)} \\
-Q_{(3)} & -Q_{(1)} & 1 & 0 \\
-Q_{(3)} & -Q_{(1)} & -1 & 0
\end{bmatrix}.
\end{eqnarray*}
Therefore, we have $\hat{S}_{\alpha/\beta}=2^{-2}\det[\tilde{Q}]$.
\end{example}

A given $\alpha=\lambda\otimes\mu$, we define $A:=(a_1\ge a_2\ge\ldots\ge a_{2n})$, 
$\mathrm{sign}(A)$ and $d(A)$ as in Section \ref{Giambelli}.
For $\beta=\nu\otimes\xi$, we define $B$, $\mathrm{sign}(B):=(b_1\ge b_2\ge\ldots\ge b_{2m})$ 
and $d(B)$ similarly.
We define a skew-symmetric matrix $Q(A,B):=(Q_{i,j})_{1\le i,j\le 2(n+m)}$ as follows:
\begin{eqnarray*}
Q_{i,j}:=
\begin{cases}
0, & \mathrm{sign}(a_i)=\mathrm{sign}(a_j) \text { or } a_i=a_j>0, \\
Q_{(a_i,a_j)}, & a_i>a_j>0, \\
Q_{(a_i)}, & a_i>a_j=0, \\
1, & a_i=a_j=0,
\end{cases}
\end{eqnarray*}
for $1\le i\le j\le 2n$. 
For $1\le i\le 2n$ and $2n+1\le j\le2(n+m)$, we define 
\begin{eqnarray*}
Q_{i,j}:=
\begin{cases}
0, & \mathrm{sign}(a_i)\neq\mathrm{sign}(b_{j-2n}), \\
Q_{(a_i-b_{j-2n})}, & otherwise.
\end{cases}
\end{eqnarray*}
We define $Q_{i,j}:=0$ for $2n+1\le i\le j\le2(n+m)$.

\begin{theorem}
Suppose that the skew shape $\alpha/\beta$ exits for 
$\alpha=\lambda\otimes\mu$ and $\beta=\nu\otimes\xi$.
Then, we have 
\begin{eqnarray*}
\hat{S}_{\alpha/\beta}=
(-1)^{d(A)+d(B)+m}2^{l(\xi)-l(\mu)}
\mathrm{pf}\left[Q(A,B)\right]_{1\le i,j\le2(n+m)}
\end{eqnarray*}
\end{theorem}
\begin{proof}
By comparing  the Pfaffian $\mathrm{pf}[Q(A,B)]$ with Theorem 
\ref{thrm:skewSindet}, 
we can show that $\hat{S}_{\alpha/\beta}$ is proportional to 
the Pfaffian times $2^{l(\xi)-l(\mu)}$.
By calculating the overall sign, we obtain Theorem.
\end{proof}

\begin{example}
Let $\alpha:=(4,3,2,1)=(4,2)\otimes(3,1)$ and $\beta:=(1)=(1)\otimes\emptyset$.
Then, we have $A=(4,3,2,1,0,0)$ and $B=(1,0)$.
The skew-symmetric matrix $Q(A,B)$ is given by 
\begin{eqnarray*}
Q(A,B)=
\begin{bmatrix}
0 & Q[(4, 3)] & 0 & Q[(4, 1)] & 0 & Q[(4)] & Q[(3)] & 0 \\
-Q[(4, 3)] & 0 & Q[(3, 2)] & 0 & Q[(3)] & 0 & 0 & Q[(3)] \\
0 & -Q[(3, 2)] & 0 & Q[(2, 1)] & 0 & Q[(2)] & Q[(1)] & 0 \\
-Q[(4, 1)] & 0 & -Q[(2, 1)] & 0 & Q[(1)] & 0 & 0 & Q[(1)] \\
0 & -Q[(3)] & 0 & -Q[(1)] & 0 & 1 & 0 & 0  \\
-Q[(4)] & 0 & -Q[(2)] & 0 & -1 & 0 & 0 & 1 \\
-Q[(3)] & 0 & -Q[(1)] & 0 & 0 & 0 & 0 & 0  \\
0 & -Q[(3)] & 0 & -Q[(1)] & 0 & -1 & 0 & 0
\end{bmatrix}.
\end{eqnarray*}
Then, we have $\hat{S}_{\alpha/\beta}=-2^{-2}\cdot \mathrm{pf}[Q(A,B)]$.

\end{example}

\bibliographystyle{amsplainhyper} 
\bibliography{biblio}

\end{document}